\documentclass[11pt]{amsart}

\usepackage{graphicx} % Required for inserting images
\usepackage[margin=1in]{geometry}
\usepackage{amsmath,amssymb,amsfonts,amsthm,bm,mathtools,stmaryrd}
\usepackage{subcaption}

\usepackage{tikz}
\graphicspath{{./figures/tikz}}
\usepackage[foot]{amsaddr}
\usepackage[T1]{fontenc}
\usepackage{esint}
\usepackage{lineno}
\usepackage{bbold}
\usepackage{makecell}
\usepackage[numbers, sort]{natbib}
\usepackage{pifont}
\usepackage[dvipsnames]{xcolor}
\usepackage{arydshln}
\usepackage{comment}

\makeatother
\numberwithin{equation}{section}
\makeatother
\newtheorem{theorem}{Theorem}[section]
\newtheorem{lemma}{Lemma}[section]

\newtheorem{remark}{Remark}[section]

\newtheorem{proposition}{Proposition}[section]

\newcommand{\bu}{\mathbf{u}}
\newcommand{\bom}{\mathbf{m}}
\newcommand{\bx}{\mathbf{x}}

\newcommand{\markyes}{{\color{ForestGreen} \ding{51} }}
\newcommand{\markno}{{\color{red} \ding{55} }}
\newcommand{\bkt}[1]{\textcolor{blue}{[BKT\@: #1]}}
\newcommand{\wjb}[1]{\textcolor{red}{[WJB\@: #1]}}

\usepackage[colorlinks=true, linkcolor=blue, urlcolor=black, citecolor=red]{hyperref}
\usepackage{cleveref}
\usepackage{tikz}
\usepackage{pgfplots}
\usetikzlibrary{shapes.arrows, patterns, calc}
\usepackage{tikz-3dplot}
\usepackage{bbm}
\ifpdf
  \DeclareGraphicsExtensions{.eps,.pdf,.png,.jpg}
\else
  \DeclareGraphicsExtensions{.eps}
\fi

\usepgfplotslibrary{groupplots}
\usetikzlibrary{arrows}
\usetikzlibrary{shapes}
\usetikzlibrary{decorations.text}
\usetikzlibrary{quantikz}
\usepackage{siunitx}
\usetikzlibrary{arrows.meta}

\pgfplotsset{ compat=1.18,
    standard/.style={
    scale only axis,
    width=0.5\textwidth,
    enlarge x limits=0.05,
    enlarge y limits=0.05,
    max space between ticks=40,
    every axis/.append style={font=\normalsize},
	every legend/.append style={font=\normalsize},
	every node/.append style={font=\normalsize},	
	}
}

\definecolor{steelblue}{HTML}{A1BDC7}
\definecolor{orange}{HTML}{D98C21}
\definecolor{silver}{HTML}{B0ABA8}
\definecolor{rust}{HTML}{B8420F}
\definecolor{seagreen}{HTML}{2E6B69}
\definecolor{joshua}{HTML}{FBDC7F}
\definecolor{darksky}{HTML}{154c79}

\colorlet{lightsilver}{silver!30!white}
\colorlet{darkorange}{orange!85!black}
\colorlet{darksilver}{silver!85!black}
\colorlet{darksteelblue}{steelblue!85!black}
\colorlet{darkrust}{rust!85!black}
\colorlet{darkseagreen}{seagreen!85!black}

\title[Hamiltonian Information Geometric Regularization of Compressible Euler]{Hamiltonian Information Geometric Regularization\\of the Compressible Euler Equations}
\author{William Barham$^1$, Brian K. Tran$^2$, Ben S. Southworth$^1$, and Florian Schäfer$^3$}
\date{\today}
\address{$^1$Los Alamos National Laboratory, Theoretical Division, Los Alamos, New Mexico 87545}
\address{$^2$University of Colorado Boulder, Department of Applied Mathematics, Boulder, CO 80309}
\address{$^3$New York University, Courant Institute of Mathematical Sciences, New York, NY 10012}
\email{wbarham@lanl.gov,\,brian.tran@colorado.edu,\,southworth@lanl.gov,\,florian.schaefer@nyu.edu}

\begin{document}

\maketitle

\begin{comment}
\wjb{To do: 
\begin{itemize}
    \item Fix errors appearing in sections 5 and 6 due to incorrect Korteweg pressure.
    \item Fix appendix B.4 simplifying the conservative entropic pressure elliptic equation.
    \item Correct derivation of dissipative subsystem leave total energy equation unchanged from thermodynamic HRE case. 
    \item Add section 7 on one-dimensional numerical tests. \bkt{A suggestion that we can discuss on Monday during meeting, it might be better to present all of the models' names (+formulas if not too lengthy) and the numerical results + discussion of results, after the motivation section 1.2. This brings results of the long derivations to the front and then the rest of the paper is dedicated to the derivations etc. Arguably, most readers would care more about this before deciding to sit through derivations and so might be more readable/impactful this way.}
    \item Change tone and purpose of paper described in abstract, introduction, and conclusion to better reflect the new findings.
\end{itemize}
}
\end{comment}

%================================================================
\section*{Abstract}
The recently proposed information geometric regularization (IGR) was the first inviscid regularization of the multi-dimensional compressible Euler equations, which enabled the simulation of realistic compressible fluid models at an unprecedented scale. However, the thermodynamic effects of this regularization have not yet been understood in a principled manner. To achieve a proper understanding of the thermodynamic aspects of the IGR, we decompose the regularization into its conservative dynamics, framed as a Hamiltonian subsystem, and its dissipative dynamics. In so doing, we further introduce two more models to compare to IGR, the Hamiltonian regularized Euler (HRE) model, which is the first multi-dimensional, non-dispersive Hamiltonian regularization of the compressible Euler equations with energy, as well as the Hamiltonian IGR (HIGR) model, which modifies the dissipation used in the IGR model to instead utilize a metriplectic dissipative force. Despite having many attractive features, the HRE and HIGR models exhibit notable defects in numerical tests on colliding shock problems, which preclude their use as computational tools without further study of dissipative weak solutions to these models. Additionally, our analysis presents new results on the IGR model itself, including its ability to conserve acoustic waves, as well as local energy transport laws and entropy production rates. By separating the conservative and dissipative dynamics of the IGR, our hope is that subsequent of analysis of the IGR model can benefit from this natural decomposition, such as, for example, rigorous proofs of strong solutions for multi-dimensional IGR for the compressible Euler system with thermodynamics. 

{\hypersetup{linkcolor=black}\tableofcontents}

%================================================================
\section{Introduction}

%---------------------------------------------------------------
\subsection{Background}\label{sec:intro}

The compressible Euler equations and variations thereof are fundamental to numerous complex and important physical phenomena, including atmospheric science and numerical weather prediction, aerospace engineering, and high-energy density physics, particularly inertial confinement fusion. They are also the core partial differential equations surrounding shock physics \cite{courant1948supersonic}, and further give rise to Rayleigh-Taylor, Richtmyer-Meshkov, and Kelvin-Helmholtz instabilities. These aspects in particular can make their numerical simulation exceedingly difficult \cite{toro2009riemann}. Problems arise largely from not being able to resolve the spatio-temporal length scales associated with instabilities and/or nonlinearly interacting shocks in the discrete setting. Meanwhile, it is also critical to maintain exactly or to high accuracy physical invariants such as conservation of mass, energy, and momentum to avoid inducing numerical instabilities, as well as physical properties such as positivity \cite{guermond2018second, guermond2019invariant, clayton2022invariant}. Negative energy values, for example, are not only unphysical, but can immediately break a code/simulation, e.g., when evaluating an equation of state model or coupling to other physics. 

To address these challenges, there is a broad body of literature regarding theory and numerical techniques for approximating solutions of the compressible Euler equations. A common stabilization approach is the use of artificial viscosity to stabilize the solution near discontinuities or sharp gradients. Early work did this on the continuous level, adding some small amount of hyperviscosity to the PDE formulation \cite{vonneumann1950method}, with modern hyperviscosity methods using shock indicators to localize diffusion \cite{fiorina2007artificial, bhagatwala2009modified, kawai2010assessment, mani2009suitability, guermond2011entropy, barter2010shock, chan2025artificial, bruno2022fc}, but it can also be performed on the discrete level, such as upwinding a discontinuous Galerkin discretization for improved stabilization \cite{DG2000}. Stabilization can also be addressed through the use of specialized nonlinear flux limiters. There is also significant literature on shock capturing and reconstruction techniques, particularly in the context of finite volume methods \cite{LeVeque_2002}, such as the well-known TVD \cite{Shu.1988fyej, Gottlieb.1996} and ENO/WENO methods \cite{Shu.1988}. Aside from the introduction of continuous hyperviscosity, most of these methods are discrete in nature, and do not modify the strong or weak PDE formulation/model. More recently, there has been increasing interest and study in the geometry and structure of the underlying PDE, and modified problem formulations and/or associated numerical methods that identify and/or preserve certain structure, such as a Hamiltonian or bracket formulations, e.g. \cite{VaMa2005, pavelka2016, khesin_2021, TrBuSo2025, Hirvijoki_2022, Burby_2015, PhysRevE.109.045202, bressan2025metriplectic}. Although interesting theoretically, many of these methods are not currently viable or competitive to solve realistic problems at scale due to either computational expense or practical limitations. 

Recently, a new framework was proposed for regularizing the compressible Euler equations on the PDE level from the perspective of information geometry \cite{cao2023information, cao2024information}. On a high level, the information geometric regularization (IGR) formulation adds a regularization to the pressure tensor to smooth out shocks and steep gradients such that they can be captured by a discrete representation, while maintaining correct limiting properties as the regularization goes to zero. The method was originally proposed for barotropic fluids \cite{cao2023information}. 
An informal approach to extending it to more general equations was suggested but deferred to later work.
In the unidimensional pressureless case, the existence of smooth solutions of the regularized equation and their convergence in the limit of vanishing regularization strength has been shown \cite{cao2024information}.
The suitability of this approach for compressible Navier-Stokes systems with an energy equation was demonstrated at extreme scale for complex aerospace applications in \cite{wilfong2025simulatingmanyenginespacecraftexceeding}. Yet this extension is informal in the sense that the information geometric analysis was not modified to account for the additional conservative integrated variable. Together these works provide a theoretical basis from which to regularize the compressible Euler equations and maintain important physical structure, while also extending to practical application at scale, and are what we will build on in this paper. 
%\bkt{``In a companion paper, we will consider numerical methods and applications of the Hamiltonian IGR model introduced here" - should we add something like this?}

%\wjb{Add paragraph here with high-level summary of paper results and consequences. Move some of section 1.5 here.}

The introduction is a self-contained account of the core results of the paper, while the body of the paper provides a account of the details leading to these results. In~\Cref{subsec:motivation}, we provide a brief sketch of the thermodynamics implied by a prior IGR model \cite{wilfong2025simulatingmanyenginespacecraftexceeding} to motivate the need for a detailed investigation of the coupling between kinematics and thermodynamics in IGR. In~\Cref{subsec:HIGR}, we present the Hamiltonian IGR (HIGR) system, which represents the core contribution of this work. In~\Cref{subsec:1d_numerics}, we compare the behavior of HIGR and IGR in a simple one-dimensional numerical experiment. Finally, ~\Cref{subsec:outline} provides an overview of the remainder of the article. In brief, we find that there is a natural way to connect IGR with the Hamiltonian formalism, yielding the conservative part of the HIGR system. This facilitates a principled treatment of the coupling between kinetic and internal energy. However, dissipation plays an essential role in ensuring stability in numerical applications of the regularization and in ensuring convergence to a physical solution as the regularization parameter goes to zero. While the structure of the conservative part of the regularization seems to follow as an inevitable consequence of the modeling assumptions, see Sections~\ref{sec:pressureless-IGR} and \ref{sec:dispersion-free-extension}, the theory motivating the dissipative part of the regularization is incomplete and requires further study. This gap in the modeling formalism will be addressed in a future work. 

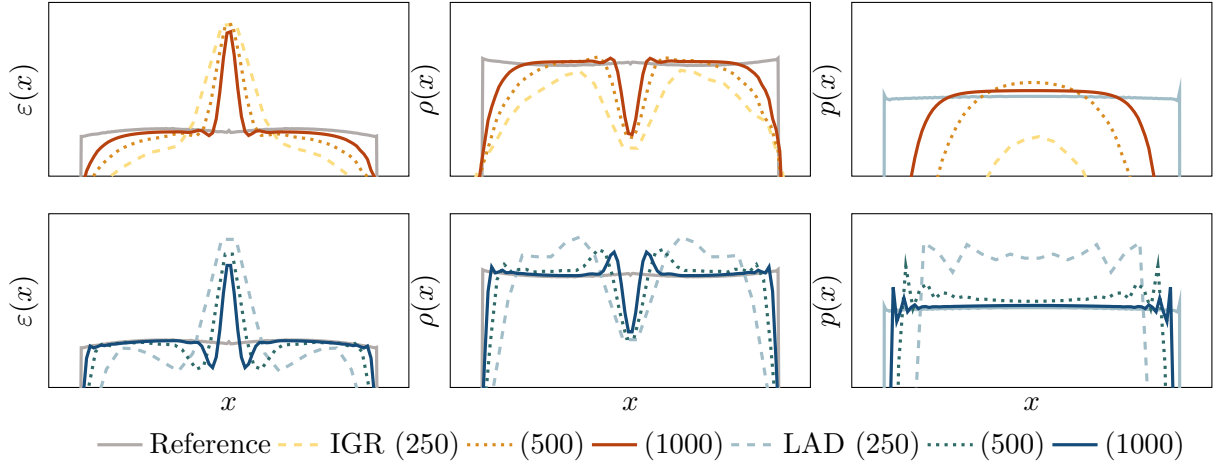
\begin{figure}
    \centering
    \begin{tikzpicture}
 \begin{groupplot}[
 	compat=1.3,
 	group style={group size=3 by 2,
 	horizontal sep=0.55cm,
     vertical sep=0.5cm,},
 	]
 	
 	\nextgroupplot[
 		standard,
 		ylabel={$\varepsilon(x)$}, %changed from e
 % 		xmode=log,
 % 		ymode=log,
 		height=0.14\textwidth,
 		width=0.29\textwidth,
         ymin=7.0,
         ymax=9.0,
         xmin=0.45,
         xmax=0.55,
         yticklabels={{$0$}},
 		xmajorticks=false,
 		xtick={0.87, 0.9, 0.93},
 		ymajorticks=false,
 		enlarge x limits=0.,
 		enlarge y limits=0.,
 		legend style={
 			at={(1.65,-1.70)},inner sep=3pt,anchor=south,legend columns=7,legend cell align={left}, draw=none,fill=none},
 		]	
 		\addlegendimage{silver,very thick}
 		\addlegendimage{joshua, dashed, very thick}
 		\addlegendimage{orange, dotted, very thick}
 		\addlegendimage{rust,very thick}
 		\addlegendimage{steelblue, dashed, very thick}
 		\addlegendimage{seagreen, dotted, very thick}
 		\addlegendimage{darksky,very thick}
 		\legend{Reference, IGR (250), (500), (1000), LAD (250), (500), (1000)}

         \addplot[silver, very thick,each nth point={10}, filter discard warning=false, unbounded coords=discard] table[x index={0},y index={3}, col sep=space] {figures/data/energy_spike/ref.csv};
         \addplot[joshua, very thick, dashed] table[x index={0},y index={3}, col sep=space] {figures/data/energy_spike/llf_igr_low.csv};
         \addplot[orange, very thick, dotted] table[x index={0},y index={3}, col sep=space] {figures/data/energy_spike/llf_igr_mid.csv};
         \addplot[rust, very thick] table[x index={0},y index={3}, col sep=space] {figures/data/energy_spike/llf_igr_high.csv};

 	\nextgroupplot[
 		standard,
 		ylabel={$\rho(x)$},
 % 		xmode=log,
 % 		ymode=log,
 		height=0.14\textwidth,
 		width=0.29\textwidth,
         ymin=1.5,
         ymax=2.10,
         xmin=0.45,
         xmax=0.55,
         yticklabels={{$0$}},
 		xmajorticks=false,
 		xtick={0.87, 0.9, 0.93},
 		ymajorticks=false,
 		enlarge x limits=0.,
 		enlarge y limits=0.,
 		]	

        \addplot[silver, very thick,each nth point={10}, filter discard warning=false, unbounded coords=discard] table[x index={0},y index={2}, col sep=space] {figures/data/energy_spike/ref.csv};

        \addplot[joshua, very thick, dashed] table[x index={0},y index={2}, col sep=space] {figures/data/energy_spike/llf_igr_low.csv};
        \addplot[orange, very thick, dotted] table[x index={0},y index={2}, col sep=space] {figures/data/energy_spike/llf_igr_mid.csv};
        \addplot[rust, very thick] table[x index={0},y index={2}, col sep=space] {figures/data/energy_spike/llf_igr_high.csv};

 	\nextgroupplot[
 		standard,
 		ylabel={$p(x)$},
 % 		xmode=log,
 % 		ymode=log,
 		height=0.14\textwidth,
 		width=0.29\textwidth,
         ymin=13.75,
         ymax=14.75,
         xmin=0.45,
         xmax=0.55,
         yticklabels={{$0$}},
 		xmajorticks=false,
 		xtick={0.87, 0.9, 0.93},
 		ymajorticks=false,
 		enlarge x limits=0.,
 		enlarge y limits=0.,
 		]	
 
         \addplot[steelblue, very thick,each nth point={10}, filter discard warning=false, unbounded coords=discard] table[x index={0},y expr={\thisrowno{2}*\thisrowno{3}}, col sep=space] {figures/data/energy_spike/ref.csv};

		 \addplot[joshua, very thick, dashed] table[x index={0},y expr={\thisrowno{2}*\thisrowno{3}}, col sep=space] {figures/data/energy_spike/llf_igr_low.csv};
		 \addplot[orange, very thick, dotted] table[x index={0},y expr={\thisrowno{2}*\thisrowno{3}}, col sep=space] {figures/data/energy_spike/llf_igr_mid.csv};
		 \addplot[rust, very thick] table[x index={0},y expr={\thisrowno{2}*\thisrowno{3}}, col sep=space] {figures/data/energy_spike/llf_igr_high.csv};

 	\nextgroupplot[
 		standard,
         xlabel={$x$},
 		ylabel={$\varepsilon(x)$},
 % 		xmode=log,
 % 		ymode=log,
 		height=0.14\textwidth,
 		width=0.29\textwidth,
         ymin=7.0,
         ymax=9.0,
         xmin=0.45,
         xmax=0.55,
         yticklabels={{$0$}},
 		xmajorticks=false,
 		xtick={0.87, 0.9, 0.93},
 		ymajorticks=false,
 		enlarge x limits=0.,
 		enlarge y limits=0.,
 		]	
 
         \addplot[silver, very thick,each nth point={10}, filter discard warning=false, unbounded coords=discard] table[x index={0},y index={3}, col sep=space] {figures/data/energy_spike/ref.csv};
         \addplot[steelblue, very thick, dashed] table[x index={0},y index={3}, col sep=space] {figures/data/energy_spike/llf_lad_low.csv};
         \addplot[seagreen, very thick, dotted] table[x index={0},y index={3}, col sep=space] {figures/data/energy_spike/llf_lad_mid.csv};
         \addplot[darksky, very thick] table[x index={0},y index={3}, col sep=space] {figures/data/energy_spike/llf_lad_high.csv};

 	\nextgroupplot[
 		standard,
         xlabel={$x$},
 		ylabel={$\rho(x)$},
 % 		xmode=log,
 % 		ymode=log,
 		height=0.14\textwidth,
 		width=0.29\textwidth,
         ymin=1.5,
         ymax=2.10,
         xmin=0.45,
         xmax=0.55,
         yticklabels={{$0$}},
 		xmajorticks=false,
 		xtick={0.87, 0.9, 0.93},
 		ymajorticks=false,
 		enlarge x limits=0.,
 		enlarge y limits=0.,
 		]	
 
         \addplot[silver, very thick,each nth point={10}, filter discard warning=false, unbounded coords=discard] table[x index={0},y index={2}, col sep=space] {figures/data/energy_spike/ref.csv};

         \addplot[steelblue, very thick, dashed] table[x index={0},y index={2}, col sep=space] {figures/data/energy_spike/llf_lad_low.csv};
         \addplot[seagreen, very thick, dotted] table[x index={0},y index={2}, col sep=space] {figures/data/energy_spike/llf_lad_mid.csv};
         \addplot[darksky, very thick] table[x index={0},y index={2}, col sep=space] {figures/data/energy_spike/llf_lad_high.csv};

 	\nextgroupplot[
 		standard,
         xlabel={$x$},
 		ylabel={$p(x)$},
 % 		xmode=log,
 % 		ymode=log,
 		height=0.14\textwidth,
 		width=0.29\textwidth,
         ymin=13.75,
         ymax=14.75,
         xmin=0.45,
         xmax=0.55,
         yticklabels={{$0$}},
 		xmajorticks=false,
 		xtick={0.87, 0.9, 0.93},
 		ymajorticks=false,
 		enlarge x limits=0.,
 		enlarge y limits=0.,
 		]	
 
		 \addplot[steelblue, very thick,each nth point={10}, filter discard warning=false, unbounded coords=discard] table[x index={0},y expr={\thisrowno{2}*\thisrowno{3}}, col sep=space] {figures/data/energy_spike/ref.csv};

		 \addplot[steelblue, very thick, dashed] table[x index={0},y expr={\thisrowno{2}*\thisrowno{3}}, col sep=space] {figures/data/energy_spike/llf_lad_low.csv};
		 \addplot[seagreen, very thick, dotted] table[x index={0},y expr={\thisrowno{2}*\thisrowno{3}}, col sep=space] {figures/data/energy_spike/llf_lad_mid.csv};
		 \addplot[darksky, very thick] table[x index={0},y expr={\thisrowno{2}*\thisrowno{3}}, col sep=space] {figures/data/energy_spike/llf_lad_high.csv};

 \end{groupplot}
\end{tikzpicture}
    \caption{\textbf{Thermodynamic inconsistency.} Both IGR and local artificial diffusivity (LAD) are purely momentum-based Navier-Stokes-type regularizations. As such, they are prone to thermodynamic inconsistencies manifesting themselves as spurious spikes in the internal energy after the collision of two shocks, particularly at lower grid resolutions---note the difference in defect size at $250$ versus $1000$ grid points. Observe that the pressure (which features in the momentum equation) is not affected by this.}
    \label{fig:energy_spikes}
\end{figure}

\subsection{Motivation: entropy production in IGR}\label{subsec:motivation}

The information geometric regularization of the compressible Euler equations---the first inviscid, multi-dimensional regularization of compressible flow---was originally derived for barotropic fluids \cite{cao2023information}, and is written as follows:
\begin{equation} \label{eq:standard_igr}
\begin{cases}
    \partial_t \rho + \nabla_{\mathbf{x}} \cdot (\rho \bu) = 0 \,, & \\
    \partial_t (\rho \bu) + \nabla_{\mathbf{x}} \cdot ( \rho \bu \otimes \bu + (p + \Sigma) \mathbb{I} ) = 0 \,, & \\
    \rho^{-1} \Sigma - \alpha \nabla_{\mathbf{x}} \cdot (\rho^{-1} \nabla_{\mathbf{x}} \Sigma) = \alpha \left[ \mathrm{tr}((\nabla_{\mathbf{x}} \bu)^2) + \mathrm{tr}^2(\nabla_{\mathbf{x}} \bu) \right],
\end{cases}
\end{equation}
where $\rho$ is the fluid density, $\mathbf{u}$ is the fluid velocity, $p$ is the pressure, $\Sigma$ is the entropic pressure, and $\alpha > 0$ is the regularization parameter. Its derivation is based on the interpretation of Euler's equations as a geodesic flow on the space of diffeomorphisms \cite{arnold1966geometrie}, with an added logarithmic barrier---inspired by similar constructions used in interior point methods---which prevents the Lagrangian flow map from becoming singular. Whereas the classical interpretation of Euler's equations as a geodesic flow is based on the Levi-Civita connection coming from the kinetic energy metric, IGR uses the dually-flat connection associated with the Hessian manifold structure defined by the kinetic energy metric augmented with a logarithmic barrier. While the compressible Euler system generically evolves towards a state which invalidates its interpretation as a geodesic flow, namely when shocks form and the diffeomorphism between the Eulerian and Lagrangian reference frames ceases to exist, the IGR model is geodesically complete. In fact, in the unidimensional, pressureless case, it has been proven that this regularization yields global strong solutions \cite{cao2024information}. 
For extensions beyond the barotropic case, \cite{cao2023information} proposed to simply add the entropic pressure $\Sigma$ to the pressure, but did not study the behavior of this extension.

Based on this approach, \cite{wilfong2025simulatingmanyenginespacecraftexceeding} applied IGR to the compressible Navier--Stokes equation with energy equation. In this model, the standard IGR model in Equation \eqref{eq:standard_igr} is simply augmented with an energy equation,
\begin{equation} \label{eq:standard_igr_energy}
    \partial_t E + \nabla_{\mathbf{x}} \cdot ( (E + p + \Sigma) \bu ) = 0 \,, 
\end{equation}
where the pressure is determined through an equation of state, $p = p(e, \rho)$, and the internal energy density is
\begin{equation} \label{eq:total-internal-kinetic-relation}
    e = E - \frac{1}{2} \rho |\mathbf{u}|^2 \,.
\end{equation}
This model enabled the solution of realistic compressible fluid models at an unprecedented scale \cite{wilfong2025simulatingmanyenginespacecraftexceeding}. However the coupling of kinematics---described by the continuity/momentum subsystem---with thermodynamics---described by the energy subsystem---is not derived with the geometric motivation of the original IGR model. As such, a detailed understanding of the energy conservation law implied by the barotropic IGR model, its extension to general equations of state, and the thermodynamic consistency of this model remain unexplored. 

%\fs{I feel like here we are getting into the weeds a little too quickly. Maybe having a figure of the internal energy spikes is helpful?}
In Figure~\ref{fig:energy_spikes}, we observe a spurious spike in internal energy arising from the collision of two waves in a one-dimensional simulation; this defect, arising from an inconsistent treatment of thermodynamics, is observed both in simulations which regularize shock interfaces using IGR and those which use localized artificial diffusivity (LAD) \cite{kawai2010assessment, bhagatwala2009modified, dolejvsi2003some, bruno2022fc, guermond2011entropy, barter2010shock, mani2009suitability,fiorina2007artificial, cook2005hyperviscosity, puppo2004numerical, vonneumann1950method, chan2025artificial}. Deriving the entropy equation implied by IGR model demonstrates how certain flow configurations can exhibit undesirable, thermodynamically-inconsistent behavior. By construction, total energy, $E$, is conserved. However, the system is not conservative, in the sense that total entropy is not conserved. As written, the irreversible character of this model is not immediately obvious. We assume that there exists an equation of state for the specific internal energy, $\varepsilon = \varepsilon(\rho, s)$, where $s$ is the entropy, and the internal energy density is related to specific internal energy via $e = \rho \varepsilon$. The fundamental thermodynamic relation, which encodes the first law of thermodynamics, asserts that
\begin{equation}\label{eq:fundamental-thermo-relation}
    \mathsf{d} \varepsilon = \vartheta \mathsf{d} s + p \mathsf{d}(1/\rho) \,,
\end{equation}
where $\vartheta$ is the temperature, implying a functional relationship among the thermodynamic variables. The second law of thermodynamics states that $\vartheta \mathsf{d} s = \delta Q \geq 0$, where $\delta Q$ is a heat source in the system. This infinitesimal heating may take several forms and need not only depend on the thermodynamic state variables. 

A straightforward calculation using \eqref{eq:standard_igr_energy}, \eqref{eq:fundamental-thermo-relation}, and the chain rule for $\varepsilon = \rho^{-1} (E - \frac{1}{2} \rho |\mathbf{u}|^2)$ yields the entropy production rate in either the Lagrangian \eqref{eq:entropy-IGR-Lagrangian} or Eulerian \eqref{eq:entropy-IGR-Eulerian} frame,
\begin{subequations}
    \begin{align}
        \rho  \frac{\mathsf{D} s}{\mathsf{D} t} &= - \vartheta^{-1}\Sigma \nabla_{\mathbf{x}} \cdot \mathbf{u}, \label{eq:entropy-IGR-Lagrangian} \\
            \partial_t (\rho s) + \nabla_{\mathbf{x}} \cdot (\rho s \mathbf{u}) &= - \vartheta^{-1} \Sigma \nabla_{\mathbf{x}} \cdot \mathbf{u} \label{eq:entropy-IGR-Eulerian} \,.
    \end{align}
\end{subequations}
The non-conservative $-\vartheta^{-1} \Sigma \nabla_{\mathbf{x}} \cdot \mathbf{u}$ term arises from a heat source/sink. For positive $\Sigma$, this means that entropy increases when the flow compresses ($\nabla_{\mathbf{x}} \cdot \mathbf{u} < 0$) and decreases when it diverges. 
For sufficiently small $\alpha$, $\Sigma$ is negligible away from shocks. Thus, so long as $\Sigma$ is positive, it achieves the desired effect of increasing entropy appreciably only in shocks. 
Although the elliptic operator defining $\Sigma$ enjoys a maximum principle, see \Cref{appendix:maximum_principle}, it nevertheless fails to be strictly positive for general flows. Indeed we can decompose the strain tensor, see \Cref{appendix:rhs_decomp}, to find that
\begin{equation}
    \mathrm{tr}((\nabla_{\mathbf{x}} \mathbf{u})^2)
    =
    \|\mathbb{S}\|^2_F - \|\mathbb{\Omega}\|^2_F
\end{equation}
where $\| \cdot \|_F$ denotes the Frobenius norm, and $\mathbb{S}$ and $\mathbb{\Omega}$ are the symmetric strain tensor and the vorticity tensor, respectively. Hence, for flows with a strong rotational component, the sign of entropy production is not identically determined by the sign of the divergence of the velocity field. 

In the formulation of IGR with \Cref{eq:standard_igr} augmented by the total energy evolution \Cref{eq:standard_igr_energy}, the entirety of the entropic pressure, $\Sigma$, functions as a heat source/sink in its thermodynamic coupling. Moreover, the sign of entropy production is relatively difficult to predict for an arbitrary flow configuration. In this work, we show that it is possible to split the entropic pressure into two parts, with one part modeled as a reversible process using a Hamiltonian framework, and the remaining part as a heat source/sink using a sign-indefinite generalization of the metriplectic formalism, see \Cref{appendix:metriplectic}. In doing so, we weaken the coupling of the information geometric regularization to irreversible processes, potentially reducing the amount of dissipation implied by this model, and we find that $\dot{S} > 0$ when $\nabla_{\mathbf{x}} \cdot \mathbf{u} < 0$, and $\dot{S} < 0$ when $\nabla_{\mathbf{x}} \cdot \mathbf{u} > 0$. While the second law of thermodynamics is still violated in this new model, we can ensure positive entropy production where wave-steepening occurs. 

\subsection{The Hamiltonian Information Geometric Regularization}\label{subsec:HIGR}

This paper derives a model which we call Hamiltonian IGR (HIGR), which extends standard IGR such that its conservative part is Hamiltonian, and entropy production from its non-conservative part is sign definite with respect to the divergence of the velocity field. Like standard IGR, the model couples a nonlinear hyperbolic conservation law with an elliptic subsystem, and is written as follows:
\begin{equation}
\left\{
\begin{aligned}
    &\partial_t
    \begin{bmatrix}
        \rho \\
        \rho \bu \\
        E
    \end{bmatrix}
    +
    \nabla_{\bx} \cdot
    \begin{bmatrix}
        \rho \mathbf{u} \\
        \rho \mathbf{u} \otimes \mathbf{u} + \mathbb{\Sigma} \\
        E \mathbf{u} + \mathbb{\Sigma} \cdot \mathbf{u} + \alpha c_s^2 (\nabla_\bx \cdot \bu) \nabla_\bx \rho / \rho
    \end{bmatrix}
    =
    0 \,, \\
    &\begin{aligned}
    \rho^{-1} \mathbb{\Sigma} - \alpha \nabla_{\mathbf{x}} \cdot (\rho^{-1} \nabla_{\mathbf{x}} \cdot \mathbb{\Sigma} )\mathbb{I}
    &=
    \left( \rho^{-1} p - \alpha \nabla_{\mathbf{x}} \cdot (\rho^{-1} c_s^2 \nabla_\bx \rho) \right) \mathbb{I}  + \alpha (\mathrm{tr}^2(\nabla_\bx \bu) + \mathrm{tr}((\nabla_{\mathbf{x}} \mathbf{u})^2))\mathbb{I} \\
    &+ 
    \alpha
    \left[
    \left( \left( \frac{ \rho^2 \left. \partial_\rho^3 (\rho \varepsilon) \right|_s }{2} - \frac{ c_s^2}{2} \right) \left|\frac{\nabla_\bx \rho}{\rho} \right|^2 \right) \mathbb{I} + c_s^2 \frac{\nabla_\bx \rho}{\rho} \otimes \frac{\nabla_\bx \rho}{\rho}
    \right]
    \end{aligned} \,.
\end{aligned}
\right.
\end{equation}
where the pressure, $p$, sound speed, $c_s$, and the third derivative of internal energy with respect to density at fixed entropy $\left. \partial_\rho^3 (\rho \varepsilon) \right|_s$ are obtained through an equation of state. The entropy, also obtained through the equation of state, evolves as
\begin{equation}
    \partial_t (\rho s)
    +
    \nabla_\bx (\rho s \bu)
    =
    - \alpha \rho \vartheta^{-1} (\nabla_\bx \cdot \bu)^3 \,.
\end{equation}
The HIGR system formally recovers the compressible Euler equation as $\alpha \to 0$. At a glance, the purpose of this paper is to derive this model and to describe its properties. 

%================================================================
\subsection{One dimensional numerical experiments} \label{subsec:1d_numerics}

Before proceeding with the derivation of the HIGR model, we consider some simple numerical experiments comparing the IGR and HIGR models for a one-dimensional ideal gas. The HIGR model for a one-dimensional ideal gas may be written as follows:
\begin{equation}
\left\{
\begin{aligned}
    &\rho_t + (\rho u)_x = 0 \,, \\
    &(\rho u)_t + (\rho u^2 + (\gamma - 1) \rho \varepsilon + \Sigma)_x = 0 \,, \\
    &E_t + ( (E + (\gamma - 1) \rho \varepsilon + \Sigma)u + \alpha \gamma (\gamma - 1) \rho \varepsilon (\rho_x/\rho) u_x)_x = 0 \,, \\
    &\rho^{-1} \Sigma - \alpha (\rho^{-1} \Sigma_x)_x = 
    \alpha 
    \left[
    k u_x^2 
    +
    (\gamma - 1)
    \left(
    \varepsilon_x
    -
    (\gamma - 1) \varepsilon \rho_x/\rho
    \right)_x
    + 
    \frac{\gamma (\gamma - 1)^2 \varepsilon}{2} (\rho_x/\rho)^2
    \right]
    \,,
\end{aligned}
\right.
\end{equation}
with $k = 2$, and where $\varepsilon = \rho^{-1}(E - (1/2) \rho ( u^2 + \alpha (\partial_x u)^2)$. The Hamiltonian Regularized Euler (HRE) model, see ~\Cref{sec:HRE-conservative-variables}, is obtained when $k=1$, and represents the dissipation-free part of the HIGR model. Finally, the one-dimensional IGR model is given by
\begin{equation}
\left\{
\begin{aligned}
    &\rho_t + (\rho u)_x = 0 \,, \\
    &(\rho u)_t + (\rho u^2 + (\gamma - 1) \rho \varepsilon + \Sigma)_x = 0 \,, \\
    &E_t + ( (E + (\gamma - 1) \rho \varepsilon + \Sigma)u)_x = 0 \,, \\
    &\rho^{-1} \Sigma - \alpha (\rho^{-1} \Sigma_x)_x = 2 \alpha u_x^2 \,,
\end{aligned}
\right.
\end{equation}
where $\varepsilon = \rho^{-1}(E - (1/2) \rho u^2)$. 

We let $\gamma = 1.4$, the adiabatic index for dry air. We use a piecewise linear discontinuous-Galerkin (DG) method to discretize the model, SSPRK3 time-stepping, local Lax-Friedrichs flux for the hyperbolic conservation law, and using a Symmetric Interior Penalty Galerkin (SIPG) method for the elliptic solve. Further, we write the system in first order form by introducing auxiliary variables to project the derivatives into the DG approximation space. 

Explicitly, to compute the derivative of a quantity, e.g.\! $q_u \approx u_x$, in a piecewise linear DG finite element space, $V = \mathrm{P}^1 \mathrm{DG}$, we use a weak form with jump penalization to improve numerical stability: 
\begin{equation}
    \int_\Omega q_u \phi \, dx
    + \sum_{\text{faces}} \frac{C_{\mathrm{penalty}} \, p^2}{h} \, \llbracket q_u \rrbracket \, \llbracket \phi \rrbracket \, dS
    =
    - \int_\Omega u \, \partial_x \phi \, dx
    + \sum_{\text{faces}} \langle u \rangle \, \llbracket \phi \rrbracket \, dS \,, 
    \quad \forall \phi \in V \,,
\end{equation}
where $\llbracket \cdot \rrbracket$ denotes the jump across cell interfaces, $\langle \cdot \rangle$ denotes the average, $h$ is the cell size, $p$ is the polynomial order, and $C_{\mathrm{penalty}}$ is a user-defined penalty coefficient. We use $C_{\mathrm{penalty}} = 20$. Only a single auxiliary variable is used to implement IGR, whereas HIGR requires auxiliary variables for several fields: $u_x$, $(\ln(\rho))_x = \rho_x/\rho$, and $\varepsilon_x$. We solve the elliptic problem for $\Sigma$ using a symmetric interior penalty Galerkin (SIPG) formulation:
\begin{equation}
\begin{aligned}
& \int_\Omega \rho^{-1} \, \Sigma \, \phi \, dx
+ \alpha \int_\Omega \rho^{-1} \, \partial_x \Sigma \, \partial_x \phi \, dx
- \alpha \sum_{\text{faces}} \Big( \big\langle \rho^{-1} \, \partial_x \Sigma \big\rangle \, \llbracket \phi \rrbracket
+ \big\langle \rho^{-1} \, \partial_x \phi \big\rangle \, \llbracket \Sigma \rrbracket \Big) \, dS \\
& \quad + \alpha \sum_{\text{faces}} \frac{C_{\mathrm{penalty}} p^2}{h} \, \llbracket \Sigma \rrbracket \, \llbracket \phi \rrbracket \, dS
= \alpha \int_\Omega 2 q_u^2 \, \phi \, dx \,.
\end{aligned}
\end{equation}
The elliptic solve for HIGR is similar, with additional terms on the right-hand side. 

We use a smoothed, periodized Sod shock initial condition, see ~\Cref{fig:initial_condition}. This test features: (1) the propagation of regularized shock waves, (2) rarefaction waves, and (3) the collision of two shock waves. Smoothing the initial data is necessary, as IGR curtails shock steepening, but does not address already-shocked profiles. The domain is made periodic for simplicity, and we use the relatively high resolution of $N = 512$ cells in order to clearly illustrate the well-resolved solution profile of each model. The time-step is selected such that $\Delta t = 0.95 h/u_{max}$, where $u_{max}$ is the maximum wave-speed and $h=1/N$ is the grid spacing. Finally, we let $\alpha = 5 h^2$. 

 \begin{figure}
     \centering
     \pgfplotscreateplotcyclelist{mycolors}{
    {blue, line width=1pt},
    {red, line width=1pt},
    {green!50!black, line width=1pt},
    {orange, line width=1pt}
}

\begin{tikzpicture}
\begin{axis}[
    width=0.8\textwidth,           % full text width
    height=0.4\textwidth,       % 2:1 width:height ratio
    xmin=0, xmax=1,              % exact x-range, removes padding
    grid=both,
    xlabel={$x$},
    ylabel={},                   % remove y-label
    grid=both,
    cycle list name=mycolors,
    legend style={font=\small, at={(0.95,0.95)}, anchor=north east},
    title={},
    title style={yshift=1.5ex},   % shift title slightly up
]

\addplot table [x=x, y=rho, col sep=comma] {figures/data/higr_igr_comparison/nc512_higr_t_0.000000.csv};
\addplot table [x=x, y=u, col sep=comma] {figures/data/higr_igr_comparison/nc512_higr_t_0.000000.csv};
\addplot table [x=x, y=E, col sep=comma] {figures/data/higr_igr_comparison/nc512_higr_t_0.000000.csv};
\addplot table [x=x, y=eps, col sep=comma] {figures/data/higr_igr_comparison/nc512_higr_t_0.000000.csv};

\legend{$\rho$, $u$, $E$, $\varepsilon$}

\end{axis}
\end{tikzpicture}
     \caption{Periodized smoothed Sod shock initial data.}
     \label{fig:initial_condition}
 \end{figure}
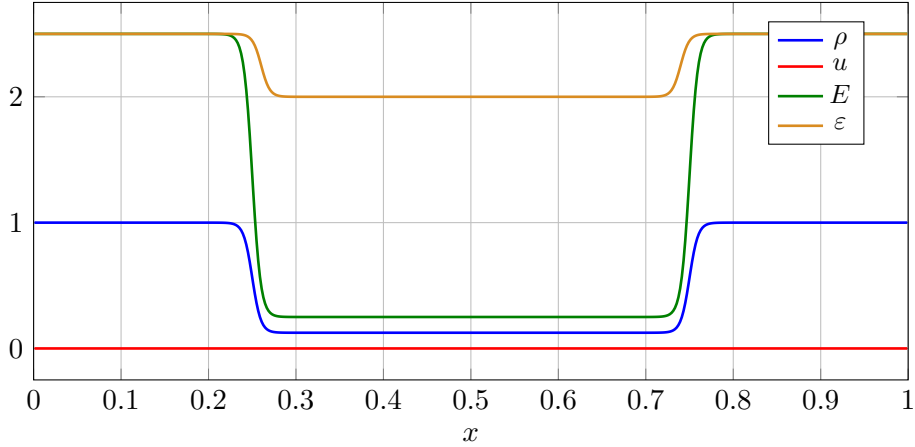

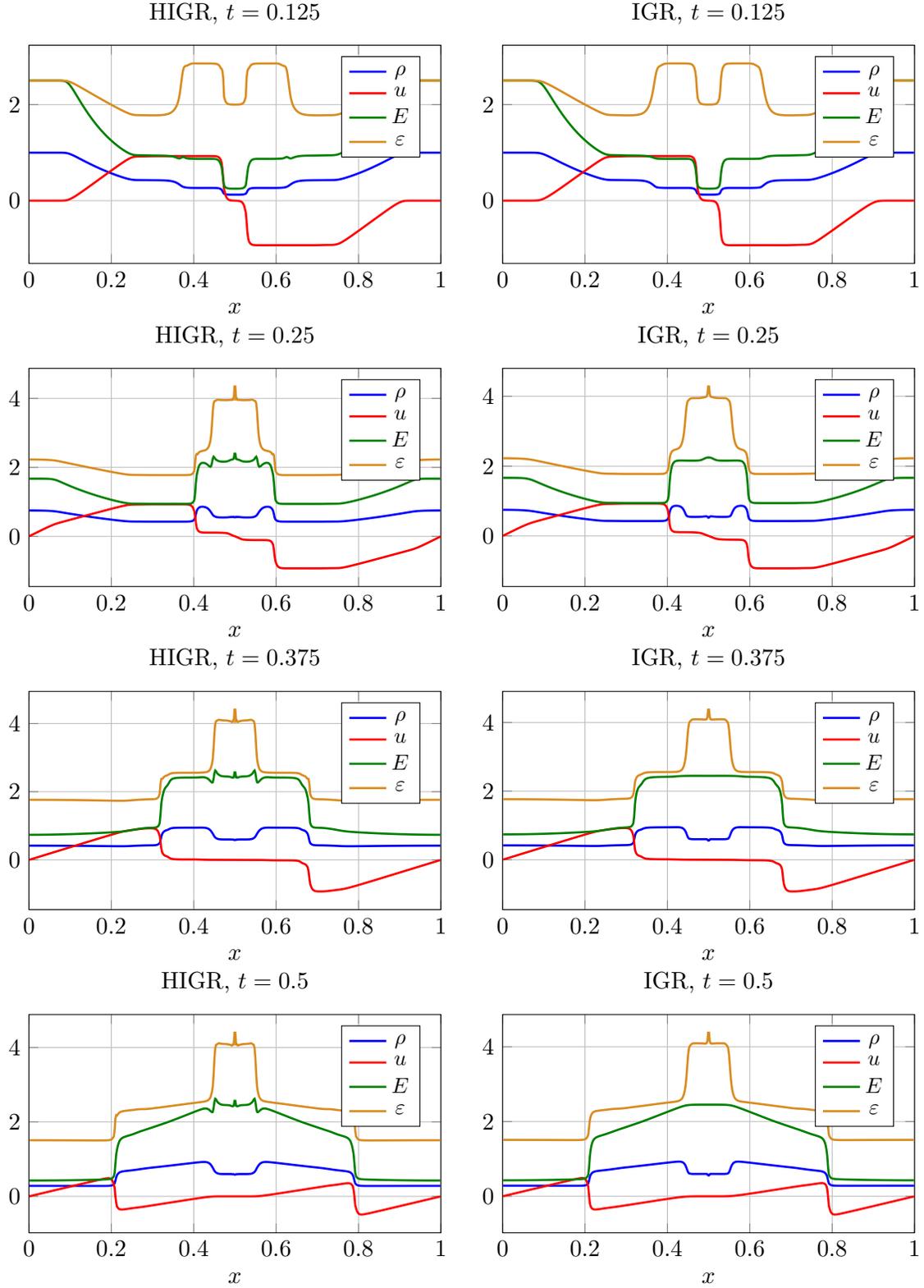
\begin{figure}
    \centering
    \pgfplotscreateplotcyclelist{mycolors}{
    {blue, line width=1pt},
    {red, line width=1pt},
    {green!50!black, line width=1pt},
    {orange, line width=1pt}
}

\begin{tikzpicture}
\begin{groupplot}[
    group style={
        group size=2 by 4,
        vertical sep=1.7cm,     % more vertical whitespace
        horizontal sep=1cm
    },
    xmin=0, xmax=1,              % exact x-range, removes padding
    width=0.5\textwidth,       % nearly half text width per plot
    height=0.31\textwidth,       % taller plots
    xlabel={$x$},
    ylabel={},                   % remove y-label
    grid=both,
    cycle list name=mycolors,
    legend style={font=\small, at={(0.95,0.95)}, anchor=north east},
]

% -------------------
% Row 1: t = 0.125
% -------------------
\nextgroupplot[title={HIGR, $t=0.125$}]
\addplot table [x=x, y=rho, col sep=comma] {figures/data/higr_igr_comparison/nc512_higr_t_0.125000.csv};
\addplot table [x=x, y=u, col sep=comma] {figures/data/higr_igr_comparison/nc512_higr_t_0.125000.csv};
\addplot table [x=x, y=E, col sep=comma] {figures/data/higr_igr_comparison/nc512_higr_t_0.125000.csv};
\addplot table [x=x, y=eps, col sep=comma] {figures/data/higr_igr_comparison/nc512_higr_t_0.125000.csv};
\legend{$\rho$, $u$, $E$, $\varepsilon$}

\nextgroupplot[title={IGR, $t=0.125$}]
\addplot table [x=x, y=rho, col sep=comma] {figures/data/higr_igr_comparison/nc512_igr_t_0.125000.csv};
\addplot table [x=x, y=u, col sep=comma] {figures/data/higr_igr_comparison/nc512_igr_t_0.125000.csv};
\addplot table [x=x, y=E, col sep=comma] {figures/data/higr_igr_comparison/nc512_igr_t_0.125000.csv};
\addplot table [x=x, y=eps, col sep=comma] {figures/data/higr_igr_comparison/nc512_igr_t_0.125000.csv};
\legend{$\rho$, $u$, $E$, $\varepsilon$}

% -------------------
% Row 2: t = 0.25
% -------------------
\nextgroupplot[title={HIGR, $t=0.25$}]
\addplot table [x=x, y=rho, col sep=comma] {figures/data/higr_igr_comparison/nc512_higr_t_0.250000.csv};
\addplot table [x=x, y=u, col sep=comma] {figures/data/higr_igr_comparison/nc512_higr_t_0.250000.csv};
\addplot table [x=x, y=E, col sep=comma] {figures/data/higr_igr_comparison/nc512_higr_t_0.250000.csv};
\addplot table [x=x, y=eps, col sep=comma] {figures/data/higr_igr_comparison/nc512_higr_t_0.250000.csv};
\legend{$\rho$, $u$, $E$, $\varepsilon$}

\nextgroupplot[title={IGR, $t=0.25$}]
\addplot table [x=x, y=rho, col sep=comma] {figures/data/higr_igr_comparison/nc512_igr_t_0.250000.csv};
\addplot table [x=x, y=u, col sep=comma] {figures/data/higr_igr_comparison/nc512_igr_t_0.250000.csv};
\addplot table [x=x, y=E, col sep=comma] {figures/data/higr_igr_comparison/nc512_igr_t_0.250000.csv};
\addplot table [x=x, y=eps, col sep=comma] {figures/data/higr_igr_comparison/nc512_igr_t_0.250000.csv};
\legend{$\rho$, $u$, $E$, $\varepsilon$}

% -------------------
% Row 3: t = 0.375
% -------------------
\nextgroupplot[title={HIGR, $t=0.375$}]
\addplot table [x=x, y=rho, col sep=comma] {figures/data/higr_igr_comparison/nc512_higr_t_0.375000.csv};
\addplot table [x=x, y=u, col sep=comma] {figures/data/higr_igr_comparison/nc512_higr_t_0.375000.csv};
\addplot table [x=x, y=E, col sep=comma] {figures/data/higr_igr_comparison/nc512_higr_t_0.375000.csv};
\addplot table [x=x, y=eps, col sep=comma] {figures/data/higr_igr_comparison/nc512_higr_t_0.375000.csv};
\legend{$\rho$, $u$, $E$, $\varepsilon$}

\nextgroupplot[title={IGR, $t=0.375$}]
\addplot table [x=x, y=rho, col sep=comma] {figures/data/higr_igr_comparison/nc512_igr_t_0.375000.csv};
\addplot table [x=x, y=u, col sep=comma] {figures/data/higr_igr_comparison/nc512_igr_t_0.375000.csv};
\addplot table [x=x, y=E, col sep=comma] {figures/data/higr_igr_comparison/nc512_igr_t_0.375000.csv};
\addplot table [x=x, y=eps, col sep=comma] {figures/data/higr_igr_comparison/nc512_igr_t_0.375000.csv};
\legend{$\rho$, $u$, $E$, $\varepsilon$}

% -------------------
% Row 4: t = 0.5
% -------------------
\nextgroupplot[title={HIGR, $t=0.5$}]
\addplot table [x=x, y=rho, col sep=comma] {figures/data/higr_igr_comparison/nc512_higr_t_0.500000.csv};
\addplot table [x=x, y=u, col sep=comma] {figures/data/higr_igr_comparison/nc512_higr_t_0.500000.csv};
\addplot table [x=x, y=E, col sep=comma] {figures/data/higr_igr_comparison/nc512_higr_t_0.500000.csv};
\addplot table [x=x, y=eps, col sep=comma] {figures/data/higr_igr_comparison/nc512_higr_t_0.500000.csv};
\legend{$\rho$, $u$, $E$, $\varepsilon$}

\nextgroupplot[title={IGR, $t=0.5$}]
\addplot table [x=x, y=rho, col sep=comma] {figures/data/higr_igr_comparison/nc512_igr_t_0.500000.csv};
\addplot table [x=x, y=u, col sep=comma] {figures/data/higr_igr_comparison/nc512_igr_t_0.500000.csv};
\addplot table [x=x, y=E, col sep=comma] {figures/data/higr_igr_comparison/nc512_igr_t_0.500000.csv};
\addplot table [x=x, y=eps, col sep=comma] {figures/data/higr_igr_comparison/nc512_igr_t_0.500000.csv};
\legend{$\rho$, $u$, $E$, $\varepsilon$}

\end{groupplot}
\end{tikzpicture}
    \caption{Comparison of the solution profiles of the HIGR (left column) and IGR (right column) systems for the smoothed, periodized Sod shock problem.}
    \label{fig:igr_higr_comparison}
\end{figure}

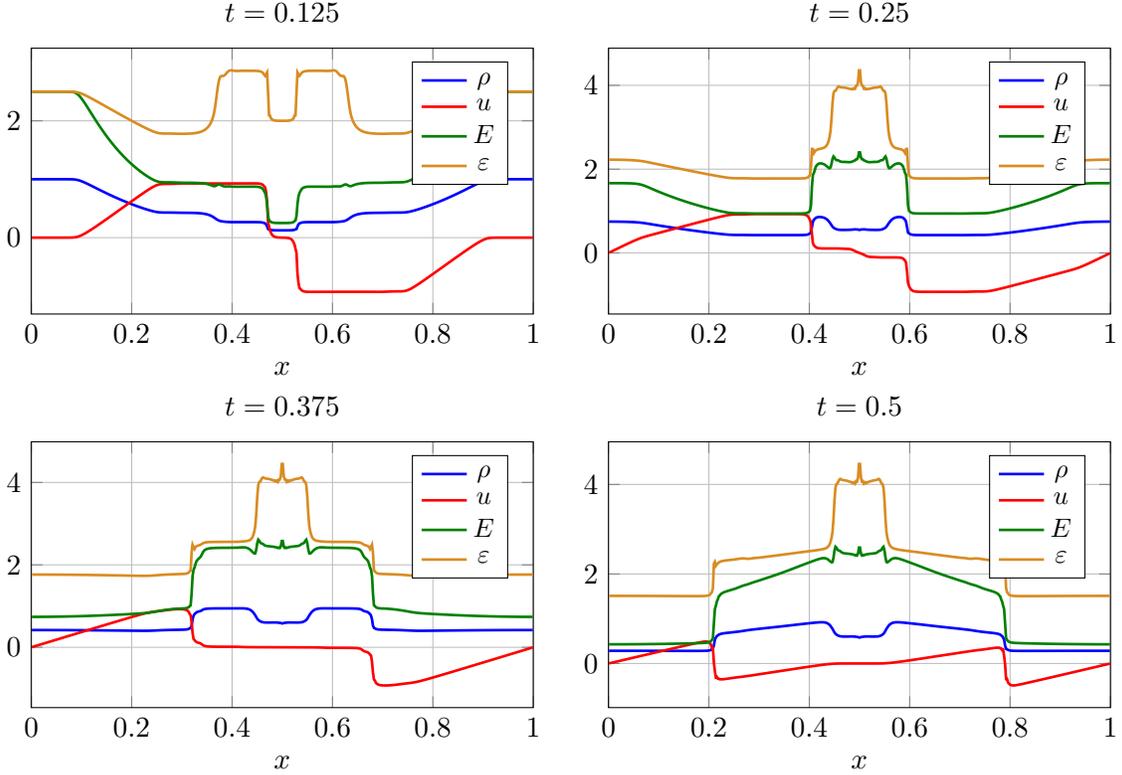
\begin{figure}
    \centering
    \pgfplotscreateplotcyclelist{mycolors}{
    {blue, line width=1pt},
    {red, line width=1pt},
    {green!50!black, line width=1pt},
    {orange, line width=1pt}
}

\begin{tikzpicture}
\begin{groupplot}[
    group style={
        group size=2 by 2,        % 2 columns x 2 rows
        vertical sep=1.7cm,       % similar spacing to original
        horizontal sep=1cm
    },
    xmin=0, xmax=1,              % exact x-range, removes padding
    width=0.50\textwidth,        % two columns should fit across the page
    height=0.31\textwidth,       % keep the taller aspect from original
    xlabel={$x$},
    ylabel={},                   % keep y-label empty as before
    grid=both,
    cycle list name=mycolors,
    legend style={font=\small, at={(0.95,0.95)}, anchor=north east},
]

% -------------------
% Top-left: t = 0.125 (HRE)
% -------------------
\nextgroupplot[title={$t=0.125$}]
\addplot table [x=x, y=rho, col sep=comma] {figures/data/higr_igr_comparison/nc512_hre_t_0.125000.csv};
\addplot table [x=x, y=u,   col sep=comma] {figures/data/higr_igr_comparison/nc512_hre_t_0.125000.csv};
\addplot table [x=x, y=E,   col sep=comma] {figures/data/higr_igr_comparison/nc512_hre_t_0.125000.csv};
\addplot table [x=x, y=eps, col sep=comma] {figures/data/higr_igr_comparison/nc512_hre_t_0.125000.csv};
\legend{$\rho$, $u$, $E$, $\varepsilon$}

% -------------------
% Top-right: t = 0.25 (HRE)
% -------------------
\nextgroupplot[title={$t=0.25$}]
\addplot table [x=x, y=rho, col sep=comma] {figures/data/higr_igr_comparison/nc512_hre_t_0.250000.csv};
\addplot table [x=x, y=u,   col sep=comma] {figures/data/higr_igr_comparison/nc512_hre_t_0.250000.csv};
\addplot table [x=x, y=E,   col sep=comma] {figures/data/higr_igr_comparison/nc512_hre_t_0.250000.csv};
\addplot table [x=x, y=eps, col sep=comma] {figures/data/higr_igr_comparison/nc512_hre_t_0.250000.csv};
\legend{$\rho$, $u$, $E$, $\varepsilon$}

% -------------------
% Bottom-left: t = 0.375 (HRE)
% -------------------
\nextgroupplot[title={$t=0.375$}]
\addplot table [x=x, y=rho, col sep=comma] {figures/data/higr_igr_comparison/nc512_hre_t_0.375000.csv};
\addplot table [x=x, y=u,   col sep=comma] {figures/data/higr_igr_comparison/nc512_hre_t_0.375000.csv};
\addplot table [x=x, y=E,   col sep=comma] {figures/data/higr_igr_comparison/nc512_hre_t_0.375000.csv};
\addplot table [x=x, y=eps, col sep=comma] {figures/data/higr_igr_comparison/nc512_hre_t_0.375000.csv};
\legend{$\rho$, $u$, $E$, $\varepsilon$}

% -------------------
% Bottom-right: t = 0.5 (HRE)
% -------------------
\nextgroupplot[title={$t=0.5$}]
\addplot table [x=x, y=rho, col sep=comma] {figures/data/higr_igr_comparison/nc512_hre_t_0.500000.csv};
\addplot table [x=x, y=u,   col sep=comma] {figures/data/higr_igr_comparison/nc512_hre_t_0.500000.csv};
\addplot table [x=x, y=E,   col sep=comma] {figures/data/higr_igr_comparison/nc512_hre_t_0.500000.csv};
\addplot table [x=x, y=eps, col sep=comma] {figures/data/higr_igr_comparison/nc512_hre_t_0.500000.csv};
\legend{$\rho$, $u$, $E$, $\varepsilon$}

\end{groupplot}
\end{tikzpicture}
    \caption{Solution profile of the HRE system for the smoothed, periodized Sod shock problem.}
    \label{fig:hre_sod_shock}
\end{figure}

~\Cref{fig:igr_higr_comparison} shows a comparison of the solution profiles in the HIGR and IGR systems, and ~\Cref{fig:hre_sod_shock} shows the solution profile the HRE model. The solution profiles for the HIGR and HRE models are nearly identical, with the HRE model exhibiting more oscillations from numerical instability and enhancing the cusp-like singularities observed in both profiles. The regularized shock fronts propagate at the same rate with comparable shapes in all three models, especially prior to the collision of the counter-propagating shocks, with the HIGR and HRE models exhibiting several new defects not present in the IGR model, discussed below. Each solution produces spurious heat at $x=0$ after the shocks have collided, leading to a small defect in the density profile. This defect arises, because the entropy-producing part of the regularization deposits a small amount of heat to keep the profile smooth, leading to a local pressure spike and a local reduction in density. Note that the numerical solution of HRE produces entropy, despite the conservative nature of the model due to the dissipative local Lax-Friedrichs flux. The HIGR and HRE profiles have a marginally smaller defect in the density profile compared with IGR, but a more pronounced defect in the total energy profile. We believe that the more complex coupling of the pressure-driven and inertial dynamics are responsible for these defects observed in HIGR and HRE. In particular, a noticeable oscillatory structure arises in the total energy profile whenever a shock wave passes through a stationary density interface (a regularized contact discontinuity). This is because the entropic pressure no longer depends only on the gradient of the velocity field as in IGR (which localizes its impact to shock fronts), but rather also depends on density gradients. This has the unintended consequence of modifying the solution profile at contact discontinuities, and not just shock fronts. Finally, the HIGR profile is more oscillatory than IGR due to the difficulty of stabilizing the higher-order derivatives appearing in the HIGR model. By removing these higher-order derivatives, one may mitigate this issue, see the ablation study in ~\Cref{appendix:ablation_study}. On the other hand, the HRE model, which is fully-conservative, exhibits more numerical instability with spiky features. This demonstrates the need for some positive amount of entropy production to obtain numerically-stable and physically-meaningful solution profiles.  

In addition to these problems with the HIGR model that are visible in the solution profile, it is also worth noting that a stable implementation of a DG method for HIGR is significantly more difficult to obtain than for IGR. The HIGR method is sensitive to our choice of $\alpha$ parameter, becoming unstable for $\alpha$ much larger than $\alpha = 5h^2$. Moreover, we were unable to stabilize DG methods for HIGR with quadratic or higher interpolation. On the other hand, IGR is highly robust with respect to discretization and parameter choice, admitting stable solutions with the local Lax-Friedrichs flux nearly eliminated (maintaining stability with the dissipative flux scaled to $10\%$ its typical value), for large or small $\alpha$, and for quadratic interpolation. The simplicity of IGR is a significant advantage. 

These results strongly favor IGR as the more practical regularization compared to HIGR. Nonetheless, the derivation of HIGR sheds light on one of the few defects observed in solution profiles obtained from IGR: namely, the internal energy spike arising from colliding shocks, see ~\Cref{fig:energy_spikes}. Indeed, IGR couples the entropic pressure with internal energy such that the entirety of entropic pressure functions like a source/sink term rather than being fully incorporated into the variational structure. We will see in ~\Cref{sec:pressureless-IGR} that this treatment is somewhat pessimistic in terms of its entropy production: the entropic pressure decomposes into conservative and dissipative parts, with the conservative component having a Hamiltonian formulation. The HIGR model attempts to leverage the Hamiltonian formalism to consistently treat the coupling of IGR with thermodynamics, with the aim of reducing the artificial dissipation responsible for the spikes. However, the numerical results show that this approach comes at a high cost, substantially increasing model complexity and adding new oscillations in the total-energy profile. While the HIGR model does not seem to be extremely useful in and of itself, its derivation helps to clarify the intricate manner in which thermodynamics interacts with IGR. 

\subsection{Outline}\label{subsec:outline}

%\wjb{Here's a basic overview of the structure of the paper. We can workshop this once we have the graphical aids.}

The paper follows the derivation and analysis of five closely related inviscid regularizations of compressible flow. We call these models pressureless IGR, pressureless Hamiltonian regularized Euler (HRE), thermodynamic IGR, thermodynamic HRE, and Hamiltonian IGR. Of particular interest are the energy transport and entropy production implied by these models. 

Pressureless IGR and thermodynamic IGR were derived and studied previously \cite{cao2023information, cao2024information, wilfong2025simulatingmanyenginespacecraftexceeding}; here, we consider the local energy transport and entropy production implied by these models. The only force appearing in the momentum equation in pressureless IGR is the entropic pressure arising from the information geometric regularization: no barotropic or thermodynamic pressure is considered. In Section~\ref{sec:Hamiltonian-sub-system}, we derive the pressureless HRE model as a Hamiltonian subsystem of pressureless IGR by partitioning the entropic pressure into conservative (energy and entropy conserving) and non-conservative parts. In Section~\ref{sec:pressureless-IGR}, we derive the local kinetic energy transport equation of the pressureless IGR and HRE models. In Section~\ref{sec:dispersion-free-extension}, we derive the thermodynamic HRE model, which extends pressureless HRE by adding an internal energy to the Hamiltonian. A nonphysical dispersion relation in acoustic waves arises if this is done na{\"i}vely, due to the $H(\mathrm{div})$-type modified kinetic energy metric of HRE. This obstacle is overcome by adding a corresponding regularization to the internal energy, which resembles a capillary energy, so as to match the linear characteristics with those of compressible Euler. In Section~\ref{sec:HRE-conservative-variables}, we show that the thermodynamic HRE model can be written in locally conservative variables: mass, momentum, and total energy. In Section~\ref{sec:full-IGR-dissipative-extension}, we derive the Hamiltonian IGR model, which extends the thermodynamic HRE model by adding in the non-conservative force from IGR using a modification of the metriplectic formalism. The resulting model arises from a bracket-based formalism, locally conserves mass, momentum, and energy, and enjoys sign-definite entropy production relative to the sign of the divergence of the velocity. 

In brief, the paper proceeds by identifying a Hamiltonian subsystem of IGR, and sequentially extending this system using the Hamiltonian and metriplectic formalisms. See \Cref{tab:system-properties} for a summary of these models and their thermodynamic behavior. We include barotropic IGR in this table for completeness, as IGR was originally introduced in the barotropic setting \cite{cao2023information}, but focus on the other five models in this work. Moreover, while a Hamiltonian regularization of the barotropic Euler equation is not explicitly referenced in \Cref{tab:system-properties}, it may be regarded as a special case of the thermodynamic HRE model in which the equation of state depends on density alone. 

Some useful background knowledge and technical derivations are moved to the appendix. In \Cref{appendix:bracket_formalisms}, we provide a brief review of Hamiltonian and metriplectic modeling formalisms in mechanics. In addition to defining Hamiltonian and metriplectic systems, we describe the energy--Casimir method for studying the linear stability of Hamiltonian systems. In \Cref{appendix:igr_elliptic_operators}, we establish some properties of the elliptic operators which appear in IGR models of compressible flow, e.g.\! commutation relations. In \Cref{appendix:change_of_variables}, we provide a detailed derivation of the noncanonical Hamiltonian structure of the HRE model in Eulerian coordinates. Finally, ~\Cref{appendix:ablation_study} illustrates how the deletion of higher-order derivatives improves numerical stability in the HIGR system. 

\begin{center}
\small
\begin{table}
\begin{tabular}{ | c | c c c c c c | }
 \hline
  & \makecell{Pressureless \\ IGR} & \makecell{Pressureless \\ HRE} & \makecell{Barotropic \\ IGR} & \makecell{Thermo. \\ IGR} & \makecell{Thermo. \\ HRE} & \makecell{Hamiltonian \\ IGR} \\ 
  & \eqref{eq:pressureless-IGR} & \eqref{eq:pressureless-HRE} & \eqref{eq:standard_igr} & \eqref{eq:standard_igr}-\eqref{eq:standard_igr_energy} & \eqref{eq:thermal-HRE} & \eqref{eq:Hamiltonian-IGR} \\
 \hline  & & & & & & 
 \\ [-0.65 em]
 \makecell{Reference} & \cite{cao2023information} & new  & \cite{cao2023information} & \cite{wilfong2025simulatingmanyenginespacecraftexceeding} & new & new \\[5 pt]
 \hdashline
 \makecell{Conserves \\ total energy} & \markno & \markyes & \markno & \markyes & \markyes & \markyes \\
 \hdashline
 \makecell{Dissipative entropic \\ pressure coupling} & \markyes & \markno & \markyes & \markyes & \markno & \markyes \\
 \hdashline
 \makecell{Sign-matching \\ entropy production} & n/a & n/a & n/a & \markno & \markyes & \markyes \\ 
 \hline
\end{tabular}
\caption{A summary of the basic properties of the regularized compressible flow models considered in this paper. }\label{tab:system-properties}
\end{table}
\end{center}

 %\hdashline
 %\makecell{Conserves \\ kinetic energy} & \markno & \markyes  & \markno & \markno & \markno \\ 

%================================================================
\section{A Hamiltonian sub-system of pressureless IGR}\label{sec:Hamiltonian-sub-system}

Let $\Omega \subset \mathbb{R}^d$. Fluid mechanics in the Lagrangian reference frame is described by a flow map, $\bm{\Phi}_t: \Omega \to \Omega$, which is a one-parameter family of diffeomorphisms of the fluid domain, indexed by time $t$. The group of diffeomorphisms of the fluid domain, $\mathrm{Diff}(\Omega)$, is the natural configuration space for geometric fluid mechanics in the Lagrangian reference frame \cite{arnold1966geometrie, ebin1970groups, MaRa1999, VaMa2005}. Classical ideal fluid mechanics then arises from an action principle on $T\mathrm{Diff}(\Omega)$. Given a metric, $g[\bm{\Phi}_t]: T_{\bm{\Phi}_t}\mathrm{Diff}(\Omega) \times T_{\bm{\Phi}_t}\mathrm{Diff}(\Omega) \to \mathbb{R}$,
\begin{equation}
    (\mathbf{U},\mathbf{V}) \mapsto g[\bm{\Phi}_t](\mathbf{U},\mathbf{V}) \in \mathbb{R} \,,
\end{equation}
the Lagrangian flow map is obtained as a stationary point of the action:
\begin{equation}
    \mathcal{S}[\bm{\Phi}_t]
    =
    \frac{1}{2}
    \int_{t_1}^{t_2}
    g[\bm{\Phi}_t](\dot{\bm{\Phi}}_t, \dot{\bm{\Phi}}_t) \mathsf{d} t \,.
\end{equation}
This yields geodesic flow associated with the metric-compatible Levi-Civita connection. This connection is not the one which generates the flow defining the information geometric regularization \cite{cao2023information}. The IGR model is derived as a geodesic flow with respect to a different connection structure, the dual connection of a Hessian manifold, which rigorously enforces the positivity of the Jacobian determinant of the flow map. However, the Levi-Civita and dual connections on a Hessian manifold are closely related, in fact, they are scalar multiples of each other \cite{leok2017connecting}.

Despite the fact that the Levi-Civita flow does not robustly enforce positivity of the Jacobian determinant in the way that the dual connection flow does (indeed, the boundary of the feasible set can be approached approached exponentially fast), this flow has the advantage of conserving a generalized kinetic energy arising from the metric: the Levi-Civita connection is the unique torsion-free connection with this property \cite{levi1916nozione, lee2018introduction}. Hence, we will proceed in this derivation so as to interpret as much of the IGR model as possible in terms of an action principle, and will find that a significant portion of the IGR model is recovered by the HRE model. This allows us to partition the entropic pressure, which enforces the positivity of the Jacobian of the flow map, into conservative (Hamiltonian) and non-conservative (non-Hamiltonian) parts. We call the flow based on the Levi-Civita connection the Hamiltonian regularized Euler (HRE) system. In particular, because we consider no potential energy in this section, the Levi-Civita geodesic flow based on the regularized kinetic energy metric is called the pressureless HRE model. 

\subsection{The IGR metric structure}\label{subsec:regularized-metric}

Different fluid models can be obtained from choosing various metric structures on $\mathrm{Diff}(\Omega)$. Many generalized fluid models may be described in this manner: e.g.\! Camassa--Holm \cite{camassa1993integrable, misiolek1998shallow, kouranbaeva1999camassa, khesin2003euler} and LANS/LAE-$\alpha$ \cite{holm1998euler, holm1999nonlinear, holm2002lagrangian} models use an $H^1$ kinetic energy norm, and the Hunter--Saxton equation \cite{hunter1991dynamics, lenells2008hunter, khesin2003euler} uses an $H^1$ seminorm. More generally, a family of such models sometimes known as the Euler--Arnold or EPDiff (Euler--Poincar{\'e} Diffeomorphism) equations \cite{khesin2003euler, holm1998euler, constantin2003geodesic, holm2005momentum} arise from choosing various kinetic energy metrics, usually a Sobolev-type norm. Frequently, these models are assumed to be of incompressible fluids, although nothing explicitly requires this. Indeed, LANS-$\alpha$ models for compressible flow have been studied \cite{bhat2005lagrangian}. The pressureless HRE model fits into this general context, being the geodesic flow associated with a weighted $H(\mathrm{div})$ metric. This model and IGR distinguish themselves in how their kinetic energy metric is selected so as to enforce positivity of the Jacobian determinant of the flow map. 

IGR uses a special type of Hessian metric. To begin, define the following convex potential:
\begin{equation}
    \psi_0[\bm{\Phi}_t]
    =
    \frac{1}{2}
    \int_{\Omega} | \bm{\Phi}_t(\mathbf{X}_0) |^2 \mathsf{d}^d \mathbf{X}_0 \,.
\end{equation}
Its second variation yields the standard kinetic energy metric:
\begin{equation}
    g_0(\mathbf{U}, \mathbf{V})
    =
    D^2 \psi_0[\mathbf{U},\mathbf{V}]
    =
    \int_\Omega \mathbf{U} \cdot \mathbf{V} \mathsf{d}^d \mathbf{X}_0 \,.
\end{equation}
The information geometric regularization \cite{cao2023information, cao2024information} proceeds by adding a logarithmic-barrier to avoid a convex region of space whose boundary is given by the zero level-set of a convex function, $f: \mathrm{Diff}(\Omega) \to \mathbb{R}$. We define the potential to be
\begin{equation}
    \psi_\alpha[\bm{\Phi}_t]
    =
    \int_{\Omega} \left( \frac{1}{2} | \bm{\Phi}_t(\mathbf{X}_0) |^2 -
    \alpha \log(f(\bm{\Phi}_t(\mathbf{X}_0))) \right) \mathsf{d}^d \mathbf{X}_0 \,. 
\end{equation}
A metric is then obtained by taking the Hessian of this generating function. 

The flow map becomes singular when $\det(\nabla_{\mathbf{X}} \bm{\Phi}_t) = 0$. Therefore, we would like to penalize $f(\bm{\Phi}_t) = \det(\nabla_{\mathbf{X}} \bm{\Phi}_t)$. However, the determinant is not convex unless it is restricted to symmetric positive definite matrices. This difficulty may be circumvented by embedding in a higher dimensional space. Define $\iota: \mathrm{Diff}(\Omega) \to \mathrm{Diff}(\Omega) \times \mathbb{R}$ via
\begin{equation}
    \iota(\bm{\Phi})
    =
    (\bm{\Phi}, \det(\nabla_{\mathbf{X}} \bm{\Phi})) \,,
\end{equation}
and define $\tilde{\psi}_\alpha: \mathrm{Diff}(\Omega) \times \mathbb{R} \to \mathbb{R}$ to be
\begin{equation}
    \tilde{\psi}_\alpha[\bm{\Phi}, \Phi']
    =
    \int_{\Omega} \left( \frac{1}{2} | \bm{\Phi}(\mathbf{X}_0) |^2 -
    \alpha \log(\Phi') \right) \mathsf{d}^d \mathbf{X}_0 \,.
\end{equation}
Noting Jacobi's identity,
\begin{equation}
    D \iota[\bm{\Phi}_t]\mathbf{U}
    =
    (\mathbf{U}, \det(\nabla_{\mathbf{X}} \bm{\Phi}_t) \mathrm{tr}( \nabla_{\mathbf{X}} \bm{\Phi}_t^{-1} \nabla_{\mathbf{X}} \mathbf{U} ) ) \,,
\end{equation}
the metric structure is defined by pulling back the second variation of $\tilde{\psi}_\alpha$ by $\iota$, i.e.,
\begin{equation}
\begin{aligned}
    g_\alpha[\bm{\Phi}_t](\mathbf{U},\mathbf{V})
    &=
    (\iota^* D^2 \tilde{\psi}_\alpha)[\bm{\Phi}_t](\mathbf{U},\mathbf{V})
    =
    D^2 \tilde{\psi}_\alpha[\iota(\bm{\Phi}_t)](D \iota[\bm{\Phi}_t]\mathbf{U}, D \iota[\bm{\Phi}_t]\mathbf{V}) \\
    &=
    \int_\Omega \left( \mathbf{U} \cdot \mathbf{V}
    +
    \alpha \mathrm{tr}(\nabla_{\mathbf{X}} \bm{\Phi}_t^{-1} \nabla_{\mathbf{X}} \mathbf{U}) \mathrm{tr}(\nabla_{\mathbf{X}} \bm{\Phi}_t^{-1} \nabla_{\mathbf{X}} \mathbf{V})  \right) \mathsf{d}^d \mathbf{X}_0 \,.
\end{aligned}
\end{equation}
In the Eulerian reference frame, letting $\mathbf{u}(\bm{\Phi}(\mathbf{X})) = \dot{\bm{\Phi}}(\mathbf{X})$ (similarly for $\mathbf{v}$), and $\rho(\mathbf{X}) = \det(\nabla_{\mathbf{X}} \bm{\Phi})^{-1}$, this metric becomes
\begin{equation}
    g_\alpha(\mathbf{u}, \mathbf{v})
    =
    \int_\Omega \rho \left( \mathbf{u} \cdot \mathbf{v} + \alpha (\nabla_{\mathbf{x}} \cdot \mathbf{u}) (\nabla_{\mathbf{x}} \cdot \mathbf{v}) \right) \mathsf{d}^d \mathbf{x} \,.
\end{equation}
Hence, as mentioned previously, this is simply a weighted $H(\mathrm{div})$ inner product. 

\subsection{The regularized Euler--Lagrange equations}\label{subsec:regularized-EL-equations}
Now, we would like to obtain the equations of motion (i.e., the Euler--Lagrange equations) corresponding to geodesics of this metric and compare to the standard IGR equations \eqref{eq:standard_igr}. In doing so, this allows us to extract the ``conservative'' component of the IGR model. The Lagrangian functional is defined as
\begin{equation}
    L[\bm{\Phi}_t, \dot{\bm{\Phi}}_t]
    =
    \frac{1}{2}
    g_\alpha[\bm{\Phi}_t](\dot{\bm{\Phi}}_t, \dot{\bm{\Phi}}_t)
    =
    \frac{1}{2}
    D^2 \tilde{\psi}_\alpha[\iota(\bm{\Phi}_t)](D \iota[\bm{\Phi}_t] \dot{\bm{\Phi}}_t, D \iota[\bm{\Phi}_t] \dot{\bm{\Phi}}_t) \,,
\end{equation}
and the action functional is then defined to be
\begin{equation}
    \mathcal{S}[\bm{\Phi}_t]
    =
    \int_{t_1}^{t_2}
    L[\bm{\Phi}_t, \dot{\bm{\Phi}}_t] \mathsf{d} t
    =
    \frac{1}{2}
    \int_{t_1}^{t_2}
    \int_\Omega \left( | \dot{\bm{\Phi}}_t |^2
    +
    \alpha \mathrm{tr}(\nabla_{\mathbf{X}} \bm{\Phi}_t^{-1} \nabla_{\mathbf{X}} \dot{\bm{\Phi}}_t)^2  \right) \mathsf{d}^d \mathbf{X}_0 \mathsf{d} t \,.
\end{equation}
Stationary points of this action functional yield the Euler--Lagrange equation: for all variations $\mathbf{U}$,
\begin{equation}\label{eq:EL-generic}
    0 = \delta \mathcal{S}[\bm{\Phi}_t] \cdot \mathbf{U} = \int_{t_1}^{t_2} \left( D_{\bm{\Phi}_t} L - \frac{\mathsf{d}}{\mathsf{d}t} D_{\dot{\bm{\Phi}}_t} L \right) \cdot \mathbf{U} \mathsf{d} t \,.
\end{equation}
Equivalently, one may obtain the evolution equations in Hamiltonian form, see appendix \ref{appendix:canonical_pressureless_HRE}. 

To obtain an expression for the Euler--Lagrange equation \eqref{eq:EL-generic}, we need to take derivatives of the Lagrangian. The derivative with respect to $\dot{\bm{\Phi}}_t$ is simply
\begin{equation}
    D_{\dot{\bm{\Phi}}_t} L[\bm{\Phi}_t, \dot{\bm{\Phi}}_t] \mathbf{U}
    =
    D^2 \tilde{\psi}_\alpha[\iota(\bm{\Phi}_t)](D \iota[\bm{\Phi}_t] \dot{\bm{\Phi}}_t, D \iota[\bm{\Phi}_t] \mathbf{U}) \,,
\end{equation}
because of the symmetry of the metric. On the other hand, the derivative with respect to $\bm{\Phi}_t$ is obtained by taking the variation of the pullback metric. Using the chain rule, we have:
\begin{multline}
    D_{\bm{\Phi}_t} L[\bm{\Phi}_t, \dot{\bm{\Phi}}_t] \mathbf{U}
    =
    \frac{1}{2}
    D_{\bm{\Phi}_t} \left( D^2 \tilde{\psi}_\alpha[\iota(\bm{\Phi}_t)](D \iota[\bm{\Phi}_t] \dot{\bm{\Phi}}_t, D \iota[\bm{\Phi}_t] \dot{\bm{\Phi}}_t) \right) \mathbf{U}
    \\
    =
    \frac{1}{2} D^3 \tilde{\psi}_\alpha[\iota(\bm{\Phi}_t)](D \iota[\bm{\Phi}_t] \dot{\bm{\Phi}}_t, D \iota[\bm{\Phi}_t] \dot{\bm{\Phi}}_t, D \iota[\bm{\Phi}_t] \mathbf{U}) 
    \\
    + D^2 \tilde{\psi}_\alpha[\iota(\bm{\Phi}_t)](D^2 \iota[\bm{\Phi}_t](\dot{\bm{\Phi}}_t, \mathbf{U}), D \iota[\bm{\Phi}_t] \dot{\bm{\Phi}}_t) \,.
\end{multline}
Next, it is necessary to compute the time derivative:
\begin{multline}
    \frac{\mathsf{d}}{\mathsf{d} t} \left( D^2 \tilde{\psi}_\alpha[\iota(\bm{\Phi}_t)](D \iota[\bm{\Phi}_t] \dot{\bm{\Phi}}_t, D \iota[\bm{\Phi}_t] \mathbf{U}) \right) \\
    =
    D^3 \tilde{\psi}_\alpha[\iota(\bm{\Phi}_t)](D \iota[\bm{\Phi}_t] \dot{\bm{\Phi}}_t, D \iota[\bm{\Phi}_t] \dot{\bm{\Phi}}_t, D \iota[\bm{\Phi}_t] \mathbf{U}) 
    + D^2 \tilde{\psi}_\alpha[\iota(\bm{\Phi}_t)](D \iota[\bm{\Phi}_t] \ddot{\bm{\Phi}}_t, D \iota[\bm{\Phi}_t] \mathbf{U}) \\
    +
    D^2 \tilde{\psi}_\alpha[\iota(\bm{\Phi}_t)](D^2 \iota[\bm{\Phi}_t](\dot{\bm{\Phi}}_t, \dot{\bm{\Phi}}_t), D \iota[\bm{\Phi}_t] \mathbf{U}) 
    + 
    D^2 \tilde{\psi}_\alpha[\iota(\bm{\Phi}_t)](D^2 \iota[\bm{\Phi}_t](\dot{\bm{\Phi}}_t, \mathbf{U}), D \iota[\bm{\Phi}_t] \dot{\bm{\Phi}}_t) 
    \,.
\end{multline}
Therefore, we find the Euler--Lagrange equation in weak form:
\begin{multline}
    0 = \int_{t_1}^{t_2} \bigg[ 
    - D^2 \tilde{\psi}_\alpha[\iota(\bm{\Phi}_t)](D \iota[\bm{\Phi}_t] \ddot{\bm{\Phi}}_t, D \iota[\bm{\Phi}_t] \mathbf{U}) \\
    - \frac{1}{2} D^3 \tilde{\psi}_\alpha[\iota(\bm{\Phi}_t)](D \iota[\bm{\Phi}_t] \mathbf{U}, D \iota[\bm{\Phi}_t] \dot{\bm{\Phi}}_t, D \iota[\bm{\Phi}_t] \dot{\bm{\Phi}}_t)  
    \\
    - D^2 \tilde{\psi}_\alpha[\iota(\bm{\Phi}_t)](D^2 \iota[\bm{\Phi}_t](\dot{\bm{\Phi}}_t, \dot{\bm{\Phi}}_t), D \iota[\bm{\Phi}_t] \mathbf{U}) 
    \bigg] \mathsf{d} t 
    \,.
\end{multline}
Since this must hold for all variations $\mathbf{U}$, we obtain the strong form of the equation, which is the pressureless HRE system in the Lagrangian frame,
\begin{multline}\label{eq:pressureless-HRE-abstract}
    (D \iota[\bm{\Phi}_t])^* (D^2 \tilde{\psi}_\alpha[\iota(\bm{\Phi}_t)]) (D \iota[\bm{\Phi}_t] \ddot{\bm{\Phi}}_t, \cdot) \\
    =
    - (D \iota[\bm{\Phi}_t])^* \bigg[
    D^2 \tilde{\psi}_\alpha[\iota(\bm{\Phi}_t)](D^2 \iota[\bm{\Phi}_t](\dot{\bm{\Phi}}_t, \dot{\bm{\Phi}}_t), \cdot)
    +
    \frac{1}{2} D^3 \tilde{\psi}_\alpha[\iota(\bm{\Phi}_t)](D \iota[\bm{\Phi}_t] \dot{\bm{\Phi}}_t, D \iota[\bm{\Phi}_t] \dot{\bm{\Phi}}_t, \cdot)
    \bigg] 
    \,.
\end{multline}
In contrast, the pressureless IGR model \cite{cao2023information} is obtained from a geodesic equation of the form
\begin{multline}
    (D \iota[\bm{\Phi}_t])^* (D^2 \tilde{\psi}_\alpha[\iota(\bm{\Phi}_t)]) (D \iota[\bm{\Phi}_t] \ddot{\bm{\Phi}}_t, \cdot) \\
    =
    - (D \iota[\bm{\Phi}_t])^* \bigg[
    D^2 \tilde{\psi}_\alpha[\iota(\bm{\Phi}_t)](D^2 \iota[\bm{\Phi}_t](\dot{\bm{\Phi}}_t, \dot{\bm{\Phi}}_t), \cdot)
    +
    D^3 \tilde{\psi}_\alpha[\iota(\bm{\Phi}_t)](D \iota[\bm{\Phi}_t] \dot{\bm{\Phi}}_t, D \iota[\bm{\Phi}_t] \dot{\bm{\Phi}}_t, \cdot)
    \bigg] 
    \,.
\end{multline}
This factor of ``$1/2$'' is a crucial difference leading to qualitatively different solution profiles. 

\begin{remark}
In the unidimensional case, it has been seen that the strict energy conservation of $H^1$-type Euler--Arnold flows leads to cusp-like singularities \cite{liu2019well, pu2018weakly}. The non-conservative terms in the IGR flow mitigate these effects, maintaining global strong solutions \cite{cao2024information}. On the other hand, the strict energy conservation implied by the Levi-Civita flow mitigates shock formation by letting the derivative of the flow map grow without bound (provided the singularity remain integrable). 
\end{remark}

We can write the pressureless HRE system \eqref{eq:pressureless-HRE-abstract} in more familiar notation by observing the following:
\begin{equation}
    D^2 \tilde{\psi}_\alpha[\bm{\Phi}, \Phi']((\mathbf{U},U'), (\mathbf{V},V'))
    =
    \int_\Omega 
    \left(
    \mathbf{U} \cdot \mathbf{V}
    +
    \alpha
    \frac{U' V'}{(\Phi')^2}
    \right)
    \mathsf{d}^d \mathbf{X}_0 \,,
\end{equation}
\begin{equation}
    D^3 \tilde{\psi}_\alpha[\bm{\Phi}, \Phi']((\mathbf{U},U'), (\mathbf{V},V'), (\mathbf{W},W'))
    =
    -
    2 \alpha
    \int_\Omega
    \frac{U' V' W'}{(\Phi')^3} \mathsf{d}^d \mathbf{X}_0 \,,
\end{equation}
\begin{equation}
    D \iota[\bm{\Phi}_t]\mathbf{U}
    =
    (\mathbf{U}, \det(\nabla_{\mathbf{X}} \bm{\Phi}_t) \mathrm{tr}( \nabla_{\mathbf{X}} \bm{\Phi}_t^{-1} \nabla_{\mathbf{X}} \mathbf{U} ) ) \,,
\end{equation}
and
\begin{multline}
    D^2 \iota[\bm{\Phi}_t](\mathbf{U},\mathbf{V})
    =
    (0, \det(\nabla_{\mathbf{X}} \bm{\Phi}_t) ( \mathrm{tr}( \nabla_{\mathbf{X}} \bm{\Phi}_t^{-1} \nabla_{\mathbf{X}} \mathbf{U} ) \mathrm{tr}( \nabla_{\mathbf{X}} \bm{\Phi}_t^{-1} \nabla_{\mathbf{X}} \mathbf{V}) \\
    - \mathrm{tr}( (\nabla_{\mathbf{X}} \bm{\Phi}_t^{-1} \nabla_{\mathbf{X}} \mathbf{U}) (\nabla_{\mathbf{X}} \bm{\Phi}_t^{-1} \nabla_{\mathbf{X}} \mathbf{V})) ) ) \,.
\end{multline}
Therefore, we find that
\begin{multline} \label{eq:embedding_curvature}
    D^2 \tilde{\psi}_\alpha[\iota(\bm{\Phi}_t)](D^2 \iota[\bm{\Phi}_t](\dot{\bm{\Phi}}_t, \dot{\bm{\Phi}}_t), D \iota[\bm{\Phi}_t]\mathbf{U}) \\
    =
    \alpha \int_\Omega 
    \bigg(
    \mathrm{tr}( \nabla_{\mathbf{X}} \bm{\Phi}_t^{-1} \nabla_{\mathbf{X}} \dot{\bm{\Phi}}_t )^2  
    - \mathrm{tr}( (\nabla_{\mathbf{X}} \bm{\Phi}_t^{-1} \nabla_{\mathbf{X}} \dot{\bm{\Phi}}_t)^2 ) \bigg)
    \mathrm{tr} (\nabla_{\mathbf{X}} \bm{\Phi}_t^{-1} \nabla_{\mathbf{X}} \mathbf{U})
    \mathsf{d}^d \mathbf{X}_0 \,,
\end{multline}
and
\begin{equation} \label{eq:potential_third_deriv}
    D^3 \tilde{\psi}_\alpha[\iota(\bm{\Phi}_t)](D \iota[\bm{\Phi}_t] \dot{\bm{\Phi}}_t, D \iota[\bm{\Phi}_t] \dot{\bm{\Phi}}_t, D \iota[\bm{\Phi}_t]\mathbf{U})
    =
    -
    2 \alpha
    \int_\Omega
    \left(
    \mathrm{tr}( \nabla_{\mathbf{X}} \bm{\Phi}_t^{-1} \nabla_{\mathbf{X}} \dot{\bm{\Phi}}_t )^2
    \mathrm{tr}( \nabla_{\mathbf{X}} \bm{\Phi}_t^{-1} \nabla_{\mathbf{X}} \mathbf{U})
    \right)
    \mathsf{d}^d \mathbf{X}_0 \,.
\end{equation}
This allows us to express the abstract form of the Euler--Lagrange equations of the HRE model in Lagrangian coordinates:
\begin{multline} \label{eq:hre_geodesic_eqn}
    \int_\Omega 
    \left(
    \ddot{\bm{\Phi}}_t \cdot \mathbf{U}
    +
    \alpha 
    \mathrm{tr} (\nabla_{\mathbf{X}} \bm{\Phi}_t^{-1} \nabla_{\mathbf{X}} \ddot{\bm{\Phi}}_t)
    \mathrm{tr} (\nabla_{\mathbf{X}} \bm{\Phi}_t^{-1} \nabla_{\mathbf{X}} \mathbf{U})
    \right)
    \mathsf{d}^d \mathbf{X}_0 \\
    =
    - \alpha \int_\Omega 
    \bigg(
    \mathrm{tr}( \nabla_{\mathbf{X}} \bm{\Phi}_t^{-1} \nabla_{\mathbf{X}} \dot{\bm{\Phi}}_t )^2  
    - \mathrm{tr}( (\nabla_{\mathbf{X}} \bm{\Phi}_t^{-1} \nabla_{\mathbf{X}} \dot{\bm{\Phi}}_t)^2 ) \bigg)
    \mathrm{tr} (\nabla_{\mathbf{X}} \bm{\Phi}_t^{-1} \nabla_{\mathbf{X}} \mathbf{U})
    \mathsf{d}^d \mathbf{X}_0 \\
    +
    \alpha
    \int_\Omega
    \left(
    \mathrm{tr}( \nabla_{\mathbf{X}} \bm{\Phi}_t^{-1} \nabla_{\mathbf{X}} \dot{\bm{\Phi}}_t )^2
    \mathrm{tr}( \nabla_{\mathbf{X}} \bm{\Phi}_t^{-1} \nabla_{\mathbf{X}} \mathbf{U})
    \right)
    \mathsf{d}^d \mathbf{X}_0 \\
    =
    \alpha \int_\Omega 
    \mathrm{tr}( (\nabla_{\mathbf{X}} \bm{\Phi}_t^{-1} \nabla_{\mathbf{X}} \dot{\bm{\Phi}}_t)^2 ) 
    \mathrm{tr} (\nabla_{\mathbf{X}} \bm{\Phi}_t^{-1} \nabla_{\mathbf{X}} \mathbf{U})
    \mathsf{d}^d \mathbf{X}_0 \,,
\end{multline}
for all $\mathbf{U} \in T_{\bm{\Phi}_t} \mathrm{Diff}(\Omega)$. For comparison, the IGR equation may be written as
\begin{multline}
    \int_\Omega 
    \left(
    \ddot{\bm{\Phi}}_t \cdot \mathbf{U}
    +
    \alpha 
    \mathrm{tr} (\nabla_{\mathbf{X}} \bm{\Phi}_t^{-1} \nabla_{\mathbf{X}} \ddot{\bm{\Phi}}_t)
    \mathrm{tr} (\nabla_{\mathbf{X}} \bm{\Phi}_t^{-1} \nabla_{\mathbf{X}} \mathbf{U})
    \right)
    \mathsf{d}^d \mathbf{X}_0 \\
    =
    - \alpha \int_\Omega 
    \bigg(
    \mathrm{tr}( \nabla_{\mathbf{X}} \bm{\Phi}_t^{-1} \nabla_{\mathbf{X}} \dot{\bm{\Phi}}_t )^2  
    - \mathrm{tr}( (\nabla_{\mathbf{X}} \bm{\Phi}_t^{-1} \nabla_{\mathbf{X}} \dot{\bm{\Phi}}_t)^2 ) \bigg)
    \mathrm{tr} (\nabla_{\mathbf{X}} \bm{\Phi}_t^{-1} \nabla_{\mathbf{X}} \mathbf{U})
    \mathsf{d}^d \mathbf{X}_0 \\
    +
    2 \alpha
    \int_\Omega
    \left(
    \mathrm{tr}( \nabla_{\mathbf{X}} \bm{\Phi}_t^{-1} \nabla_{\mathbf{X}} \dot{\bm{\Phi}}_t )^2
    \mathrm{tr}( \nabla_{\mathbf{X}} \bm{\Phi}_t^{-1} \nabla_{\mathbf{X}} \mathbf{U})
    \right)
    \mathsf{d}^d \mathbf{X}_0 \\
    =
    \alpha \int_\Omega 
    \bigg(
    \mathrm{tr}( \nabla_{\mathbf{X}} \bm{\Phi}_t^{-1} \nabla_{\mathbf{X}} \dot{\bm{\Phi}}_t )^2  
    + \mathrm{tr}( (\nabla_{\mathbf{X}} \bm{\Phi}_t^{-1} \nabla_{\mathbf{X}} \dot{\bm{\Phi}}_t)^2 ) \bigg)
    \mathrm{tr} (\nabla_{\mathbf{X}} \bm{\Phi}_t^{-1} \nabla_{\mathbf{X}} \mathbf{U})
    \mathsf{d}^d \mathbf{X}_0
    \,,
\end{multline}
for all $\mathbf{U} \in T_{\bm{\Phi}_t} \mathrm{Diff}(\Omega)$ \cite{cao2023information}. 

\begin{remark}
Integration by parts never occurred in this derivation of these weak-form evolution equations. As written, this equation is valid for any boundary conditions and on any subdomain, $\mathcal{D} \subseteq \Omega$. This fact will be leveraged to derive a local conservation law for kinetic energy in Section~\ref{sec:pressureless-IGR}.
\end{remark}

\subsection{Deriving the Eulerian pressureless IGR and HRE models}\label{sec:eulerianization}

Finally, we translate to the Eulerian reference frame by letting $\mathbf{u}(\bm{\Phi}(\mathbf{X})) = \dot{\bm{\Phi}}(\mathbf{X})$ and $\rho(\mathbf{X}) = \det(\nabla_{\mathbf{X}} \bm{\Phi})^{-1}$. Because $\ddot{\bm{\Phi}} = \mathsf{D} \mathbf{u}/\mathsf{D} t = \partial_t \mathbf{u} + (\mathbf{u} \cdot \nabla_{\mathbf{x}}) \mathbf{u}$ in the Eulerian reference frame, we find that
\begin{equation}
    \mathcal{L}_\alpha(\rho) \frac{\mathsf{D} \mathbf{u}}{\mathsf{D} t} 
    \coloneq
    \rho \frac{\mathsf{D} \mathbf{u}}{\mathsf{D} t} - \alpha \nabla_{\mathbf{x}} \left( \rho \nabla_{\mathbf{x}} \cdot \frac{\mathsf{D} \mathbf{u}}{\mathsf{D} t} \right) 
    =
    \bm{F} \,,
\end{equation}
where, for pressureless HRE, 
\begin{equation}
    \bm{F} = - \alpha \nabla_{\mathbf{x}} \cdot (\rho \mathrm{tr}(\nabla_{\mathbf{x}} \mathbf{u}^2)) \,,
\end{equation}
and, for pressureless IGR, 
\begin{equation}
    \bm{F} = - \alpha \nabla_{\mathbf{x}} \cdot ( \rho(\mathrm{tr}^2(\nabla_{\mathbf{x}} \mathbf{u}) + \mathrm{tr}(\nabla_{\mathbf{x}} \mathbf{u}^2)))
\end{equation}
More generally, when we add other forces, such as pressure in Section~\ref{sec:dispersion-free-extension}, we find that $\bm{F} = -\nabla_{\mathbf{x}} \cdot \mathbb{T}$, for some matrix $\mathbb{T}$. In this more general case, using the commutation relations in Appendix~\ref{appendix:commutation_relations}, we find that
\begin{equation}
    \rho \frac{\mathsf{D} \mathbf{u}}{\mathsf{D} t}
    =
    - \nabla_{\mathbf{x}} \cdot \mathbb{\Sigma} \,,
    \quad \text{where} \quad
    \rho^{-1} \mathbb{\Sigma} - \alpha \nabla_{\mathbf{x}} \cdot ( \rho^{-1} \nabla_{\mathbf{x}} \cdot \mathbb{\Sigma}) \mathbb{I}
    =
    \mathbb{T} \,.
\end{equation}
See \cite{cao2023information} for more details on converting the IGR model from the Lagrangian to Eulerian reference frame.

Hence, we find that the pressureless IGR model becomes
\begin{equation}\label{eq:pressureless-IGR}
\begin{aligned}
    &\partial_t 
    \begin{bmatrix}
        \rho \mathbf{u} \\
        \rho
    \end{bmatrix}
    +
    \nabla_{\mathbf{x}} \cdot
    \begin{bmatrix}
        \rho \mathbf{u} \otimes \mathbf{u}
        +
        \Sigma \mathbb{I} \\
        \rho \mathbf{u}
    \end{bmatrix}
    = 
    \begin{bmatrix}
        0 \\
        0
    \end{bmatrix} \\
    &
    \rho^{-1} \Sigma - \alpha \nabla_{\mathbf{x}} \cdot( \rho^{-1} \nabla_{\mathbf{x}} \Sigma)
    =
    \alpha(\mathrm{tr}^2(\nabla_{\mathbf{x}} \mathbf{u}) + \mathrm{tr}( \nabla_{\mathbf{x}} \mathbf{u}^2)) \,,
\end{aligned}
\end{equation}
whereas the HRE flow becomes
\begin{equation}\label{eq:pressureless-HRE}
\begin{aligned}
    &\partial_t 
    \begin{bmatrix}
        \rho \mathbf{u} \\
        \rho
    \end{bmatrix}
    +
    \nabla_{\mathbf{x}} \cdot
    \begin{bmatrix}
        \rho \mathbf{u} \otimes \mathbf{u}
        +
        \Sigma \mathbb{I} \\
        \rho \mathbf{u}
    \end{bmatrix}
    = 
    \begin{bmatrix}
        0 \\
        0
    \end{bmatrix} \\
    &
    \rho^{-1} \Sigma - \alpha \nabla_{\mathbf{x}} \cdot( \rho^{-1} \nabla_{\mathbf{x}} \Sigma)
    =
    \alpha \mathrm{tr}( \nabla_{\mathbf{x}} \mathbf{u}^2) \,.
\end{aligned}
\end{equation}
The factor of ``1/2'' in the Levi-Civita connection causes the $\mathrm{tr}^2(\nabla_{\mathbf{x}} \mathbf{u})$ contribution to the entropic pressure to exactly cancel in the HRE model. Hence, we split the entropic pressure into conservative and dissipative parts: $\Sigma = \Sigma_C + \Sigma_D$, where
\begin{equation}
    \rho^{-1} \Sigma_C - \alpha \nabla_{\mathbf{x}} \cdot( \rho^{-1} \nabla_{\mathbf{x}} \Sigma_C)
    =
    \alpha \mathrm{tr}( \nabla_{\mathbf{x}} \mathbf{u}^2) \,,
    \quad \text{and} \quad
    \rho^{-1} \Sigma_D - \alpha \nabla_{\mathbf{x}} \cdot( \rho^{-1} \nabla_{\mathbf{x}} \Sigma_D)
    =
    \alpha \mathrm{tr}^2( \nabla_{\mathbf{x}} \mathbf{u}) \,.
\end{equation}
In Section~\ref{sec:full-IGR-dissipative-extension}, we will show that the dissipative component of the entropic pressure may be obtained from a sign-indefinite version of the metriplectic formalism. 

\begin{remark}
    We derived the pressureless HRE system in conservative form in the coordinates $(\rho \mathbf{u}, \rho)$ by directly finding the Euler--Lagrange equations in the Lagrangian reference frame. One could alternatively use the Lie--Poisson formalism to directly derive the model in Eulerian coordinates. The Hamiltonian and Poisson bracket are given by
    \begin{equation}
        H[\mathbf{m}, \rho]
        =
        \frac{1}{2} \int_\Omega \mathbf{m} \cdot \mathbf{u}[\mathbf{m}, \rho] \mathsf{d}^d \mathbf{x} \,,
    \end{equation}
    and
    \begin{equation}
    \begin{aligned}
        \{F, G\}
        =
        &-\int_\Omega \mathbf{m} \left[ \left(\frac{\delta F}{\delta \mathbf{m}} \cdot \nabla_{\mathbf{x}} \right) \frac{\delta G}{\delta \mathbf{m}} - \left(\frac{\delta G}{\delta \mathbf{m}} \cdot \nabla_{\mathbf{x}} \right) \frac{\delta F}{\delta \mathbf{m}} \right] \mathsf{d}^d \mathbf{x} \\
        &-\int_\Omega \rho \left[ \frac{\delta F}{\delta \mathbf{m}} \cdot \nabla_{\mathbf{x}} \frac{\delta G}{\delta \rho} - \frac{\delta G}{\delta \mathbf{m}} \cdot \nabla_{\mathbf{x}} \frac{\delta F}{\delta \rho} \right] \mathsf{d}^d \mathbf{x} \,,
    \end{aligned}
    \end{equation}
    respectively, where $\mathbf{m} = \rho \mathbf{u} - \alpha \nabla_{\mathbf{x}}( \rho \nabla_{\mathbf{x}} \cdot \mathbf{u})$. The implied evolution equations are
    \begin{equation}
    \begin{aligned}
        \partial_t \rho &= - \nabla_{\mathbf{x}} \cdot(\rho \mathbf{u}) \,, \\
        \partial_t \mathbf{m} &= - \nabla_{\mathbf{x}} \cdot \left( \mathbf{m} \otimes \mathbf{u} + \alpha \rho ( \nabla_{\mathbf{x}} \cdot \mathbf{u}) (\nabla_{\mathbf{x}} \mathbf{u})^T \right) \,.
    \end{aligned}
    \end{equation}
    See \Cref{appendix:change_of_variables} for the full details. With suitable boundary conditions and regularity assumptions, these are equivalent to Equation~\eqref{eq:pressureless-HRE}.
\end{remark}

%================================================================
\section{Kinetic energy transport in the pressureless IGR and HRE systems}\label{sec:pressureless-IGR}
The natural kinetic energy of the HRE and IGR models are given by the Hessian metric:
\begin{equation}
    K[\bm{\Phi}_t, \dot{\bm{\Phi}}_t]
    =
    \frac{1}{2}
    (\iota^* D^2 \tilde{\psi}_\alpha [\bm{\Phi}_t])(\dot{\bm{\Phi}}_t, \dot{\bm{\Phi}}_t) \,.
\end{equation}
Its evolution is obtained using the chain rule:
\begin{equation} \label{eq:kinetic_energy_law}
\begin{aligned}
    \dot{K}
    &=
    \frac{1}{2}
    \frac{\mathsf{d}}{\mathsf{d} t} \left( D^2 \tilde{\psi}_\alpha[\iota(\bm{\Phi}_t)](D \iota[\bm{\Phi}_t] \dot{\bm{\Phi}}_t, D \iota[\bm{\Phi}_t] \dot{\bm{\Phi}}_t) \right) \\
    &=
    \frac{1}{2} D^3 \tilde{\psi}_\alpha[\iota(\bm{\Phi}_t)](D \iota[\bm{\Phi}_t] \dot{\bm{\Phi}}_t, D \iota[\bm{\Phi}_t] \dot{\bm{\Phi}}_t, D \iota[\bm{\Phi}_t] \dot{\bm{\Phi}}_t) + 
    D^2 \tilde{\psi}_\alpha[\iota(\bm{\Phi}_t)](D \iota[\bm{\Phi}_t] \ddot{\bm{\Phi}}_t, D \iota[\bm{\Phi}_t] \dot{\bm{\Phi}}_t) \\
    &\hspace{5em} + D^2 \tilde{\psi}_\alpha[\iota(\bm{\Phi}_t)](D^2 \iota[\bm{\Phi}_t](\dot{\bm{\Phi}}_t, \dot{\bm{\Phi}}_t), D \iota[\bm{\Phi}_t] \dot{\bm{\Phi}}_t) 
    \,.
\end{aligned}
\end{equation}
Equations \eqref{eq:embedding_curvature} and \eqref{eq:potential_third_deriv} imply that the energy evolution in the Eulerian reference frame is given by
\begin{equation}
    \dot{K}
    =
    \int_{\Omega}
    \left( 
    \rho \frac{D \mathbf{u}}{Dt} \cdot \mathbf{u}
    +
    \alpha \rho (\nabla_{\mathbf{x}} \cdot \mathbf{u}) \nabla_{\mathbf{x}} \cdot \left( \frac{D \mathbf{u}}{Dt} \right)
    -
    \alpha \rho (\nabla_{\mathbf{x}} \cdot \mathbf{u}) \mathrm{tr}((\nabla_{\mathbf{x}} \mathbf{u})^2)
    \right)
    \mathsf{d}^d \mathbf{x}\,.
\end{equation}
Notice that $\dot{K} = 0$ for the HRE equation. On the other hand, for the IGR flow, we find that
\begin{equation}
    \dot{K}
    =
    - \frac{1}{2} D^3 \tilde{\psi}_\alpha[\iota(\bm{\Phi}_t)](D \iota[\bm{\Phi}_t] \dot{\bm{\Phi}}_t, D \iota[\bm{\Phi}_t] \dot{\bm{\Phi}}_t, D \iota[\bm{\Phi}_t] \dot{\bm{\Phi}}_t) \,.
\end{equation}
Equation~\eqref{eq:potential_third_deriv} implies that, in Eulerian coordinates, $\dot{K} = \alpha \int_\Omega \rho ( \nabla_{\mathbf{x}} \cdot \mathbf{u})^3 \mathsf{d}^d \mathbf{x}$. 

If we consider a similar expression where the kinetic energy density is integrated over an arbitrary subdomain, equations of the same form are obtained. This provides a key abstraction in deriving a local kinetic energy transport equation. That is, we let
\begin{equation}
    \tilde{\psi}_{\alpha, \mathcal{D}}[\bm{\Phi}, \Phi']
    =
    \int_{\mathcal{D}} \left( \frac{1}{2} | \bm{\Phi}(\mathbf{X}_0) |^2 -
    \alpha \log(\Phi') \right) \mathsf{d}^d \mathbf{X}_0 \,,
\end{equation}
for arbitrary $\mathcal{D} \subseteq \Omega$, and we define
\begin{equation}
    K_{\mathcal{D}}[\bm{\Phi}_t, \dot{\bm{\Phi}}_t]
    =
    \frac{1}{2}
    (\iota^* D^2 \tilde{\psi}_{\alpha,\mathcal{D}} [\bm{\Phi}_t])(\dot{\bm{\Phi}}_t, \dot{\bm{\Phi}}_t) \,.
\end{equation}
The time derivative of $K_{\mathcal{D}}$ formally the same as \eqref{eq:kinetic_energy_law} with $\tilde{\psi}_\alpha$ replaced by $\tilde{\psi}_{\alpha,\mathcal{D}}$.

\subsection{Local energy transport in pressureless HRE} \label{subsec:local-energy-transport-pressureless-HRE}

In Eulerian coordinates, the momentum equation of the pressureless HRE system can alternatively be written as
\begin{equation}
    \rho \frac{D \mathbf{u}}{Dt}
    =
    - \nabla_{\mathbf{x}} \Sigma \,,
\end{equation}
where $D \mathbf{u}/Dt = \partial_t \mathbf{u} + (\mathbf{u} \cdot \nabla_{\mathbf{x}}) \mathbf{u}$ is the material derivative. The kinetic energy density in the Eulerian reference frame is given by
\begin{equation}
    K_E
    =
    \frac{1}{2}
    \rho \left( |\mathbf{u}|^2 + \alpha (\nabla_{\mathbf{x}} \cdot \mathbf{u})^2 \right) \,.
\end{equation}
As mentioned at the beginning of this section, the evolution of kinetic energy over an arbitrary subdomain takes the same form as in Equation~\eqref{eq:kinetic_energy_law}. That is, we find that
\begin{equation}
    \dot{K}_{\mathcal{D}_t}
    =
    \frac{\mathsf{d}}{\mathsf{d}t}
    \int_{\mathcal{D}_t} K_E \mathsf{d}^d \mathbf{x}
    =
    \int_{\mathcal{D}_t}
    \left( 
    \rho \frac{D \mathbf{u}}{Dt} \cdot \mathbf{u}
    +
    \alpha \rho (\nabla_{\mathbf{x}} \cdot \mathbf{u}) \nabla_{\mathbf{x}} \cdot \left( \frac{D \mathbf{u}}{Dt} \right)
    -
    \alpha \rho (\nabla_{\mathbf{x}} \cdot \mathbf{u}) \mathrm{tr}((\nabla_{\mathbf{x}} \mathbf{u})^2)
    \right)
    \mathsf{d}^d \mathbf{x}\,,
\end{equation}
where $\mathcal{D}_t \subseteq \Omega$ is an arbitrary subdomain moving with the flow. Hence, Reynolds' transport theorem implies that
\begin{equation}
    \partial_t K_E
    +
    \nabla_{\mathbf{x}} \cdot(K_E \mathbf{u})
    =
    \rho \frac{D \mathbf{u}}{Dt} \cdot \mathbf{u}
    +
    \alpha \rho (\nabla_{\mathbf{x}} \cdot \mathbf{u}) \nabla_{\mathbf{x}} \cdot \left( \frac{D \mathbf{u}}{Dt} \right)
    -
    \alpha \rho (\nabla_{\mathbf{x}} \cdot \mathbf{u}) \mathrm{tr}((\nabla_{\mathbf{x}} \mathbf{u})^2) \,.
\end{equation}
The remaining task is to express the right-hand side in divergence form. Notice,
\begin{multline}
    \rho \frac{D \mathbf{u}}{Dt} \cdot \mathbf{u}
    +
    \alpha \rho (\nabla_{\mathbf{x}} \cdot \mathbf{u}) \nabla_{\mathbf{x}} \cdot \left( \frac{D \mathbf{u}}{Dt} \right)
    -
    \alpha (\nabla_{\mathbf{x}} \cdot \mathbf{u}) \mathrm{tr}((\nabla_{\mathbf{x}} \mathbf{u})^2) \\
    =
    - \nabla_{\mathbf{x}} \Sigma \cdot \mathbf{u}
    -
    \alpha \rho (\nabla_{\mathbf{x}} \cdot \mathbf{u}) \nabla_{\mathbf{x}} \cdot \left( \frac{\nabla_{\mathbf{x}} \Sigma}{\rho} \right)
    -
    \alpha \rho (\nabla_{\mathbf{x}} \cdot \mathbf{u}) \mathrm{tr}((\nabla_{\mathbf{x}} \mathbf{u})^2) \\
    =
    - \nabla_{\mathbf{x}} \Sigma \cdot \mathbf{u}
    -
    \alpha \rho (\nabla_{\mathbf{x}} \cdot \mathbf{u}) \nabla_{\mathbf{x}} \cdot \left( \frac{\nabla_{\mathbf{x}} \Sigma}{\rho} \right)
    -
    \rho (\nabla_{\mathbf{x}} \cdot \mathbf{u}) \left( \rho^{-1} \Sigma - \alpha \nabla_{\mathbf{x}} \cdot ( \rho^{-1} \nabla_{\mathbf{x}} \Sigma) \right) \\
    =
    - \nabla_{\mathbf{x}} \Sigma \cdot \mathbf{u} - (\nabla_{\mathbf{x}} \cdot \mathbf{u}) \Sigma
    =
    - \nabla_{\mathbf{x}} \cdot( \Sigma \mathbf{u}) 
    \,.
\end{multline}
Hence, for the pressureless HRE model, one finds that
\begin{equation} \label{eq:igr_kinetic_nrg_evo}
    \partial_t K_E + \nabla_{\mathbf{x}} \cdot( (K_E + \Sigma)\mathbf{u}) = 0 \,.
\end{equation}
Therefore, the energy flux due to the conservative component of entropic pressure takes an intuitive form, which is formally identical to the energy flux due to thermodynamic pressure. 

\subsection{Local energy transport in pressureless IGR}\label{subsec:local-energy-transport-pressureless-IGR}

In pressureless IGR, the derivation of the local kinetic energy transport is identical to the pressureless HRE case, except the right-hand side includes the additional non-conservative contribution from the entropic pressure, see ~\Cref{eq:pressureless-IGR}: 
\begin{multline}
    \rho \frac{D \mathbf{u}}{Dt} \cdot \mathbf{u}
    +
    \alpha \rho (\nabla_{\mathbf{x}} \cdot \mathbf{u}) \nabla_{\mathbf{x}} \cdot \left( \frac{D \mathbf{u}}{Dt} \right)
    -
    \alpha (\nabla_{\mathbf{x}} \cdot \mathbf{u}) \mathrm{tr}((\nabla_{\mathbf{x}} \mathbf{u})^2) \\
    =
    - \nabla_{\mathbf{x}} \Sigma \cdot \mathbf{u}
    -
    \alpha \rho (\nabla_{\mathbf{x}} \cdot \mathbf{u}) \nabla_{\mathbf{x}} \cdot \left( \frac{\nabla_{\mathbf{x}} \Sigma}{\rho} \right)
    -
    \alpha \rho (\nabla_{\mathbf{x}} \cdot \mathbf{u}) \mathrm{tr}((\nabla_{\mathbf{x}} \mathbf{u})^2) \\
    =
    - \nabla_{\mathbf{x}} \Sigma \cdot \mathbf{u}
    -
    \alpha \rho (\nabla_{\mathbf{x}} \cdot \mathbf{u}) \nabla_{\mathbf{x}} \cdot \left( \frac{\nabla_{\mathbf{x}} \Sigma}{\rho} \right)
    -
    \rho (\nabla_{\mathbf{x}} \cdot \mathbf{u}) \left( \rho^{-1} \Sigma - \alpha \nabla_{\mathbf{x}} \cdot ( \rho^{-1} \nabla_{\mathbf{x}} \Sigma) - \alpha (\nabla_{\mathbf{x}} \cdot \mathbf{u})^2 \right) \\
    =
    - \nabla_{\mathbf{x}} \Sigma \cdot \mathbf{u} - (\nabla_{\mathbf{x}} \cdot \mathbf{u}) \Sigma + \alpha \rho (\nabla_{\mathbf{x}} \cdot \mathbf{u})^3
    =
    - \nabla_{\mathbf{x}} \cdot( \Sigma \mathbf{u}) + \alpha \rho (\nabla_{\mathbf{x}} \cdot \mathbf{u})^3
    \,.
\end{multline}
Hence, we obtain the following energy transport equation:
\begin{equation} \label{eq:pressureless_IGR_kinetic_energy_transport}
    \partial_t K_E + \nabla_{\mathbf{x}} \cdot( (K_E + \Sigma)\mathbf{u}) = \alpha \rho (\nabla_{\mathbf{x}} \cdot \mathbf{u})^3 \,,
\end{equation}
where the elliptic equation for $\Sigma$ is given in ~\Cref{eq:pressureless-IGR}. This recovers the previously observed fact that total kinetic energy is not conserved in the pressureless IGR system. 

\begin{remark}
    We will find in Section~\ref{sec:full-IGR-dissipative-extension} that this non-conservative term on the right-hand side of the energy evolution law may be eliminated by introducing thermodynamics. By asserting that the non-conserved kinetic energy is converted into internal energy (through entropy generation), we may restore total energy conservation. 
\end{remark}

\subsection{The non-conservative term in IGR conserves acoustic waves}\label{subsec:acoustic-waves}
To better understand the non-conservative term in the IGR model, we now consider how the non-conservative term in the local kinetic energy transport equation impacts acoustic waves in the full IGR system. To build intuition, consider a more general kinetic energy evolution law of the form 
\begin{equation}
    \dot{K}_n
    =
    \alpha (-1)^{n+1} \int_\Omega \rho (\nabla_{\mathbf{x}} \cdot \mathbf{u})^n \mathsf{d}^d \mathbf{x}\,.
\end{equation}
This corresponds to:
\begin{itemize}
    \item a bulk (or volume) viscosity when $n=2$, and bulk hyper-viscosity when $n = 2m$, $m > 1$;
    \item the fluctuating energy in IGR when $n=3$, and similar behavior when $n=2m+1$, $m > 1$.
\end{itemize}
For the sake of argument, suppose that we have added thermodynamic pressure to the HRE and IGR models, so that the medium supports acoustic waves. We will discuss the addition of thermodynamic pressure to the HRE model in detail in section \Cref{sec:dispersion-free-extension}. Consider linear acoustic waves (i.e.,\! small amplitude plane waves on an infinite domain):
\begin{equation}
    \mathbf{u} = \mathbf{u}' \cos(\mathbf{k} \cdot \mathbf{x}- \omega t) \,,
    \quad \text{and} \quad
    \rho = \rho_0 + \rho' \cos(\mathbf{k} \cdot \mathbf{x}- \omega t) \,,
\end{equation}
where
\begin{equation}
    \omega = c_s | \mathbf{\mathbf{k}} | \,,
    \quad \text{and} \quad
    \mathbf{u}'
    =
    \frac{c_s \rho'}{\rho_0} \frac{\mathbf{k}}{|\mathbf{k}|} \,,
\end{equation}
where $c_s$ is the sound speed. From equation \ref{eq:igr_kinetic_nrg_evo}, we see that the kinetic energy of acoustic waves, averaged over the wave period, evolves like
\begin{equation} 
    \langle \dot{K}_n \rangle
    =
    (-1)^{n+1} \left( \frac{\omega \rho'}{\rho_0} \right)^n
    \int_0^{2\pi}
    (\rho_0 + \rho' \cos(\theta)) \sin^n(\theta) \mathsf{d} \theta 
    =
    \begin{cases}
        0 \,, & n \text{ odd} \,, \\
        (-1)^{n+1} \dfrac{\omega^n \rho'^n }{\rho_0^{n-1}} \dfrac{2 \pi (n-1)!!}{n!!} \,, & n \text{ even} \,. 
    \end{cases}
\end{equation}
Despite the presence of a non-conservative force in the pressureless IGR model, the energy of acoustic waves is nonetheless conserved on average. This is the utility of the sign-indefinite character of the non-conservative heat source/sink term in IGR. Wave steepening is mitigated because kinetic energy is dissipated in compressive flow regimes, while acoustic wave energy is conserved on average. On the other hand, a positive semi-definite non-conservative force, like bulk viscosity, monotonically depletes the energy of acoustic waves. 

%================================================================
\section{A dispersion-free, thermodynamic extension of the HRE model}\label{sec:dispersion-free-extension}

Now, we wish to extend the pressureless HRE model to include thermodynamic pressure. If we directly extend the pressureless HRE formalism by adding a potential energy to the Hamiltonian, we encounter a serious issue. All forces arising from a potential in Newton's law on a curved manifold have their gradient computed with respect to the metric. More concretely, if we attempt to extend the Hamiltonian formalism to include barotropic pressure by defining the Hamiltonian
\begin{equation}
    H[\mathbf{u},\rho]
    =
    \int_\Omega \left[ \frac{1}{2} \rho \left( |\mathbf{u}|^2 + \alpha (\nabla_{\mathbf{x}} \cdot \mathbf{u})^2 \right) + \rho \varepsilon(\rho) \right] \mathsf{d}^d \mathbf{x} \,,
\end{equation}
we find that the pressure appearing in the momentum equation is not $\nabla_{\mathbf{x}} p$ where $p(\rho) = \rho^2 \varepsilon'(\rho)$, but rather
\begin{equation}
    \rho^{-1} p - \alpha \nabla_{\mathbf{x}} \cdot( \rho^{-1} \nabla_{\mathbf{x}} p)
    =
    \rho \varepsilon'(\rho) \,,
\end{equation}
the exact same elliptic operator defining the entropic pressure. Although this model formally recovers the desired pressure as $\alpha \to 0$, it induces a spurious dispersion relation in acoustic waves, leading to unphysical behavior. In the unidimensional, barotropic setting, it has been found that adding density gradient dependence to the potential energy can correct the spurious dispersion relation \cite{guelmame2022hamiltonian}. This yields a model that is similar to a Korteweg-type fluid, with the density dependent potential resembling a capillary energy. However, the capillary energy is added purely to correct the spurious dispersion relation, not to add surface tension at fluid interfaces. The Hamiltonian structure of Korteweg-type fluids is well understood \cite{suzuki2020generic}, and we leverage this prior work in our analysis. In the following discussion, we obtain a multidimensional generalization of \cite{guelmame2022hamiltonian} with advected entropy.

Our proposed Hamiltonian regularized Euler (HRE) system is most easily stated using the Lie--Poisson Hamiltonian formalism \cite{morrison1998hamiltonian}. See \Cref{appendix:change_of_variables} for an account of how this Hamiltonian formalism in noncanonical coordinates is obtained from the canonical Hamiltonian formalism in Lagrangian coordinates. In entropy coordinates, $(\mathbf{m},\rho, \sigma)$, the model is prescribed by the bracket
\begin{equation} \label{eq:lie_poisson_bracket}
\begin{aligned}
    \{F, G\}
    =
    &-\int_\Omega \mathbf{m} \left[ \left(\frac{\delta F}{\delta \mathbf{m}} \cdot \nabla_{\mathbf{x}} \right) \frac{\delta G}{\delta \mathbf{m}} - \left(\frac{\delta G}{\delta \mathbf{m}} \cdot \nabla_{\mathbf{x}} \right) \frac{\delta F}{\delta \mathbf{m}} \right] \mathsf{d}^d \mathbf{x} \\
    & \quad -\int_\Omega \rho \left[ \frac{\delta F}{\delta \mathbf{m}} \cdot \nabla_{\mathbf{x}} \frac{\delta G}{\delta \rho} - \frac{\delta G}{\delta \mathbf{m}} \cdot \nabla_{\mathbf{x}} \frac{\delta F}{\delta \rho} \right] \mathsf{d}^d \mathbf{x} \\
    &\quad -\int_\Omega \sigma \left[ \frac{\delta F}{\delta \mathbf{m}} \cdot \nabla_{\mathbf{x}} \frac{\delta G}{\delta \sigma} - \frac{\delta G}{\delta \mathbf{m}} \cdot \nabla_{\mathbf{x}} \frac{\delta F}{\delta \sigma} \right] \mathsf{d}^d \mathbf{x} \,,
\end{aligned}
\end{equation}
and the Hamiltonian
\begin{equation}
    H[\mathbf{m},\rho, \sigma]
    =
    \int_\Omega 
    \left(
    \frac{1}{2}
    \mathbf{m} \cdot \mathbf{u}[\mathbf{m},\rho]
    +
    \rho
    \varepsilon \left(\rho, \frac{\sigma}{\rho}, \nabla_{\mathbf{x}} \rho \right)
    \right)
    \mathsf{d}^d \mathbf{x} \,,
\end{equation}
where the fluid velocity $\mathbf{u}$ is implicitly prescribed by an elliptic equation: 
\begin{equation}
    \mathbf{m} = \mathcal{L}_\alpha(\rho) \mathbf{u} = \rho \mathbf{u} - \alpha \nabla_{\mathbf{x}} (\rho \nabla_{\mathbf{x}} \cdot \mathbf{u}) \,.    
\end{equation}
To preserve Galilean invariance, we require that $\varepsilon$ depend on $\nabla_{\mathbf{x}} \rho$ only through $|\nabla_{\mathbf{x}} \rho|^2$, which, by the chain rule, yields $\varepsilon_{\nabla_{\mathbf{x}} \rho} = 2 \varepsilon_{|\nabla_{\mathbf{x}} \rho|^2} \nabla_{\mathbf{x}} \rho$. The evolution equations of this system are
\begin{equation} \label{eq:hre_entropy_evolution_equations}
\begin{aligned}
    \partial_t \rho &= - \nabla_{\mathbf{x}} \cdot(\rho \mathbf{u}) \,, \\
    \partial_t \mathbf{m} &= - \nabla_{\mathbf{x}} \cdot \left( \mathbf{m} \otimes \mathbf{u} + \alpha \rho ( \nabla_{\mathbf{x}} \cdot \mathbf{u}) (\nabla_{\mathbf{x}} \mathbf{u})^T + (p - \rho (\nabla_{\mathbf{x}} \cdot \bm{\xi})) \mathbb{I} + \nabla_{\mathbf{x}} \rho \otimes \bm{\xi} \right) \\
    \partial_t \sigma &= - \nabla_{\mathbf{x}} \cdot \left( \sigma \mathbf{u} \right) \,,
\end{aligned}
\end{equation}
where $p = \rho^2 \varepsilon_\rho$, and $\bm{\xi} = \rho \varepsilon_{\nabla_{\mathbf{x}} \rho}$. In the following sections, we analyze the linear waves arising as infinitesimal perturbations from equilibrium. The Lie--Poisson form in the coordinates $(\mathbf{m}, \rho, \sigma)$ is convenient, as it facilitates the use of Hamiltonian perturbation theory, see appendix \ref{appendix:energy_casimir}. 

\begin{remark}
    The reader should be cautious here when interpreting the notation used in this abstract form of the model. For a Hamiltonian of the form
    \begin{equation}
        H[\mathbf{m},\rho, \sigma]
        =
        \int_\Omega 
        \left(
        \frac{1}{2}
        \mathbf{m} \cdot \mathbf{u}[\mathbf{m},\rho]
        +
        \rho
        \left(
        \varepsilon \left(\rho, \frac{\sigma}{\rho}\right)
        +
        \frac{1}{2} \kappa \left(\rho, \frac{\sigma}{\rho}\right) |\nabla_{\mathbf{x}} \rho |^2
        \right)
        \right)
        \mathsf{d}^d \mathbf{x} \,,
    \end{equation}
    which has separated the capillary energy from the rest of the potential energy, we find that
    \begin{equation}
        p
        =
        \rho^2 \left( \varepsilon_\rho + \frac{\kappa_\rho}{2} | \nabla_\bx \rho |^2 \right) \,,
        \quad \text{and} \quad
        \bm{\xi}
        =
        \rho \kappa \nabla_\bx \rho \,.
    \end{equation}
    That is, if the coefficient for capillary energy depends on density, then the pressure receives an additional contribution from the capillary force. 
\end{remark}

\begin{remark}
It may not be immediately obvious that this system generalizes the pressureless HRE model, since we have expressed the model in an alternative system of coordinates. We will show in Section~\ref{sec:HRE-conservative-variables} that this model may be written in locally conservative form in the variables $(\rho \mathbf{u}, \rho, E)$, where $E$ is the total energy. 
\end{remark}

\subsection{Constant equilibria and Galilean invariance}\label{subsec:constant-equilibria}

We now derive a linear dispersion relation for the HRE model. We are interested in equilibria of the form $(\mathbf{u}, \rho, s) = (\mathbf{u}_0, \rho_0, s_0)$, where each of these fields is constant in space and time. This type of equilibrium is singled out, because solutions to the Riemann problem---which connect left and right states across a discontinuity---arise as nonlinear continuations of the linear characteristics of the transport equation \cite{godlewski2013numerical}. Therefore, for wave-fronts to propagate at the correct speed without dispersion, it is necessary that these linear waves have the correct dispersion relation. Since the regularization terms, which depend on velocity and density gradients, are weak away from sharp interfaces, the dominant requirement for physical fidelity is that the linear dispersion relation be matched correctly. This line of thought guides the design of the Hamiltonian regularization, and motivates our interest in linear waves. 

If our model is Galilean invariant, then we may reduce the analysis to static equilibria: $(\mathbf{u}, \rho, s) = (\mathbf{0}, \rho_0, s_0)$. This simplification is not only convenient, but Galilean invariance is fundamentally important for physical realism in fluid models. We only demonstrate invariance with respect to Galilean boosts, as this transformation is relevant to our discussion: invariance with respect to the remaining Galilean symmetries proceed similarly. To begin, let $\mathbf{V} \in \mathbb{R}^d$, and define 
\begin{equation}
    \mathbf{x}' = \mathbf{x} - \mathbf{V} t \,,
    \quad 
    t' = t \,.
\end{equation}
Then the mass and entropy density and velocity fields in the moving frame shifted by $\mathbf{V}$ are given by
\begin{equation}
    \rho'(\mathbf{x}',t')
    =
    \rho(\mathbf{x},t) \,,
    \quad
    \sigma'(\mathbf{x}',t')
    =
    \sigma(\mathbf{x},t) \,,
    \quad \text{and} \quad
    \mathbf{u}'(\mathbf{x}',t')
    =
    \mathbf{u}(\mathbf{x},t) - \mathbf{V} \,.
\end{equation}
If we let $\mathbf{m}'(\mathbf{x}',t') = \mathcal{L}_\alpha(\rho'(\mathbf{x}',t')) \mathbf{u}'(\mathbf{x}',t')$, then we find that
\begin{equation}
    \mathcal{L}_\alpha(\rho)(\mathbf{u}-\mathbf{V})
    =
    \mathcal{L}_\alpha(\rho)\mathbf{u} - \rho \mathbf{V} \,,
\end{equation}
because $\mathbf{V}$ is constant, so that $\mathbf{m}'(\mathbf{x}',t') = \mathbf{m}(\mathbf{x},t) - \rho(\mathbf{x},t) \mathbf{V}$. Hence, a boost acts by
\begin{equation}
    (\rho, \sigma, \mathbf{m})
    \mapsto
    (\rho, \sigma, \mathbf{m} - \rho \mathbf{V}) \,,
    \quad
    (\mathbf{x},t) 
    \mapsto
    (\mathbf{x} - \mathbf{V} t, t) \,.
\end{equation}
Letting $F'[\mathbf{m}',\rho',\sigma'] = F[\mathbf{m}, \rho, \sigma]$ we find
\begin{equation}
    \frac{\delta F'}{\delta \mathbf{m}'}
    =
    \frac{\delta F}{\delta \mathbf{m}} \,, 
    \quad
    \frac{\delta F'}{\delta \rho'}
    =
    \frac{\delta F}{\delta \rho}
    + 
    \mathbf{V} \cdot \frac{\delta F}{\delta \mathbf{m}} \,,
    \quad \text{and} \quad
    \frac{\delta F'}{\delta \sigma'}
    =
    \frac{\delta F}{\delta \sigma} \,.
\end{equation}
Moreover, $\mathsf{d}^d \mathbf{\mathbf{x}}' = \mathsf{d}^d \mathbf{x}$ and $\nabla_{\mathbf{x}'} = \nabla_{\mathbf{x}}$. It follows that the Poisson bracket is invariant under the boost. The kinetic energy transforms as
\begin{equation}
\begin{aligned}
    \frac{1}{2} \int_\Omega \mathbf{u} \cdot &\mathbf{m} \mathsf{d}^d \mathbf{x} 
    = \frac{1}{2} \int_{\Omega'} (\mathbf{u}' + \mathbf{V}) \cdot (\mathbf{m}' + \rho' \mathbf{V}) \mathsf{d}^d \mathbf{x}' \\
    &= \frac{1}{2} \int_{\Omega'} \mathbf{m}' \cdot \mathbf{u}' \mathsf{d}^d \mathbf{x}' 
        + \frac{1}{2} \mathbf{V} \cdot \int_\Omega \rho' \mathbf{u}' \mathsf{d}^d \mathbf{x} 
        + \frac{1}{2} \mathbf{V} \cdot \int_\Omega \mathbf{m}' \mathsf{d}^d \mathbf{x}
        + \frac{1}{2} \left( \int_\Omega \rho' \mathsf{d}^d \mathbf{x} \right) | \mathbf{V} |^2 \,. 
\end{aligned}
\end{equation}
Assuming periodic, decaying, or no-flux boundary conditions, we find that
\begin{equation}
    \int_\Omega \mathbf{m} \mathsf{d}^d \mathbf{x}
    =
    \int_\Omega (\rho \mathbf{u} - \alpha \nabla_{\mathbf{x}}( \rho \nabla_{\mathbf{x}} \cdot \mathbf{u})) \mathsf{d}^d \mathbf{x}
    =
    \int_\Omega \rho \mathbf{u} \mathsf{d}^d \mathbf{x} \,.
\end{equation}
Hence,
\begin{equation}
    H[\mathbf{m},\rho,\sigma]
    =
    H[\mathbf{m}',\rho',\sigma']
    +
    \mathbf{V} \cdot P'[\rho', \mathbf{u}']
    +
    \frac{1}{2} M | \mathbf{V}|^2 \,,
\end{equation}
where
\begin{equation}
    P'
    =
    \int_{\Omega'} \rho' \mathbf{u}' \mathsf{d}^d \mathbf{x}'
    =
    \int_{\Omega'} \mathbf{m}' \mathsf{d}^d \mathbf{x}' \,.
\end{equation}
Total mass is a Casimir invariant, while total momentum is the generator for Galilean boosts. Hence, we find that the evolution equations remain invariant under constant boosts. 

\begin{remark}
    In order to preserve rotational invariance, we require that the potential energy depends on $\nabla_{\mathbf{x}} \rho$ only through $| \nabla_{\mathbf{x}} \rho|^2$. Thus, symmetry considerations also constrain the permissible set of Hamiltonian regularizations of the potential energy. 
\end{remark}

\subsection{The augmented Hamiltonian for the energy Casimir method}\label{subsec:augmented-Hamiltonian-energy-Casimir}

Having established invariance with respect to Galilean boosts, we now consider the linear stability of static equilibria: $(\mathbf{m}, \rho, \sigma) = (\mathbf{0}, \rho_0, \sigma_0)$. To do so, we use the energy--Casimir method, see \Cref{appendix:energy_casimir}. We first verify that this is, in fact, an equilibrium of the HRE system. This is manifestly the case, as the right-hand side of equation \eqref{eq:hre_entropy_evolution_equations} identically vanishes when $(\mathbf{m}, \rho, \sigma) = (\mathbf{0}, \rho_0, \sigma_0)$. Next, note that the derivatives of the Hamiltonian are given by
\begin{equation}
    \frac{\delta H}{\delta \mathbf{m}}
    =
    \mathbf{u} \,,
    \quad
    \frac{\delta H}{\delta \rho}
    =
    - \frac{1}{2} \left( |\mathbf{u}|^2 + \alpha (\nabla_{\mathbf{x}} \cdot \mathbf{u})^2 \right)
    + g - \nabla_{\mathbf{x}} \cdot \bm{\xi}  \,,
    \quad \text{and} \quad
    \frac{\delta H}{\delta \sigma} = \vartheta \,,
\end{equation}
where
\begin{equation}
    p
    =
    \rho^2 \varepsilon_\rho \,,
    \quad
    \vartheta
    =
    \varepsilon_s \,, 
    \quad
    g
    =
    \varepsilon + \frac{p}{\rho} - s \vartheta \,,
    \quad \text{and} \quad
    \bm{\xi}
    =
    \rho \varepsilon_{\nabla_{\mathbf{x}} \rho} \,,
\end{equation}
which are, respectively, the pressure, temperature, specific Gibbs free energy, and an auxiliary vector used to define the capillary stress. The gradient of the Hamiltonian does not identically vanish for static equilibria, so we seek a Casimir invariant, $\mathcal{C}$, such that $D \mathcal{H}(0, \rho_0, \sigma_0) = 0$, where $\mathcal{H} = H + \mathcal{C}$. Generic functionals of the form
\begin{equation}
    \mathcal{C}
    =
    \int_\Omega \rho f(\sigma/\rho) \mathsf{d}^d \mathbf{x} 
\end{equation}
are Casimir invariants of the Lie--Poisson bracket in Equation~\eqref{eq:lie_poisson_bracket}. We previously assumed that the potential energy depends on the gradient of $\rho$ only through $|\nabla_{\mathbf{x}} \rho|^2$. Hence, $\varepsilon_{\nabla_{\mathbf{x}} \rho} \propto \nabla_{\mathbf{x}} \rho = 0$, for constant $\rho$. If we let
\begin{equation}
    \mathcal{C}
    =
    -\int_\Omega \rho \left( g_0 + \vartheta_0 \frac{\sigma}{\rho} \right) \mathsf{d}^d \mathbf{x} \,,
\end{equation}
where $g_0 = g(\rho_0, \sigma_0)$ and $\vartheta_0 = \vartheta(\rho_0, \sigma_0)$, then
\begin{equation}
    \frac{\delta (H + \mathcal{C})}{\delta \rho}(0, \rho_0, \sigma_0)
    =
    \frac{\delta (H + \mathcal{C})}{\delta \sigma}(0, \rho_0, \sigma_0)
    =
    0 \,.
\end{equation}
Hence, if we let
\begin{equation}
    \mathcal{H}
    =
    H
    -
    \int_\Omega \rho \left( g_0 + \vartheta_0 \frac{\sigma}{\rho} \right) \mathsf{d}^d \mathbf{x} \,,
\end{equation}
then $D\mathcal{H}(0, \rho_0, \sigma_0) = 0$. Hence, it follows from appendix \ref{appendix:energy_casimir} that we may analyze the linear stability of static equilibria by studying the spectrum of the operator $J(z_0) D^2 \mathcal{H}(z_0)$, where $z_0 = (\mathbf{0}, \rho_0, \sigma_0)$ and $J(z_0)$ is the operator applying the Poisson bracket.

\subsection{Computing the dispersion relation}\label{subsec:dispersion-relation}

Next, we compute the linear operator $(J D^2 \mathcal{H})(z_0)$, where $z_0 = (\mathbf{0}, \rho_0, \sigma_0)$. First, we find the second variation of the augmented Hamiltonian. The Casimir we used to shift the Hamiltonian is linear in $\rho$ and $\sigma$, so its second variation vanishes. Moreover, we find that
\begin{equation}
    D^2 H(z_0)
    =
    \begin{pmatrix}
        \mathcal{L}_\alpha(\rho_0)^{-1} & 0 & 0 \\
        0 & \mathcal{V}_{\rho \rho} & \mathcal{V}_{\rho \sigma} \\
        0 & \mathcal{V}_{\sigma \rho} & \mathcal{V}_{\sigma \sigma} 
    \end{pmatrix} 
    \implies
    D^2 H(z_0) \delta z
    =
    \begin{pmatrix}
        \mathcal{L}_\alpha(\rho_0)^{-1} \delta \mathbf{m} \\
        \mathcal{V}_{\rho\rho} \delta \rho + \mathcal{V}_{\rho \sigma} \delta \sigma \\
        \mathcal{V}_{\sigma\rho} \delta \rho + \mathcal{V}_{\sigma \sigma} \delta \sigma 
    \end{pmatrix} \,.
\end{equation}
where $\mathcal{V} = \int \rho \varepsilon \mathsf{d}^d \mathbf{x}$ is the potential energy, and we used the fact that $\mathbf{m}_0 = 0$. Further, we know that $\mathcal{V}_{\rho \sigma} = \mathcal{V}_{\sigma \rho}$. Hence, we find that infinitesimal variations evolve as
\begin{equation}
    \delta z_t
    =
    (J D^2 \mathcal{H})(z_0) \delta z
    \quad \text{which implies} \quad
\begin{cases}
    \delta \rho_t
    =
    - \rho_0 \nabla_{\mathbf{x}} \cdot ( \mathcal{L}_\alpha(\rho_0)^{-1} \delta \mathbf{m}) \,, & \\
    \delta m_t = - \rho_0 \nabla_{\mathbf{x}}( \mathcal{V}_{\rho\rho} \delta \rho + \mathcal{V}_{\rho \sigma} \delta \sigma) \,, & \\
    \delta \sigma_t = - \sigma_0 \nabla_{\mathbf{x}} \cdot ( \mathcal{L}_\alpha(\rho_0)^{-1} \delta \mathbf{m}) \,. & 
\end{cases}
\end{equation}

As is standard, we compute the dispersion relation in Fourier space. We find that
\begin{equation}
    \mathcal{F}[\mathcal{L}_\alpha(\rho_0)](\mathbf{k})
    \coloneq
    \hat{\mathcal{L}}_0
    =
    \rho_0(\mathbb{I} + \alpha \mathbf{k} \otimes \mathbf{k})
    \implies
    \mathcal{F}[\mathcal{L}_\alpha(\rho_0)^{-1}](\mathbf{k})
    =
    \hat{\mathcal{L}}_0^{-1}
    =
    \frac{1}{\rho_0} \left( \mathbb{I} - \frac{\alpha \mathbf{k} \otimes \mathbf{k}}{1 + \alpha | \mathbf{k}|^2} \right) \,,
\end{equation}
using the Sherman-Morrison formula \cite{ShMo1950}. Recalling that $\varepsilon_{\nabla_{\mathbf{x}} \rho} = 2 \varepsilon_{|\nabla_{\mathbf{x}} \rho|^2} \nabla_{\mathbf{x}} \rho$, for constant background, mixed partial derivatives of the internal energy with respect to $\sigma$ and $\nabla_{\mathbf{x}} \rho$ vanish. The second variation of the potential with respect to $(\rho,s)$ is
\begin{equation}
\begin{cases}
    \tilde{\mathcal{V}}_{\rho\rho}
    =
    2 \varepsilon_\rho(\rho_0, s_0) + \rho_0 \varepsilon_{\rho\rho}(\rho_0, s_0) - \nabla_{\mathbf{x}} \cdot( 2 \varepsilon_{|\nabla_{\mathbf{x}} \rho|^2}(\rho_0, s_0) \nabla_{\mathbf{x}} (\cdot) ) \,, & \\
    \tilde{\mathcal{V}}_{\rho s} 
    =
    \varepsilon_s(\rho_0, s_0) + \rho_0 \varepsilon_{\rho s}(\rho_0, s_0) \,, & \\
    \tilde{\mathcal{V}}_{ss} 
    =
    \rho_0 \varepsilon_{ss}(\rho_0, s_0) \,, &
\end{cases}
\end{equation}
where $\tilde{\mathcal{V}}[\rho, s] = \mathcal{V}[\rho, \sigma]$. Since $\delta s = (\delta \sigma - s \delta \rho)/\rho$, we find that
\begin{equation}
\begin{cases}
    \mathcal{V}_{\rho\rho}
    =
    \tilde{\mathcal{V}}_{\rho \rho} - 2 \rho_0^{-1} s_0 \tilde{\mathcal{V}}_{\rho s} + \rho_0^{-2} s_0^2 \tilde{\mathcal{V}}_{ss} \,, & \\
    \mathcal{V}_{\rho \sigma} 
    =
    \rho_0^{-1} \tilde{\mathcal{V}}_{\rho s} - \rho_0^{-2} s_0 \tilde{\mathcal{V}}_{ss} \,, & \\
    \mathcal{V}_{\sigma\sigma}
    =
    \rho_0^{-2} \tilde{\mathcal{V}}_{ss} \,. &
\end{cases}
\end{equation}

If we assume that $\delta z(x,t) = \hat{z} \exp(i \mathbf{k} \cdot \mathbf{x} - i \omega t)$, then we find that the perturbations take the form
\begin{equation}
\begin{cases}
    \omega \hat{\rho} = \rho_0 \mathbf{k} \cdot \hat{\mathcal{L}}_0^{-1} \hat{\mathbf{m}} \,, & \\
    \omega \hat{\mathbf{m}} = \rho_0 \mathbf{k} ( \mathcal{V}_{\rho \rho} \hat{\rho} + \mathcal{V}_{\rho \sigma} \hat{\sigma} ) \,, & \\
    \omega \hat{\sigma} = \sigma_0 \mathbf{k} \cdot \hat{\mathcal{L}}_0^{-1} \hat{\mathbf{m}} \,. &
\end{cases}
\end{equation}
If $\omega \neq 0$, we see that $\hat{\sigma} = s_0 \hat{\rho}$. This allows us to eliminate the equations for $\hat{\rho}$ and $\hat{\sigma}$, leaving
\begin{equation}
    \omega^2 \hat{\mathbf{m}}
    =
    \rho_0 \mathbf{k} \left( \mathcal{V}_{\rho \rho} + s_0 \mathcal{V}_{\rho \sigma} \right) \omega \hat{\rho}
    =
    \rho_0^2 \mathbf{k} \left( \mathcal{V}_{\rho \rho} + s_0 \mathcal{V}_{\rho \sigma} \right) \mathbf{k} \cdot \hat{\mathcal{L}}_0^{-1} \hat{\mathbf{m}} \,.
\end{equation}
Because $c_s^2 = \rho (\rho \varepsilon)_{\rho\rho}$, we find that
\begin{equation}
    \mathcal{V}_{\rho \rho} + s_0 \mathcal{V}_{\rho \sigma}
    =
    \tilde{\mathcal{V}}_{\rho \rho}
    =
    \rho_0^{-1} c_s^2(\rho_0, s_0)
    -
    \nabla_{\mathbf{x}} \cdot( 2 \varepsilon_{|\nabla_{\mathbf{x}} \rho|^2}(\rho_0,s_0) \nabla_{\mathbf{x}} (\cdot) ) \,,
\end{equation}
or, in Fourier space,
\begin{equation}
    \mathcal{V}_{\rho \rho} + s_0 \mathcal{V}_{\rho \sigma}
    =
    \rho_0^{-1} c_s^2(\rho_0, s_0) + 2 \varepsilon_{|\nabla_{\mathbf{x}} \rho|^2}(\rho_0,s_0) |\mathbf{k}|^2 \,.
\end{equation}
So, we have
\begin{equation}
    \left( 
    \omega^2 \mathbb{I} - \frac{c_s^2 + 2 \rho_0 \varepsilon_{|\nabla_{\mathbf{x}} \rho|^2} |\mathbf{k}|^2}{1 + \alpha |\mathbf{k}|^2} \mathbf{k} \otimes \mathbf{k}
    \right)
    \hat{\mathbf{m}}
    =
    0 \,.
\end{equation}
For $\hat{\mathbf{m}} \parallel \mathbf{k}$, we have
\begin{equation}
    \omega
    =
    \pm | \mathbf{k}| 
    \sqrt{\frac{c_s^2 + 2 \rho_0 \varepsilon_{|\nabla_{\mathbf{x}} \rho|^2}}{1 + \alpha |\mathbf{k}|^2}} \,.
\end{equation}
This eigenspace has multiplicity two. For $\hat{\mathbf{m}} \perp \mathbf{k}$, we see that $\omega = 0$. Since there are $d-1$ linearly independent vectors perpendicular to $\hat{\mathbf{k}}$, we find that this eigenspace has multiplicity $d-1$. Finally, if $\hat{\mathbf{m}} = 0$, 
\begin{equation}
     \mathcal{V}_{\rho \rho} \hat{\rho} + \mathcal{V}_{\rho \sigma} \hat{\sigma}
     =
     0 
     \quad \text{which yields} \quad
     \hat{\sigma}
     =
     - \mathcal{V}_{\rho \sigma}^{-1} \mathcal{V}_{\rho \rho} \hat{\rho} \,,
\end{equation}
since $\mathcal{V}_{\rho \sigma}$ is simply a scalar. Hence, we see that, once again, $\omega = 0$, and the eigenspace is one-dimensional. In summary, we have
\begin{itemize}
    \item for $\hat{\mathbf{m}} \parallel \mathbf{k}$, and $\hat{\rho}, \hat{\sigma} \neq 0$,
    \begin{equation}
        \omega
        =
        \pm | \mathbf{k}| 
        \sqrt{\frac{c_s^2 + 2 \rho_0 \varepsilon_{|\nabla_{\mathbf{x}} \rho|^2} |\mathbf{k}|^2}{1 + \alpha |\mathbf{k}|^2}} \,;
    \end{equation}
    \item for $\hat{\mathbf{m}} \perp \mathbf{k}$ and $\hat{\rho} = \hat{\sigma} = 0$, we have $\omega = 0$;
    \item for $\hat{\mathbf{m}} = 0$ and $\mathcal{V}_{\rho \rho} \hat{\rho} + \mathcal{V}_{\rho \sigma} \hat{\sigma} = 0$, we again have $\omega = 0$. 
\end{itemize}
This accounts for the $d+2$ eigenpairs we expect for each wavenumber, $\mathbf{k}$. Finally, we can boost this into any moving reference frame using Galilean invariance and find that the Doppler shifted frequencies are $\omega_{boost}(\mathbf{k}) = \mathbf{k} \cdot \mathbf{u}_0 + \omega(\mathbf{k})$, where $\mathbf{u}_0$ is the velocity of the boosted frame.

\subsection{Eliminating nonphysical dispersion via frequency matching}\label{subsec:frequency-matching}

The $d$ eigenvalues associated with the so-called transverse and thermodynamic waves, with $\omega = \mathbf{u}_0 \cdot \mathbf{k}$, have group and phase velocity both equal to the velocity of the fluid, $\mathbf{u}_0$. This matches standard compressible Euler. However, longitudinal acoustic waves are dispersive:
\begin{equation}
    \omega_{sound}(\mathbf{k})
    =
    \mathbf{k} \cdot \mathbf{u}_0
    \pm
    |\mathbf{k}| \sqrt{\frac{c_s^2 + 2 \rho_0 \varepsilon_{|\nabla_{\mathbf{x}} \rho|^2} |\mathbf{k}|^2}{1 + \alpha |\mathbf{k}|^2}} \,.
\end{equation}
This is the reason for introducing the capillary energy to the Hamiltonian regularization. As mentioned at the beginning of this section, without the capillary energy, we obtain a non-physical dispersion relation for acoustic waves: that is, if $\varepsilon_{|\nabla_{\mathbf{x}} \rho|^2} = 0$, we find that
\begin{equation}
    \omega_{sound}(\mathbf{k})
    =
    \mathbf{k} \cdot \mathbf{u}_0
    \pm
    |\mathbf{k}| \sqrt{\frac{c_s^2}{1 + \alpha |\mathbf{k}|^2}} \,.
\end{equation}
However, if we let $\rho \varepsilon_{|\nabla_{\mathbf{x}} \rho|^2} = \frac{\alpha}{2} c_s^2$, then we instead find that
\begin{equation}
    \omega_{sound}(\mathbf{k})
    =
    \mathbf{k} \cdot \mathbf{u}_0
    \pm
    |\mathbf{k}| c_s \,,
\end{equation}
which is precisely the dispersion relation for acoustic waves in the compressible Euler system. 

Hence, a non-dispersive Hamiltonian regularization of the compressible Euler equations is prescribed as a Lie--Poisson system with bracket \eqref{eq:lie_poisson_bracket} and Hamiltonian
\begin{equation}
\begin{aligned}
    H[\mathbf{m},\rho, \sigma]
    &=
    \int_\Omega 
    \left[
    \frac{1}{2}
    \mathbf{m} \cdot \mathbf{u}[\mathbf{m},\rho]
    +
    \rho
    \varepsilon \left(\rho, \frac{\sigma}{\rho} \right)
    +
    \frac{\alpha c_s^2(\rho, \sigma/\rho)}{2\rho} | \nabla_{\mathbf{x}} \rho |^2
    \right]
    \mathsf{d}^d \mathbf{x} \\
    &= 
    \int_\Omega 
    \left[
    \frac{1}{2}
    \rho \left( |\mathbf{u}|^2 + \alpha (\nabla_{\mathbf{x}} \cdot \mathbf{u})^2 \right)
    +
    \rho
    \varepsilon \left(\rho, s \right)
    +
    \frac{\alpha (\rho \varepsilon)_{\rho\rho}(\rho, s)}{2} | \nabla_{\mathbf{x}} \rho |^2
    \right]
    \mathsf{d}^d \mathbf{x}
    \,,
\end{aligned}
\end{equation}
where $c_s^2 = 2 \rho \varepsilon_\rho + \rho^2 \varepsilon_{\rho \rho} = \rho (\rho \varepsilon)_{\rho\rho}$ is the squared sound speed. We call this system the thermodynamic HRE model. In the unidimensional barotropic case, this recovers a previously known Hamiltonian regularization \cite{guelmame2022hamiltonian}. 

%================================================================
\section{The thermodynamic HRE model in conservative variables}\label{sec:HRE-conservative-variables}
We now seek a formulation of the thermodynamic HRE model in a locally conservative form in the coordinates $(\rho \mathbf{u}, \rho, E)$, where 
\begin{equation}
    E
    =
    \frac{1}{2}
    \rho \left( |\mathbf{u}|^2 + \alpha (\nabla_{\mathbf{x}} \cdot \mathbf{u})^2 \right)
    +
    \rho
    \varepsilon \left(\rho, s \right)
    +
    \frac{\alpha (\rho \varepsilon)_{\rho\rho}(\rho, s)}{2} | \nabla_{\mathbf{x}} \rho |^2
\end{equation}
is the total energy density of the system. Although we already have a locally conservative form of the model in the coordinates $(\mathbf{m}, \rho, \sigma)$ in Equation~\eqref{eq:hre_entropy_evolution_equations}, it is preferable to obtain evolution equations in the mass, momentum, total energy variables for the purposes of developing viable numerical methods and to compare the model with thermodynamic IGR \cite{wilfong2025simulatingmanyenginespacecraftexceeding}. Unfortunately, mass, momentum, and total energy is not a convenient coordinate system to express in Hamiltonian form, as the Poisson bracket becomes unwieldy in these coordinates. Therefore, we effect the transformation from the canonical Lagrangian coordinates to the Eulerian coordinates $(\rho \mathbf{u}, \rho, E)$ sequentially. 

\subsection{The momentum and continuity equations}\label{subsec:momentum-and-continuity}

We augment the results previously found in Section~\ref{sec:pressureless-IGR} by considering what happens to the HRE model when we add a potential to the Lagrangian:
\begin{equation}
    L[\bm{\Phi}, \dot{\bm{\Phi}}]
    =
    \frac{1}{2} g_\alpha[\bm{\Phi}](\dot{\bm{\Phi}}, \dot{\bm{\Phi}})
    -
    \mathcal{V}[\bm{\Phi}] \,,
\end{equation}
where
\begin{equation}
    \mathcal{V}[\bm{\Phi}]
    =
    \int_\Omega \rho_0 \bar{\varepsilon} \left( \frac{\rho_0}{\det(\nabla_{\mathbf{X}} \bm{\Phi})}, s_0, \nabla_{\mathbf{X}} \bm{\Phi}^{-T} \nabla_{\mathbf{X}} \left( \frac{\rho_0}{\det(D\bm{\Phi})} \right) \right) \mathsf{d}^d \mathbf{X}_0 \,.
\end{equation}
This is the Lagrangian of the thermodynamic HRE model. It is possible to show, see \cite{suzuki2020generic}, that
\begin{equation} \label{eq:derivative_of_potential}
    \frac{\delta \mathcal{V}}{\delta \bm{\Phi}}
    =
    - (\det(\nabla_{\mathbf{X}} \bm{\Phi})) \left[ \nabla_{\mathbf{x}} \cdot \left( [-\bar{p} + \rho \nabla_{\mathbf{x}} \cdot \bm{\xi}] \mathbb{I} - \bm{\xi} \otimes \nabla_{\mathbf{x}} \rho \right) \right]_{x = \bm{\Phi}(\mathbf{X},t)} \,,
\end{equation}
where $\bar{p} = \rho^2 \bar{\varepsilon}_\rho$, and $\bm{\xi} = \rho \bar{\varepsilon}_{\nabla_{\mathbf{x}} \rho}$. Splitting the potential energy into the bulk and capillary energies found in the previous section,
\begin{equation}
    \bar{\varepsilon}(\rho, s, \nabla_\bx \rho)
    =
    \varepsilon(\rho, s) + \frac{\alpha (\rho \varepsilon)_{\rho\rho}}{2 \rho} | \nabla_\bx \rho |^2 \,,
\end{equation}
so that
\begin{multline}
    \bar{p}
    =
    \rho^2
    \left(
    \varepsilon_\rho + \frac{\alpha ((\rho \varepsilon)_{\rho\rho}/\rho)_\rho}{2} | \nabla_\bx \rho |^2
    \right)
    =
    p 
    +
    \left( \frac{\alpha \rho (\rho \varepsilon)_{\rho\rho\rho}}{2} - \frac{\alpha (\rho \varepsilon)_{\rho \rho}}{2} \right) | \nabla_\bx \rho |^2
    \,, \\
    \quad \text{and} \quad
    \bm{\xi}
    =
    \alpha (\rho \varepsilon)_{\rho\rho} \nabla_\bx \rho \,,
\end{multline}
where $p = \rho^2 \varepsilon_\rho$ is the thermodynamic pressure of the compressible Euler system without regularization. As a convenient shorthand, we will continue to use the generalized pressure, $\bar{p}$, and capillary force, $\bm{\xi}$, as we simplify the momentum and energy equations. In the Lagrangian reference frame, the weak-form Euler-Lagrange equation are given by
\begin{equation}
    g_\alpha[\bm{\Phi}](\ddot{\bm{\Phi}}, \cdot)
    =
    \Gamma_\alpha[\bm{\Phi}](\dot{\bm{\Phi}}, \dot{\bm{\Phi}}, \cdot) - D_{\bm{\Phi}} \mathcal{V}(\cdot) \,,
\end{equation}
where $\Gamma_\alpha$ is the regularizing force arising from the Levi-Civita geodesic flow, see \Cref{eq:pressureless-HRE-abstract}. 

As usual, the continuity equation is simply $\partial_t \rho + \nabla_{\mathbf{x}} \cdot (\rho \mathbf{u}) = 0$. It remains entirely unchanged from the regularization, since the relationship between the Lagrangian and Eulerian densities remains unchanged: $\rho(\bm{\Phi}_t(\mathbf{X})) = \det(\nabla_{\mathbf{X}} \bm{\Phi}_t)^{-1} \rho_0(\mathbf{X})$. 

The momentum equation in Eulerian coordinates is obtained by following the procedure discussed in \Cref{sec:eulerianization}:
\begin{equation}
    \mathcal{L}_\alpha(\rho) \frac{D \mathbf{u}}{D t}
    =
    - \nabla_{\bx} \cdot \left( \alpha \rho \mathrm{tr}((\nabla_{\bx} \bu)^2) \mathbb{I} + [\bar{p} - \rho \nabla_{\mathbf{x}} \cdot \bm{\xi}] \mathbb{I} + \bm{\xi} \otimes \nabla_{\mathbf{x}} \rho \right) \,,
\end{equation}
which may be rewritten as
\begin{equation} \label{eq:thermal_hre_newton_law}
    \rho \frac{D \mathbf{u}}{D t}
    =
    - \nabla_{\mathbf{x}} \cdot \mathbb{\Sigma} \,,
\end{equation}
where
\begin{equation} \label{eq:matrix_conservative_entropic_pressure}
    \rho^{-1} \mathbb{\Sigma} - \alpha \nabla_{\mathbf{x}} \cdot (\rho^{-1} \nabla_{\mathbf{x}} \cdot \mathbb{\Sigma} )\mathbb{I}
    =
    \alpha \mathrm{tr}((\nabla_{\mathbf{x}} \mathbf{u})^2) \mathbb{I}
    +
    \left(\frac{\bar{p}}{\rho} - \nabla_{\mathbf{x}} \cdot \bm{\xi} \right) \mathbb{I} + \frac{\bm{\xi}}{\rho} \otimes \nabla_{\mathbf{x}} \rho \,.
\end{equation}
Notice that the entropic pressure $\mathbb{\Sigma}$ is now a rank-$2$ tensor due to the anisotropy of the Korteweg stress. This is in contrast with IGR and the pressureless HRE models, in which the entropic pressure is isotropic. We will find that this elliptic equation may be reduced to a scalar equation.

\subsection{Local kinetic energy transport}\label{subsec:local-kinetic-energy-transport} 

As in Section~\Cref{sec:pressureless-IGR}, we define the local kinetic energy to be
\begin{equation}
    K_{\mathcal{D}}[\bm{\Phi}_t, \dot{\bm{\Phi}}_t]
    =
    \frac{1}{2} \int_{\mathcal{D}} \rho_0 \left( |\dot{\bm{\Phi}}_t|^2 + \alpha \mathrm{tr}(\nabla_{\mathbf{X}} \bm{\Phi}_t^{-1} \nabla_{\mathbf{X}} \dot{\bm{\Phi}}) \right) \mathsf{d}^d \mathbf{X}_0
\end{equation}
in the Lagrangian reference frame, or
\begin{equation}
    K_{\mathcal{D}_t} [\mathbf{u}, \rho]
    =
    \frac{1}{2} \int_{\mathcal{D}_t} \rho \left( |\mathbf{u}|^2 + \alpha (\nabla_{\mathbf{x}} \cdot \mathbf{u})^2 \right) \mathsf{d}^d \mathbf{x}
\end{equation}
in the Eulerian reference frame. As before, we find that
\begin{equation}
    \dot{K}_{\mathcal{D}_t}
    =
    \frac{\mathsf{d}}{\mathsf{d}t}
    \int_{\mathcal{D}_t} K_E \mathsf{d}^d \mathbf{x}
    =
    \int_{\mathcal{D}_t}
    \left( 
    \rho \frac{D \mathbf{u}}{Dt} \cdot \mathbf{u}
    +
    \alpha \rho (\nabla_{\mathbf{x}} \cdot \mathbf{u}) \nabla_{\mathbf{x}} \cdot \left( \frac{D \mathbf{u}}{Dt} \right)
    -
    \alpha \rho (\nabla_{\mathbf{x}} \cdot \mathbf{u}) \mathrm{tr}((\nabla_{\mathbf{x}} \mathbf{u})^2)
    \right)
    \mathsf{d}^d \mathbf{x} \,.
\end{equation}
Again, Reynolds' transport theorem implies that
\begin{equation} \label{eq:unsimplified_thermal_hre_energy_transport}
\begin{aligned}
    \partial_t K_E
    +
    \nabla_{\mathbf{x}} \cdot(K_E \mathbf{u})
    &=
    \rho \frac{D \mathbf{u}}{Dt} \cdot \mathbf{u}
    +
    \alpha \rho (\nabla_{\mathbf{x}} \cdot \mathbf{u}) \nabla_{\mathbf{x}} \cdot \left( \frac{D \mathbf{u}}{Dt} \right)
    -
    \alpha \rho (\nabla_{\mathbf{x}} \cdot \mathbf{u}) \mathrm{tr}((\nabla_{\mathbf{x}} \mathbf{u})^2) \\
    &=
    \rho \frac{D \mathbf{u}}{Dt} \cdot \mathbf{u}
    +
    \alpha \rho (\nabla_{\mathbf{x}} \cdot \mathbf{u}) \nabla_{\mathbf{x}} \cdot \left( \frac{D \mathbf{u}}{Dt} \right)
    -
    \alpha \rho \nabla_{\mathbf{x}} \mathbf{u} : \mathrm{tr}((\nabla_{\mathbf{x}} \mathbf{u})^2) \mathbb{I} \,.
\end{aligned}
\end{equation}
Up to this point, what we find is identical to the pressureless case. After substituting Equation~\eqref{eq:thermal_hre_newton_law} into Equation~\eqref{eq:unsimplified_thermal_hre_energy_transport}, and proceeding with the same simplifications as in Section~\ref{sec:pressureless-IGR}, we find that
\begin{equation} \label{eq:local-KE-transport_HRE}
    \partial_t K_E + \nabla_{\mathbf{x}} \cdot ( (K_E \mathbb{I} + \mathbb{\Sigma}) \cdot \mathbf{u}) + \nabla_{\mathbf{x}} \mathbf{u} \cdot \left( [-\bar{p} + \rho \nabla_{\mathbf{x}} \cdot  \bm{\xi}] \mathbb{I} - \bm{\xi} \otimes \nabla_{\mathbf{x}} \rho \right) = 0 \,,
\end{equation}
for the thermodynamic HRE model. The non-conservative term represents reversible transfer between kinetic and potential energy. We will find that, when combined with the evolution equation for internal energy, one obtains a fully conservative system.

If we include the non-conservative part of the entropic pressure in the momentum equation, the local kinetic energy transport is given by
\begin{equation}\label{eq:local-KE-transport}
    \partial_t K_E + \nabla_{\mathbf{x}} \cdot ( (K_E \mathbb{I} + \mathbb{\Sigma}) \cdot \mathbf{u}) + \nabla_{\mathbf{x}} \mathbf{u} \cdot \left( [-\bar{p} + \rho \nabla_{\mathbf{x}} \cdot  \bm{\xi}] \mathbb{I} - \bm{\xi} \otimes \nabla_{\mathbf{x}} \rho \right) = \alpha \rho (\nabla_{\mathbf{x}} \cdot  \mathbf{u})^3 \,,
\end{equation}
where the entropic pressure now includes the additional non-conservative term
\begin{equation} \label{eq:matrix_entropic_pressure_non_cons}
    \rho^{-1} \mathbb{\Sigma} - \alpha \nabla_{\mathbf{x}} \cdot (\rho^{-1} \nabla_{\mathbf{x}} \cdot \mathbb{\Sigma} )\mathbb{I}
    =
    \alpha (\mathrm{tr}^2(\nabla_{\mathbf{x}} \mathbf{u}) + \mathrm{tr}((\nabla_{\mathbf{x}} \mathbf{u})^2)) \mathbb{I}
    +
    \left(\frac{\bar{p}}{\rho} - \nabla_{\mathbf{x}} \cdot \bm{\xi} \right) \mathbb{I} + \frac{\bm{\xi}}{\rho} \otimes \nabla_{\mathbf{x}} \rho \,.
\end{equation}
The non-conservative term on the right-hand side represents a dissipative coupling between the kinetic and internal energy and must be coupled with a suitable entropy producing process in the internal energy equation. This is accomplished in \Cref{sec:full-IGR-dissipative-extension}. 

\subsection{Local internal and total energy transport}\label{subsec:local-internal-energy-transport}
We obtain the local internal energy transport equation via a sequence of coordinate transformations from the canonical Lagrangian coordinates to the coordinate system $(\mathbf{m},\rho, e)$, where 
\begin{equation}
    e(\rho,s)
    =
    \rho
    \varepsilon \left(\rho, s \right)
    +
    \frac{\alpha (\rho \varepsilon)_{\rho\rho}(\rho, s)}{2} | \nabla_{\mathbf{x}} \rho |^2
\end{equation}
is the internal energy density. In \Cref{appendix:change_of_variables}, we systematically derive the local transport equation for the internal energy. We find that
\begin{equation} \label{eq:internal_energy_density_transport}
    \partial_t e = - \nabla_{\mathbf{x}} \cdot ( e \mathbf{u}) + (\nabla_{\mathbf{x}} \mathbf{u}) \cdot [ (- \bar{p} + \rho (\nabla_{\mathbf{x}} \cdot \bm{\xi})) \mathbb{I} - \nabla_{\mathbf{x}} \rho \otimes \bm{\xi} ] - \nabla_{\mathbf{x}} \cdot( \rho (\nabla_{\mathbf{x}} \cdot \mathbf{u}) \bm{\xi}) \,.
\end{equation}
The transport equation for total energy density, $E = K_E + e$, is obtained by summing up Equations~\eqref{eq:local-KE-transport_HRE} and \eqref{eq:internal_energy_density_transport} to find
\begin{equation}
    \partial_t E + \nabla_{\mathbf{x}} \cdot \left( E \mathbf{u} + \mathbb{\Sigma} \cdot \mathbf{u} + \rho (\nabla_{\mathbf{x}} \cdot \mathbf{u}) \bm{\xi} \right) = 0 \,,
\end{equation}
for the thermodynamic HRE system, and
\begin{equation} \label{eq:igr_energy_transport}
    \partial_t E + \nabla_{\mathbf{x}} \cdot \left( E \mathbf{u} + \mathbb{\Sigma} \cdot \mathbf{u} + \rho (\nabla_{\mathbf{x}} \cdot \mathbf{u}) \bm{\xi} \right) = \alpha \rho (\nabla_{\mathbf{x}} \cdot \mathbf{u})^3 \,,
\end{equation}
if we include the non-conservative part of the entropic pressure in the momentum equation, where $\mathbb{\Sigma}$ solves ~\Cref{eq:matrix_entropic_pressure_non_cons}. As mentioned at the end of \Cref{subsec:local-energy-transport-pressureless-IGR} we will show in \Cref{sec:full-IGR-dissipative-extension} that we may restore total energy conservation by asserting that the dissipated kinetic energy due to the non-conservative entropic pressure is converted into internal energy as heat. 

\subsection{Simplifications for the elliptic entropic pressure equation}\label{subsec:simplified-HRE} 

Much of the complexity of the thermodynamic HRE model is hidden in the elliptic solve for the entropic pressure, $\mathbb{\Sigma}$, given in~\Cref{eq:matrix_conservative_entropic_pressure}. It is therefore reasonable to seek simplifications. Dispensing with the shorthand used throughout this section, the full entropic pressure equation is given by
\begin{multline} \label{eq:full_conservative_entropic_pressure}
    \rho^{-1} \mathbb{\Sigma} - \alpha \nabla_{\mathbf{x}} \cdot (\rho^{-1} \nabla_{\mathbf{x}} \cdot \mathbb{\Sigma} )\mathbb{I}
    =
    \left( \rho^{-1} p - \alpha \nabla_{\mathbf{x}} \cdot ((\rho \varepsilon)_{\rho\rho} \nabla_\bx \rho) \right) \mathbb{I}  + \alpha \mathrm{tr}((\nabla_{\mathbf{x}} \mathbf{u})^2) \mathbb{I} \\
    + 
    \alpha
    \left[
    \left( \left( \frac{ (\rho \varepsilon)_{\rho\rho\rho}}{2} - \frac{ (\rho \varepsilon)_{\rho \rho}}{2 \rho} \right) | \nabla_\bx \rho |^2 \right) \mathbb{I} + \frac{(\rho \varepsilon)_{\rho\rho}}{\rho} \nabla_\bx \rho \otimes \nabla_{\bx} \rho
    \right] \,.
\end{multline}
This equation takes the form
\begin{equation}
    \rho^{-1} \mathbb{\Sigma} - \alpha \nabla_{\mathbf{x}} \cdot (\rho^{-1} \nabla_{\mathbf{x}} \cdot \mathbb{\Sigma} )\mathbb{I}
    =
    f(\bx) \mathbb{I} + \mathbb{G}(\bx) \,,
\end{equation}
where $f: \Omega \to \mathbb{R}$ and $\mathbb{G}: \Omega \to \mathbb{R}^{d \times d}$. In~\Cref{appendix:simplifying_entropic_pressure}, we show that with adequately regularity conditions, this may be reduced to a scalar elliptic solve:
\begin{equation}
\left\{
\begin{aligned}
    &\mathbb{\Sigma}
    =
    \rho \mathbb{G} + \Sigma \mathbb{I} \,, \\
    &\rho^{-1} \Sigma - \alpha \nabla_{\bx} \cdot (\rho^{-1} \nabla_{\bx} \Sigma)
    =
    \alpha \nabla_{\bx} \cdot (\rho^{-1} \nabla_\bx \cdot (\rho \mathbb{G})) \,.
\end{aligned}
\right. 
\end{equation}
Reducing the matrix elliptic solve to a scalar elliptic solve comes at the cost of introducing higher-order derivatives. Indeed, whether
\begin{equation}
    \mathbb{G}(\bx)
    =
    \alpha \frac{(\rho \varepsilon)_{\rho\rho}}{\rho} \nabla_\bx \rho \otimes \nabla_{\bx} \rho
\end{equation}
is twice differentiable is unclear. The one-dimensional numerical experiments in ~\Cref{subsec:1d_numerics} indicate that the higher-order derivatives introduced in this Hamiltonian regularization are problematic.

As another consideration, the first two terms on the right-hand side of ~\Cref{eq:full_conservative_entropic_pressure} are natural to group together when one recalls that $\rho (\rho \varepsilon)_{\rho \rho} = c_s^2 = p_\rho$. In the barotropic case, the pressure may be removed from the elliptic equation since
\begin{equation}
    \rho^{-1} p - \alpha \nabla_{\mathbf{x}} \cdot ((\rho \varepsilon)_{\rho\rho} \nabla_\bx \rho)
    =
    \rho^{-1} p - \alpha \nabla_{\mathbf{x}} \cdot (\rho^{-1} p_\rho(\rho) \nabla_\bx \rho)
    =
    \rho^{-1} p - \alpha \nabla_{\mathbf{x}} \cdot (\rho^{-1} \nabla_\bx p(\rho)) \,.
\end{equation}
Thus, $\mathbb{\Sigma} = p \mathbb{I} + \tilde{\mathbb{\Sigma}}$, where $\tilde{\mathbb{\Sigma}}$ solves the same elliptic equation with the first two terms omitted from the right-hand side. This simplification is not possible for a general, energy-dependent equation of state. We may nonetheless use ~\Cref{lemma:matrix_elliptic_eqn} to extract the pressure force from the elliptic solve. We then find that $\mathbb{\Sigma} = p \mathbb{I} + \tilde{\mathbb{\Sigma}}$, where
\begin{multline} 
    \rho^{-1} \tilde{\mathbb{\Sigma}} - \alpha \nabla_{\mathbf{x}} \cdot (\rho^{-1} \nabla_{\mathbf{x}} \cdot \tilde{\mathbb{\Sigma}} )\mathbb{I}
    =
    \alpha \left( \nabla_{\bx} \cdot (\rho^{-1} \nabla_\bx p) - \nabla_{\mathbf{x}} \cdot ((\rho \varepsilon)_{\rho\rho} \nabla_\bx \rho) \right) \mathbb{I}  + \alpha \mathrm{tr}((\nabla_{\mathbf{x}} \mathbf{u})^2) \mathbb{I} \\
    + 
    \alpha
    \left[
    \left( \left( \frac{ (\rho \varepsilon)_{\rho\rho\rho}}{2} - \frac{ (\rho \varepsilon)_{\rho \rho}}{2 \rho} \right) | \nabla_\bx \rho |^2 \right) \mathbb{I} + \frac{(\rho \varepsilon)_{\rho\rho}}{\rho} \nabla_\bx \rho \otimes \nabla_{\bx} \rho
    \right] \,.
\end{multline}
This formulation has the advantage of making all terms on the right-hand side of the elliptic equation $O(\alpha)$, which we find is important for numerical stability in one-dimensional numerical experiments. However, this comes at the cost of taking higher-order derivatives of the pressure. 

\begin{remark}
    Note that the density derivatives of the internal energy density are taken at fixed entropy. This may give rise to confusion since entropy has been eliminated as a variable in factor of total energy. When using the internal energy as the prognostic variable, a more explicit notation is to write $(\rho \varepsilon)_{\rho} = \left. \partial_\rho (\rho \varepsilon) \right|_s$. 
\end{remark}

\begin{remark}
    In the limit $\alpha \to 0$, we find $\mathbb{\Sigma} = p \mathbb{I}$. Thus, in this limit, we recover the compressible Euler equations without regularization, as expected.
\end{remark}

\subsection{Summary of the thermodynamic HRE model}\label{subsec:summary-HRE} 

A non-dispersive Hamiltonian regularization of the compressible Euler equations is prescribed as a Lie--Poisson system with energy
\begin{equation}
\begin{aligned}
    H[\mathbf{m},\rho, \sigma]
    &=
    \int_\Omega 
    \left[
    \frac{1}{2}
    \mathbf{m} \cdot \mathbf{u}[\mathbf{m},\rho]
    +
    \rho
    \varepsilon \left(\rho, \frac{\sigma}{\rho} \right)
    +
    \frac{\alpha c_s^2(\rho, \sigma/\rho)}{2\rho} | \nabla_{\mathbf{x}} \rho |^2
    \right]
    \mathsf{d}^d \mathbf{x} \\
    &= 
    \int_\Omega 
    \left[
    \frac{1}{2}
    \rho \left( |\mathbf{u}|^2 + \alpha (\nabla_{\mathbf{x}} \cdot \mathbf{u})^2 \right)
    +
    \rho
    \varepsilon \left(\rho, s \right)
    +
    \frac{\alpha (\rho \varepsilon)_{\rho\rho}(\rho, s)}{2} | \nabla_{\mathbf{x}} \rho |^2
    \right]
    \mathsf{d}^d \mathbf{x}
    \,,
\end{aligned}
\end{equation}
where $c_s^2 = 2 \rho \varepsilon_\rho + \rho^2 \varepsilon_{\rho \rho} = \rho (\rho \varepsilon)_{\rho\rho}$ is the squared sound speed. In Lie--Poisson Hamiltonian coordinates, the system is given by Equation~\eqref{eq:hre_entropy_evolution_equations}, while in locally-conservative form it is given by
\begin{equation}\label{eq:thermal-HRE}
\left\{
\begin{aligned}
    &\partial_t
    \begin{bmatrix}
        \rho \\
        \rho \bu \\
        E
    \end{bmatrix}
    +
    \nabla_{\bx} \cdot
    \begin{bmatrix}
        \rho \mathbf{u} \\
        \rho \mathbf{u} \otimes \mathbf{u} + \mathbb{\Sigma} \\
        E \mathbf{u} + \mathbb{\Sigma} \cdot \mathbf{u} + \alpha c_s^2 (\nabla_\bx \cdot \bu) \nabla_\bx \rho / \rho
    \end{bmatrix}
    =
    0 \,, \\
    &\begin{aligned}
    \rho^{-1} \mathbb{\Sigma} - \alpha \nabla_{\mathbf{x}} \cdot (\rho^{-1} \nabla_{\mathbf{x}} \cdot \mathbb{\Sigma} )\mathbb{I}
    &=
    \left( \rho^{-1} p - \alpha \nabla_{\mathbf{x}} \cdot (\rho^{-1} c_s^2 \nabla_\bx \rho) \right) \mathbb{I}  + \alpha \mathrm{tr}((\nabla_{\mathbf{x}} \mathbf{u})^2) \mathbb{I} \\
    &+ 
    \alpha
    \left[
    \left( \left( \frac{ \rho^2 \left. \partial_\rho^3 (\rho \varepsilon) \right|_s }{2} - \frac{ c_s^2}{2} \right) \left|\frac{\nabla_\bx \rho}{\rho} \right|^2 \right) \mathbb{I} + c_s^2 \frac{\nabla_\bx \rho}{\rho} \otimes \frac{\nabla_\bx \rho}{\rho}
    \right]
    \end{aligned} \,,
\end{aligned}
\right.
\end{equation}
where $p$, $c_s^2$, and $\left. \partial_\rho^3 (\rho \varepsilon) \right|_s$ are computed from an equation of state. This system enjoys the following properties:
\begin{itemize}
    \item its linear dispersion relation about a constant equilibrium, $(\mathbf{u}_0, \rho_0, s_0)$, matches that of the compressible Euler equations, and it formally recovers compressible Euler in the limit $\alpha \to 0$,
    \item it recovers the pressureless HRE model if we let $\varepsilon = 0$,
    \item the model is thermodynamically consistent (conserving total energy and entropy),
    \item and the model is locally conservative in the variables $(\rho \mathbf{u}, \rho, E)$. 
    \item the Hamiltonian structure facilitates thermodynamically consistent dissipative extensions through the metriplectic formalism, see \Cref{sec:full-IGR-dissipative-extension},
\end{itemize}
This model is the only known multi-dimensional, non-dispersive Hamiltonian regularization of the compressible Euler equations with general thermodynamics. 

\begin{remark}
    While closed-form expressions for the pressure, $p$, and the squared sound speed, $c_s^2$, as functions of density and internal energy are generally available even for tabulated equations of state, the third derivative of the internal energy with respect to density at fixed entropy, $\left. \partial_\rho^3 (\rho \varepsilon) \right|_s$, is not a commonly known thermodynamic quantity. E.g., for a polytropic ideal gas, we have
    \begin{multline}
        \varepsilon(\rho, s)
        =
        \frac{\rho^{\gamma-1} e^s}{\gamma-1} \,,
        \quad \text{so that} \quad
        p(\rho,s)
        =
        \rho \varepsilon_\rho
        =
        \rho^\gamma e^s \,,
        \quad
        c_s^2
        =
        p_\rho
        =
        \gamma \rho^{\gamma-1} e^s \,, \\
        \quad \text{and} \quad
        (\rho \varepsilon)_{\rho\rho\rho}
        =
        (\gamma - 2)(\gamma-3) \rho^{\gamma - 4} e^s \,.
    \end{multline}
    Inverting these expressions, we find
    \begin{equation}
        p(\rho, \varepsilon)
        =
        (\gamma - 1) \rho \varepsilon \,,
        \quad
        c_s^2(\rho, \varepsilon)
        =
        \gamma (\gamma - 1) \varepsilon \,,
        \quad \text{and} \quad
        \left. \partial_\rho^3 (\rho \varepsilon) \right|_s (\rho, \varepsilon)
        =
        (\gamma - 1)(\gamma - 2)(\gamma-3) \varepsilon/ \rho^3 \,.
    \end{equation}
    Further, recall that the internal energy is obtained from the total energy by subtracting the generalized kinetic energy:
    \begin{equation}
        \varepsilon
        =
        \rho^{-1}
        \left(
        E - \frac{1}{2} \rho ( | \bu |^2 + \alpha (\nabla_\bx \cdot \bu)^2 ) 
        \right) \,.
    \end{equation}
\end{remark}

\begin{remark}
    By letting $\varepsilon = \varepsilon(\rho)$, we obtain a barotropic HRE model. In this case, we eliminate the entropy equation in Equation~\eqref{eq:hre_entropy_evolution_equations} and the energy equation in Equation~\eqref{eq:thermal-HRE}. However, the total energy of the system, given by the Hamiltonian, is still conserved. As mentioned in~\Cref{subsec:simplified-HRE}, in the barotropic, unidimensional case, the system simplifies significantly:
    \begin{equation}
    \begin{aligned}
        (\rho u)_t + (\rho u^2 + p(\rho) + \Sigma)_x &= 0 \,, \\
        \rho_t + (\rho u)_x &= 0 \,, \\
        \rho^{-1} \Sigma - \alpha (\rho^{-1} \Sigma_x)_x &= 
        \alpha 
        \left(
        u_x^2 + \frac{(\rho (\rho \varepsilon)_{\rho \rho})_\rho}{2 \rho} \rho_x^2
        \right) \,,
    \end{aligned}
    \end{equation}
    which recovers a previous model in the literature \cite{guelmame2022hamiltonian}. 
\end{remark}

%================================================================
\section{Hamiltonian IGR: a dissipative extension of HRE}\label{sec:full-IGR-dissipative-extension}
We now introduce the final system, which we will refer to as Hamiltonian IGR, which can be viewed as both a dissipative extension of thermodynamic HRE, as well as a entropy-consistent thermodynamic extension of the IGR (where, by entropy consistent, we mean that the sign of the entropy production rate agrees with the sign of $\nabla\cdot u$).

The remaining non-conservative force in IGR appears in the momentum equation, in the canonical coordinate $\mathbf{m}$, as
\begin{equation}
    \partial_t \mathbf{m} - \{\mathbf{m}, H\} = - \alpha \nabla_{\mathbf{x}} (\rho (\nabla_{\mathbf{x}} \cdot \mathbf{u})^2) \,,
\end{equation}
where $\{\mathbf{m}, H\}$ captures the previously described Hamiltonian contribution to the equations. The objective of this section is to find a suitable extension of HRE using the metriplectic formalism (see \cite{zaidni2025metriplectic, PhysRevE.109.045202} and also, \Cref{appendix:metriplectic}), which adds this remaining term with appropriate reciprocal couplings to total energy and entropy. This non-conservative term is crucial for the limit $\alpha \to 0$ to properly recover shock physics, as the physically valid weak solution to the shock problem is found by requiring entropy production across shocks \cite{godlewski2013numerical}. 

As we saw previously in equation \eqref{eq:igr_energy_transport}, energy transport in the IGR equation includes an additional non-conservative term, $\alpha \rho (\nabla_{\mathbf{x}} \cdot \mathbf{u})^3$, on the right-hand side. By adding a dissipative coupling to the internal energy, we can restore conservation of total energy. Thus, we posit that
\begin{equation} \label{eq:target_entropy_production}
    \partial_t e - \{e, H\} = - \alpha \rho (\nabla_{\mathbf{x}} \cdot \mathbf{u})^3 \,,
    \quad \text{which implies} \quad
    \partial_t \sigma - \{\sigma, H \} = - \alpha \frac{\rho}{\vartheta} (\nabla_{\mathbf{x}} \cdot \mathbf{u})^3 \,.
\end{equation}
This provides us with a target entropy production rate for our system that we use to design our dissipative regularization using the metriplectic formalism. 

\begin{remark}
    Strictly speaking, a true metriplectic system is a thermodynamically consistent dissipative extension of Hamiltonian mechanics, requiring non-negative entropy production. \Cref{eq:target_entropy_production} allows for entropy depletion, functioning more like a source/sink than a true dissipative process. However, this violation of the second law of thermodynamics is a feature of IGR, and not a bug. The non-conservative force functions as a sink for converging flows, thus producing entropy at shock fronts in the limit as $\alpha \to 0$, while depleting entropy in diverging flows, thus ensuring that acoustic waves are conserved on average. A true thermodynamically-consistent dissipative process could not enjoy this property of IGR. 
\end{remark}

The precise form of the Hamiltonian regularization follows as an inevitable consequence of our modeling assumptions: we penalize Lagrangian trajectories approaching one another via a logarithmic barrier and add thermodynamic pressure such that acoustic waves are non-dispersive. The dissipative regularization discussed in this section is heuristically motivated. An ideal dissipative regularization would ensure that we add exactly the right amount of entropy locally to ensure convergence to the weak solution of the compressible Euler equations satisfying the Rankine-Hugoniot conditions in the limit $\alpha \to 0$. Techniques for defining weak, dissipative solutions to similar Hamiltonian regularizations of hyperbolic conservation laws have been investigated elsewhere \cite{guelmame2024hamiltonian}, and a similar treatment for the HRE system is warranted. 

\subsection{Designing a dissipative process using an indefinite metriplectic 4-bracket}\label{subsec:metriplectic-4-bracket}

Consider the following construction based on the metriplectic $4$-bracket formalism for dissipative dynamics. Define following symmetric bilinear operators:
\begin{equation}
    A(F,G)
    =
    - \frac{\alpha \rho (\nabla_{\mathbf{x}} \cdot \mathbf{u})}{\vartheta} \left(\nabla_{\mathbf{x}} \cdot \frac{\delta F}{\delta \mathbf{m}} \right) \left(\nabla_{\mathbf{x}} \cdot \frac{\delta G}{\delta \mathbf{m}} \right) \,,
    \quad \text{and} \quad
    B(F,G)
    =
    \frac{\delta F}{\delta \sigma} \frac{\delta G}{\delta \sigma} \,.
\end{equation}
As described in \Cref{appendix:metriplectic}, the rationale for picking these components for building the metriplectic $4$-bracket is that we know the target local entropy production rate \eqref{eq:target_entropy_production}, and we know that this can be directly related to our choice of symmetric operator $A$:
\begin{equation}
    \partial_t \sigma
    -
    \{\sigma, H\}
    =
    (\sigma, H)
    =
    A(H,H)
    =
    - \alpha \frac{\rho (\nabla_{\mathbf{x}} \cdot \mathbf{u})^3}{\vartheta} \,.
\end{equation}
The 4-bracket is then obtained from the Kulkarni--Nomizu product (see \Cref{appendix:metriplectic}) of $A$ and $B$: 
\begin{equation}
\begin{aligned}
    (F&,K;G,N)
    =
    \int_\Omega
    (A\varowedge B)(F, K, G, N) \mathsf{d}^d \mathbf{x} \\
    & =
    \int_\Omega \left[
    A(F, G) B(K, N) 
    - A(F, N) B(G, K) 
    + B(F, G) A(K, N) 
    - B(F, N) A(G, K) \right] \mathsf{d}^d \mathbf{x} \\
    &=
    - \int_\Omega \frac{\alpha \rho \nabla_{\mathbf{x}} \cdot \mathbf{u}}{\vartheta} 
    \left[ \frac{\delta K}{\delta \sigma} \nabla_{\mathbf{x}} \cdot \frac{\delta F}{\delta \mathbf{m}} - \frac{\delta F}{\delta \sigma} \nabla_{\mathbf{x}} \cdot \frac{\delta K}{\delta \mathbf{m}} \right] 
    \left[ \frac{\delta N}{\delta \sigma} \nabla_{\mathbf{x}} \cdot \frac{\delta G}{\delta \mathbf{m}} - \frac{\delta G}{\delta \sigma} \nabla_{\mathbf{x}} \cdot \frac{\delta N}{\delta \mathbf{m}} \right] \mathsf{d}^d \mathbf{x} \,.
\end{aligned}
\end{equation}
Choosing $K = N = H$, we find that the dissipative bracket is given by
\begin{equation}
    (F,G)
    =
    (F,H;G,H)
    =
    - \int_\Omega \frac{\alpha \rho \nabla_{\mathbf{x}} \cdot \mathbf{u}}{\vartheta} 
    \left[ \vartheta \nabla_{\mathbf{x}} \cdot \frac{\delta F}{\delta \mathbf{m}} - \frac{\delta F}{\delta \sigma} \nabla_{\mathbf{x}} \cdot \mathbf{u} \right] 
    \left[ \vartheta \nabla_{\mathbf{x}} \cdot \frac{\delta G}{\delta \mathbf{m}} - \frac{\delta G}{\delta \sigma} \nabla_{\mathbf{x}} \cdot \mathbf{u} \right] \mathsf{d}^d \mathbf{x} \,.
\end{equation}
This procedure for designing dissipative processes automatically guarantees energy conservation while building in the desired entropy production rate \cite{PhysRevE.109.045202, zaidni2025metriplectic}. 

Letting $G = S = \int_\Omega \sigma \mathsf{d}^3 \bx$, we can recover the dissipative evolution. For the entropy production rule, we take $F = \int \psi_\sigma \sigma \mathsf{d}^d \mathbf{x}$ where $\psi_\sigma$ is arbitrary, which yields
\begin{equation}
    \partial_t \sigma
    -
    \{\sigma, H\}
    =
    - \frac{\alpha \rho (\nabla_{\mathbf{x}} \cdot \mathbf{u})^3}{\vartheta}
    \quad \text{which implies} \quad
    \dot{S}
    =
    - \int_\Omega \frac{\alpha \rho (\nabla_{\mathbf{x}} \cdot \mathbf{u})^3}{\vartheta} \mathsf{d}^d \mathbf{x} \,.
\end{equation}
Similarly, letting $F = \int \psi_{\mathbf{m}} \cdot \mathbf{m} \mathsf{d}^d \mathbf{x}$, where $\psi_{\mathbf{m}}$ is arbitrary, yields the contribution to the momentum equation,
\begin{equation}
    \int \psi_m \cdot \dot{\mathbf{m}} \mathsf{d}^d \mathbf{x}
    =
    \int_\Omega \alpha \rho (\nabla_{\mathbf{x}} \cdot \mathbf{u})^2 \nabla_{\mathbf{x}} \cdot \psi_m \mathsf{d}^d \mathbf{x} 
    \quad \text{which implies} \quad
    \partial_t \mathbf{m}
    - \{\mathbf{m}, H\}
    =
    - \alpha \nabla_{\mathbf{x}} (\rho (\nabla_{\mathbf{x}} \cdot \mathbf{u})^2) \,.
\end{equation}
Hence, we recover the desired IGR correction to the momentum equation, while entropy production is slightly more transparent:  $\dot{S} > 0$ when $\nabla_{\mathbf{x}} \cdot \mathbf{u} < 0$ and $\dot{S} < 0$ when $\nabla_{\mathbf{x}} \cdot \mathbf{u} > 0$. That we can guarantee thermodynamic entropy production at the shock fronts ensures physically plausible behavior in the shock front. Diverging flow regions are transient in nature, and negative divergences of large magnitude are rare. This should ensure (1) positive global entropy production, and (2) minimal impact of the physically incorrect negative entropy production. We now examine the effect of this dissipative bracket on the kinematic momentum $\rho\mathbf{u}$ and on the total energy.

We now explicitly verify that this symmetric dissipative bracket contributes the correct non-conservative force in the evolution equation for the kinematic momentum, $\rho \mathbf{u}$. The functional chain rule implies, see appendix \ref{appendix:momentum_coord_change}, that if $F[\mathbf{m}, \rho, \sigma] = \tilde{F}[\rho \mathbf{u}, \rho, \sigma]$, then
\begin{multline}
    (\tilde{F},\tilde{G})
    =
    - \int_\Omega \frac{\alpha \rho \nabla_{\mathbf{x}} \cdot \mathbf{u}}{\vartheta} 
    \left[ \vartheta \nabla_{\mathbf{x}} \cdot \left[ \mathcal{L}_\alpha^{-1}(\rho) \left( \rho \frac{\delta \tilde{F}}{\delta (\rho \mathbf{u})} \right) \right] - \frac{\delta \tilde{F}}{\delta \sigma} \nabla_{\mathbf{x}} \cdot \mathbf{u} \right] \\
    \times \left[ \vartheta \nabla_{\mathbf{x}} \cdot \left[ \mathcal{L}_\alpha^{-1}(\rho) \left( \rho \frac{\delta \tilde{G}}{\delta (\rho \mathbf{u})} \right) \right] - \frac{\delta \tilde{G}}{\delta \sigma} \nabla_{\mathbf{x}} \cdot \mathbf{u} \right] \mathsf{d}^d \mathbf{x} \,.
\end{multline}
Hence, we find that
\begin{multline}
    \partial_t (\rho \mathbf{u}) - \{ \rho \mathbf{u}, H \}
    =
    - \alpha \mathcal{L}_\alpha^{-1}(\rho) \nabla_{\mathbf{x}} \cdot ( \rho ( \nabla_{\mathbf{x}} \cdot u)^2) 
    =
    - \nabla_{\mathbf{x}} \Sigma_D \,, \\
    \text{where} \quad
    \rho^{-1} \Sigma_D - \alpha \nabla_{\mathbf{x}} \cdot ( \rho^{-1} \nabla_{\mathbf{x}} \Sigma_D) = \alpha ( \nabla_{\mathbf{x}} \cdot u)^2 \,,
\end{multline}
using the commutation relations established in appendix \Cref{appendix:commutation_relations}. Therefore, we indeed recover the desired IGR contribution to the momentum equation.

Finally, we find the contribution that the symmetric bracket makes to local energy transport. Since $\partial \varepsilon/ \partial s = \vartheta$, we find that
\begin{equation}
    (e,S)
    =
    - \alpha \rho (\nabla_\bx \cdot \bu)^3 \,,
    \quad \text{which yields} \quad
    \partial_t e - \{e, H\}
    = - \alpha \rho (\nabla_{\mathbf{x}} \cdot \mathbf{u})^3 \,.
\end{equation}
The Hamiltonian part of the internal energy transport is given by ~\Cref{eq:internal_energy_density_transport}. The local transport of kinetic energy is less convenient to obtain directly from the dissipative bracket. However, recall that an expression for kinetic energy transport when the non-conservative term is included in the entropic pressure was already obtained in ~\Cref{eq:local-KE-transport}. Hence, summing the kinetic and internal energy transport equations, we find that
\begin{equation}
    \partial_t E + \nabla_{\mathbf{x}} \cdot \left( E \mathbf{u} + \mathbb{\Sigma} \cdot \mathbf{u} + \rho (\nabla_{\mathbf{x}} \cdot \mathbf{u}) \bm{\xi} \right) = 0 \,,
\end{equation}
where $\mathbb{\Sigma}$ solves ~\Cref{eq:matrix_entropic_pressure_non_cons}. 

%\subsubsection{Summary of the Hamiltonian IGR model}\label{subsec:summary-metriplectic-IGR}

To summarize, the metriplectic extension of the HRE model, which we call Hamiltonian IGR, in Lie-Poisson Hamiltonian coordinates is given by
\begin{equation}
\left\{
\begin{aligned}
    &\partial_t
    \begin{bmatrix}
        \rho \\
        \mathbf{m} \\
        \sigma
    \end{bmatrix}
    +
    \nabla_{\bx} \cdot
    \begin{bmatrix}
        \rho \bu \\
        \begin{aligned}
        &\mathbf{m} \otimes \mathbf{u} + \alpha \rho ( \nabla_{\mathbf{x}} \cdot \mathbf{u}) (\nabla_{\mathbf{x}} \mathbf{u})^T \\
        &+ (p - \alpha \rho (\nabla_{\mathbf{x}} \cdot (\rho^{-1} c_S^2 \nabla_\bx \rho))) \mathbb{I} + \alpha \rho^{-1} c_S^2 \nabla_{\mathbf{x}} \rho \otimes \nabla_{\mathbf{x}} \rho 
        \end{aligned} \\
        \sigma \bu
    \end{bmatrix}
    =
    \begin{bmatrix}
        0 \\
        - \alpha \nabla_{\bx}(\rho(\nabla_\bx \cdot \bu)^2) \\
        - \alpha \vartheta^{-1} \rho (\nabla_{\mathbf{x}} \cdot \mathbf{u})^3
    \end{bmatrix} \,, \\
    &\rho \mathbf{u} - \alpha \nabla_{\mathbf{x}}(\rho \nabla_{\mathbf{x}} \cdot \mathbf{u}) = \mathbf{m} \,,
\end{aligned}
\right.
\end{equation}
For clarity, the dissipative parts of the model are written on the right-hand side and the conservative parts on the left-hand side. In conservative form, the Hamiltonian IGR model is given by
\begin{equation}\label{eq:Hamiltonian-IGR}
\left\{
\begin{aligned}
    &\partial_t
    \begin{bmatrix}
        \rho \\
        \rho \bu \\
        E
    \end{bmatrix}
    +
    \nabla_{\bx} \cdot
    \begin{bmatrix}
        \rho \mathbf{u} \\
        \rho \mathbf{u} \otimes \mathbf{u} + \mathbb{\Sigma} \\
        E \mathbf{u} + \mathbb{\Sigma} \cdot \mathbf{u} + \alpha c_s^2 (\nabla_\bx \cdot \bu) \nabla_\bx \rho / \rho
    \end{bmatrix}
    =
    0 \,, \\
    &\begin{aligned}
    \rho^{-1} \mathbb{\Sigma} - \alpha \nabla_{\mathbf{x}} \cdot (\rho^{-1} \nabla_{\mathbf{x}} \cdot \mathbb{\Sigma} )\mathbb{I}
    &=
    \left( \rho^{-1} p - \alpha \nabla_{\mathbf{x}} \cdot (\rho^{-1} c_s^2 \nabla_\bx \rho) \right) \mathbb{I}  + \alpha (\mathrm{tr}^2(\nabla_\bx \bu) + \mathrm{tr}((\nabla_{\mathbf{x}} \mathbf{u})^2))\mathbb{I} \\
    &+ 
    \alpha
    \left[
    \left( \left( \frac{ \rho^2 \left. \partial_\rho^3 (\rho \varepsilon) \right|_s }{2} - \frac{ c_s^2}{2} \right) \left|\frac{\nabla_\bx \rho}{\rho} \right|^2 \right) \mathbb{I} + c_s^2 \frac{\nabla_\bx \rho}{\rho} \otimes \frac{\nabla_\bx \rho}{\rho}
    \right]
    \end{aligned} \,.
\end{aligned}
\right.
\end{equation}
This is identical to the thermodynamic HRE model other than the additional term in the entropic pressure equation. Hamiltonian IGR extends thermodynamic HRE such that it
\begin{itemize}
    \item formally recovers compressible Euler in the limit $\alpha \to 0$,
    \item yields entropy production whose sign matches that of the divergence of the flow,
    \item recovers the non-conservative IGR contribution to the momentum equation,
    \item maintains locally-conservative mass, energy, and momentum transport equations. 
\end{itemize}

%================================================================
\subsection{Discussion}

The dissipative extension to the Hamiltonian regularization discussed in this section is not uniquely specified, despite being essential for numerical stability and convergence as $\alpha \to 0$. The ideal manner of adding dissipation to the HRE model remains an open question, with the approaches discussed here being heuristic in nature. Adequate dissipation for numerical stability could be accomplished numerically by ensuring a cell-wise entropy inequality in a grid-based discretization \cite{CHAN2025114380}. At the continuous level, a likely culprit for the instability in the HIGR model observed in ~\Cref{subsec:1d_numerics} are the higher-order derivatives, particularly of the mass density. This suggests a possible road map for future inquiry. In pressureless IGR, we saw that the dissipation part of the regularization dissipates the generalized kinetic energy with the rate proportional to the divergence of the velocity field, see ~\Cref{eq:pressureless_IGR_kinetic_energy_transport}. While the Hamiltonian regularization alone might allow the $H(\mathrm{div})$ seminorm portion of the generalized kinetic energy to grow locally unbounded, the dissipative regularization prevents singularities from forming because the rate of dissipation is proportional to precisely the problematic term. A comparable dissipation process in the internal energy acting on the $H^1$-type regularization added to the potential energy in thermodynamic HRE is a promising candidate for future studies. Such a regularization would amount to localized mass diffusion, possibly resembling the parabolic regularizations discussed in \cite{guermond2014viscous}, or the Brenner-Navier-Stokes model \cite{BRENNER200511, BRENNER200560, BRENNER2006190, BRENNER201267} (which was shown to be a special case of the class of regularizations in \cite{guermond2014viscous}). Such regularizations likewise may be obtained using the metriplectic formalism \cite{zaidni2025metriplectic}, which greatly streamlines their derivation and analysis. 

%================================================================
\section{Conclusion}
In this paper, we considered five regularizations of the compressible Euler equations, two of which were previously introduced in the literature (pressureless IGR \cite{cao2023information} and thermodynamic IGR \cite{wilfong2025simulatingmanyenginespacecraftexceeding}) and three of which were introduced here (presureless HRE, thermodynamic HRE, and Hamiltonian IGR). By extracting out a conservative (Hamiltonian) component of the presureless IGR, this yielded the presureless HRE. By incorporating thermodynamics with an additional capillary energy, this led us to the thermodynamic HRE model, the first multi-dimensional, non-dispersive Hamiltonian regularization of the compressible Euler equations incorporating general thermodynamics. On the other hand, when incorporating thermodynamics into the IGR model, there is a non-desirable entropy production rate whose sign does not necessarily agree with the sign of the divergence of the flow. The interaction between the regularization and thermodynamics in IGR is therefore somewhat opaque, and forces with a conservative character nevertheless lead to dissipation in their coupling with internal energy. To obtain the regularizing effects of the IGR while transparently exposing the coupling between kinematics and thermodynamics, we extended the thermodynamic HRE model by including the non-conservative IGR term as a dissipative force, naturally expressed in terms of the metriplectic 4-bracket framework. This yielded the Hamiltonian IGR model, which is both a dissipative extension of the thermodynamic HRE model and an entropy-consistent thermodynamic extension of the IGR model. 

One of the objectives here was to construct a structured regularization of the compressible Euler equations in conservative variables, which is amenable to high-performance simulation. However, as seen in ~\Cref{subsec:1d_numerics}, the model suffers from drawbacks which make IGR preferable for practical purposes: HIGR adds significant model complexity, contains higher-order derivatives which lead to numerical instability, and exhibits non-physical oscillations associated with the interaction of shock and contact discontinuities in the HIGR model. This work demonstrates that even the notion of total energy in IGR is nuanced: its standard role as the generator of a Hamiltonian vector field, or equivalently as the conserved quantity associated with time-translation symmetry, is not applicable for dynamics defined by parallel transport with respect to the dual connection on a Hessian manifold. This work suggests that, to fully understand IGR, one must view it in terms of its intrinsic structure as a flow on a Hessian manifold, rather than relying on standard notions from Hamiltonian or Lagrangian mechanics. 

Indeed, the original derivation of IGR \cite{cao2023information} uniquely specified both the conservative and dissipative parts of the regularization, whereas the this work only provides a satisfactory derivation of the conservative, Hamiltonian part. The HRE model is rigidly and uniquely specified by the requirements that we recover compressible Euler as $\alpha \to 0$, that we recover the conservative part of IGR in the pressureless case, see ~\Cref{sec:pressureless-IGR}, and that the inclusion of thermodynamic pressure not induce unphysical dispersion, see ~\Cref{sec:dispersion-free-extension}. On the other hand, the dissipative extension of HRE to obtain HIGR only sought to recover the dissipative part of IGR, see ~\Cref{sec:full-IGR-dissipative-extension}. Whether other dissipative processes are needed to stabilize HRE and maintain regularity of the solution profile is not addressed by the theory developed in this work. Further development is needed to identify the ideal dissipative process to couple with the HRE model. As seen in ~\Cref{subsec:1d_numerics} and further discussed in ~\Cref{appendix:ablation_study}, the inclusion of dissipative terms is necessary to obtain physically-realistic solution profiles. Hence, while this work provides a preliminary investigation into the inclusion of thermodynamics in IGR, a perfectly satisfactory framework remains elusive. 

\begin{comment}
Although the final thermodynamic HRE model has global regularization terms that require an implicit solve, the introduced auxiliary variables remain semi-explicit index-1 in a DAE sense and require inversion of a symmetric elliptic operator for which many fast solvers exist. Moreover, the regularization constant is such that one can likely get away with a few preconditioned CG iterations to solve the regularization \cite{wilfong2025simulatingmanyenginespacecraftexceeding} rather than more advanced, e.g., multilevel methods, making the additional computational cost of the regularization reasonable. In a companion paper, we plan to compare numerical methods derived from these models and in particular, to demonstrate the efficacy of the thermodynamic consistency of the Hamiltonian IGR model to regularize shocks.
\end{comment}

\section*{Acknowledgments}
\sloppy WB was supported by the Laboratory Directed Research and Development program of Los Alamos National Laboratory under project number 20251151PRD1. BSS was supported by  the U.S. Department of Energy Office of Advanced Scientific Computing under contract No. 89233218CNA000001.
FS acknowledges support from the Air Force Office of Scientific Research under award number FA9550-23-1-0668 (Information Geometric Regularization for Simulation and Optimization of Supersonic Flow) and the Predictive Science Academic Alliance Program (PSAAP Award DE-NA0004261 - “The Center for Information Geometric Mechanics and Optimization
(CIGMO)”) managed by the NNSA (National
Nuclear Security Administration) Office of Advanced Simulation.
Los Alamos Laboratory Report LA-UR-25-32021. 

\section*{Data Availability Statement}
Data generated in this study is available from the authors upon reasonable request.

\section*{Declarations}
The authors declare no competing interests.

%================================================================
%\printbibliography
\bibliographystyle{plainnat}
\bibliography{references}

@ARTICLE{CoMo_2020,
   author={{B}. Coquinot and {P}. {J}. Morrison},
   title={{A} general metriplectic framework with application to dissipative extended magnetohydrodynamics},
   journal = {{J}. Plasma Phys.},
   volume = {86},
   pages = {835860302},
   year = {2020},
   doi={https://doi.org/10.1017/S0022377820000392}
}

@article{ShMo1950,
author = {Jack {S}herman and {W}inifred {J}. {M}orrison},
title = {{{A}djustment of an {I}nverse {M}atrix {C}orresponding to a {C}hange in {O}ne {E}lement of a {G}iven {M}atrix}},
volume = {21},
journal = {The Annals of Mathematical Statistics},
number = {1},
pages = {124 -- 127},
year = {1950},
doi = {10.1214/aoms/1177729893},
}

@book{toro2009riemann,
  title={Riemann Solvers and Numerical Methods for Fluid Dynamics: {A} Practical Introduction},
  author={Toro, Eleuterio {F}.},
  edition={3rd},
  year={2009},
  publisher={Springer Berlin, Heidelberg},
  doi={10.1007/b79761}
}

@book{LeVeque_2002, 
    place={Cambridge}, 
    series={Cambridge Texts in Applied Mathematics}, 
    title={Finite Volume Methods for Hyperbolic Problems}, 
    publisher={Cambridge University Press}, 
    author={LeVeque, Randall {J}.}, 
    year={2002}, 
    collection={Cambridge Texts in Applied Mathematics}
}

@book{courant1948supersonic,
  title={Supersonic Flow and Shock Waves},
  author={Courant, Richard and Friedrichs, Kurt {O}},
  year={1948},
  publisher={Interscience Publishers},
  note={Reprinted by Springer-Verlag, 1976}
}

@book{DG2000,
editor={Bernardo Cockburn and George {E}. Karniadakis and Chi-Wang Shu},
Title={Discontinuous Galerkin Methods: Theory, Computation and Applications},
Year={2000},
PUBLISHER={Springer Berlin, Heidelberg},
series = {Lecture Notes in Computational Science and Engineering},
doi = {10.1007/978-3-642-59721-3}
}

@ARTICLE{Burby_2015,
   author={{J}. {W}. Burby and {A}. {J}. Brizard and {P}. {J}. Morrison and {H}. Qin},
   title={{Hamiltonian gyrokinetic Vlasov–Maxwell system}},
   journal = {Phys. Lett. {A}},
   volume = {379},
   pages = {2073-2077},
   year = {2015},
   doi={10.1016/j.physleta.2015.06.051}
}

@ARTICLE{Hirvijoki_2022,
   author={{E}. Hirvijoki and {J}. {W}. Burby and {A}. {J}. Brizard},
   title={{Metriplectic foundations of gyrokinetic Vlasov–Maxwell–Landau theory}},
   journal = {Phys. Plasmas},
   volume = {29},
   pages = {060701},
   year = {2022},
   doi={10.1063/5.0091727}
}

@Article{TrBuSo2025,
author = {Tran, {B}. {K}. and {B}urby, {J}. and {S}outhworth, {B}. {S}.},
title = {{A} {G}eometric {P}erspective on {K}inetic {M}atter-{R}adiation {I}nteraction and {M}oment {S}ystems},
journal = {Geometric Mechanics},
year = {2025},
volume = {2},
number = {1},
pages = {107-158},
doi = {10.1142/S2972458925500054}
}

@Article{khesin_2021,
author = {Khesin, {B}. and {M}isiolek, {G}. and {M}odin, {K}.},
title = {Geometric hydrodynamics and infinite-dimensional {N}ewton's equations},
journal = {Bull. Amer. Math. Soc.},
year = {2021},
volume = {58},
pages = {377-442},
doi = {10.1090/bull/1728}
}

@article{pavelka2016,
title = {{A} hierarchy of {P}oisson brackets in non-equilibrium thermodynamics},
journal = {Physica {D}: Nonlinear Phenomena},
volume = {335},
pages = {54-69},
year = {2016},
issn = {0167-2789},
doi = {10.1016/j.physd.2016.06.011},
author = {Michal {P}avelka and {V}áclav {K}lika and {O}ğul {E}sen and {M}iroslav {G}rmela}
}

@article{Shu.1988fyej, 
year = {1988}, 
title = {{Total-{V}ariation-{D}iminishing {T}ime {D}iscretizations}}, 
author = {Shu, {C}hi-{W}ang}, 
journal = {SIAM Journal on Scientific and Statistical Computing}, 
issn = {0196-5204}, 
doi = {10.1137/0909073}, 
pages = {1073--1084}, 
number = {6}, 
volume = {9}, 
}

@article{Shu.1988, 
year = {1988}, 
title = {{Efficient implementation of essentially non-oscillatory shock-capturing schemes}}, 
author = {Shu, {C}hi-{W}ang and {O}sher, {S}tanley}, 
journal = {Journal of Computational Physics}, 
issn = {0021-9991}, 
doi = {10.1016/0021-9991(88)90177-5}, 
pages = {439--471}, 
number = {2}, 
volume = {77}, 
}

@article{Gottlieb.1996, 
year = {1996}, 
title = {{Total variation diminishing {R}unge-{K}utta schemes}}, 
author = {Gottlieb, {S}igal and {S}hu, {C}hi-{W}ang}, 
journal = {Mathematics of Computation}, 
issn = {0025-5718}, 
doi = {10.1090/s0025-5718-98-00913-2}, 
pages = {73--85}, 
number = {221}, 
volume = {67}, 
}

@book{MaRa1999,
Author={{J}. {E}. Marsden and {T}. {S}. Ratiu},
Title={Introduction to Mechanics and Symmetry},
Year={1999},
PUBLISHER={Springer New York, NY},
Edition={Second},
doi = {10.1007/978-0-387-21792-5}
}

@article{VaMa2005, 
title={The {L}ie-{P}oisson {S}tructure of the {E}uler {E}quations of an {I}deal {F}luid}, 
volume={2}, 
ISSN={1548-159X}, 
number={4}, 
journal={Dynamics of Partial Differential Equations}, 
publisher={International Press}, 
author={Vasylkevych, {S}ergiy and {M}arsden, {J}errold {E}.}, 
year={2005}, 
pages={281–300}
}

@article{guelmame2022hamiltonian,
  title={Hamiltonian regularisation of the unidimensional barotropic Euler equations},
  author={Guelmame, Billel and Clamond, Didier and Junca, St{\'e}phane},
  journal={Nonlinear Analysis: Real World Applications},
  volume={64},
  pages={103455},
  year={2022},
  publisher={Elsevier}
}

@article{cao2023information,
  title={Information geometric regularization of the barotropic Euler equation},
  author={Cao, Ruijia and Sch{\"a}fer, Florian},
  journal={arXiv preprint arXiv:2308.14127},
  year={2023}
}

@article{cao2024information,
  title={Information geometric regularization of unidimensional pressureless Euler equations yields global strong solutions},
  author={Cao, Ruijia and Sch{\"a}fer, Florian},
  journal={arXiv preprint arXiv:2411.15121},
  year={2024}
}

@article{suzuki2020generic,
  title={{A} GENERIC formalism for Korteweg-type fluids: {I}. {A} comparison with classical theory},
  author={Suzuki, Yukihito},
  journal={Fluid Dynamics Research},
  volume={52},
  number={1},
  pages={015516},
  year={2020},
  publisher={IOP Publishing}
}

@article{morrison1980noncanonical,
  title={Noncanonical Hamiltonian density formulation of hydrodynamics and ideal magnetohydrodynamics},
  author={Morrison, Philip {J} and Greene, John {M}},
  journal={Physical Review Letters},
  volume={45},
  number={10},
  pages={790},
  year={1980},
  publisher={APS}
}

@article{morrison1998hamiltonian,
  title={Hamiltonian description of the ideal fluid},
  author={Morrison, Philip {J}},
  journal={Reviews of modern physics},
  volume={70},
  number={2},
  pages={467},
  year={1998},
  publisher={APS}
}

@inproceedings{morrison1982poisson,
  title={Poisson brackets for fluids and plasmas},
  author={Morrison, Philip {J}},
  booktitle={AIP Conf. Proc},
  volume={88},
  number={1},
  pages={13--46},
  year={1982}
}

@article{holm1985nonlinear,
  title={Nonlinear stability of fluid and plasma equilibria},
  author={Holm, Darryl {D} and Marsden, Jerrold {E} and Ratiu, Tudor and Weinstein, Alan},
  journal={Physics reports},
  volume={123},
  number={1-2},
  pages={1--116},
  year={1985},
  publisher={Elsevier}
}

@inproceedings{arnold1966geometrie,
  title={Sur la g{\'e}om{\'{e}}trie diff{\'e}rentielle des groupes de Lie de dimension infinie et ses applications {\`a} l'hydrodynamique des fluides parfaits},
  author={Arnold, Vladimir},
  booktitle={Annales de l'institut Fourier},
  volume={16},
  number={1},
  pages={319--361},
  year={1966}
}

@article{arnold1966priori,
  title={An a priori estimate in the theory of hydrodynamic stability},
  author={Arnold, Vladimir Igorevich},
  journal={Izv. Vyssh. Uchebn. Zaved. Mat.[Sov. Math. {J}.]},
  volume={5},
  pages={3},
  year={1966}
}

@incollection{mackay2020stability,
  title={Stability of equilibria of Hamiltonian systems},
  author={MacKay, RS},
  booktitle={Hamiltonian Dynamical Systems},
  pages={137--153},
  year={2020},
  publisher={CRC Press}
}

@inproceedings{kreuin1950generalization,
  title={{A} generalization of some investigations on linear differential equations with periodic coefficients},
  author={Kreĭn, MG},
  booktitle={Dokl. Akad. Nauk SSSR {A}},
  volume={73},
  pages={445},
  year={1950}
}

@article{zaidni2024thermodynamically,
  title={Thermodynamically consistent {C}ahn--{H}illiard--{N}avier--{S}tokes equations using the metriplectic dynamics formalism},
  author={Zaidni, Azeddine and Morrison, Philip {J} and Benjelloun, Saad},
  journal={Physica {D}: Nonlinear Phenomena},
  volume={468},
  pages={134303},
  year={2024},
  publisher={Elsevier}
}

@article{PhysRevE.109.045202,
  title = {Inclusive curvaturelike framework for describing dissipation: {M}etriplectic 4-bracket dynamics},
  author = {Morrison, {P}.{J}. and Updike, {M}.{H}.},
  journal = {Phys. Rev. {E}},
  volume = {109},
  pages = {045202},
  numpages = {22},
  year = {2024},
  publisher = {American Physical Society},
  doi = {10.1103/PhysRevE.109.045202},
  url = {https://link.aps.org/doi/10.1103/PhysRevE.109.045202}
}

@article{morrison1984bracket,
  title={Bracket formulation for irreversible classical fields},
  author={Morrison, Philip {J}},
  journal={Physics Letters {A}},
  volume={100},
  number={8},
  pages={423--427},
  year={1984},
  publisher={Elsevier}
}

@article{MORRISON1986410,
title = "{A paradigm for joined {H}amiltonian and dissipative systems}",
journal = {Physica {D}},
volume = {18},
pages = {410--419},
year = {1986},
issn = {0167-2789},
doi = {https://doi.org/10.1016/0167-2789(86)90209-5},
url = {https://www.sciencedirect.com/science/article/pii/0167278986902095},
author = {{P}.{J}. {M}orrison}
}

@book{de2013non,
  title={Non-equilibrium thermodynamics},
  author={De Groot, Sybren Ruurds and Mazur, Peter},
  year={2013},
  publisher={Courier Corporation}
}

@article{bressan2025metriplectic,
  title={Metriplectic relaxation to equilibria},
  author={Bressan, {C} and Kraus, {M} and Maj, {O} and Morrison, PJ},
  journal={arXiv preprint arXiv:2506.09787},
  year={2025}
}

@article{zaidni2025metriplectic,
  title={Metriplectic four-bracket algorithm for constructing thermodynamically consistent dynamical systems},
  author={Zaidni, Azeddine and Morrison, Philip {J}},
  journal={Physical Review {E}},
  volume={112},
  number={2},
  pages={025101},
  year={2025},
  publisher={APS}
}

@misc{wilfong2025simulatingmanyenginespacecraftexceeding,
      title={Simulating many-engine spacecraft: Exceeding 1 quadrillion degrees of freedom via information geometric regularization}, 
      author={Benjamin Wilfong and Anand Radhakrishnan and Henry Le Berre and Daniel {J}. Vickers and Tanush Prathi and Nikolaos Tselepidis and Benedikt Dorschner and Reuben Budiardja and Brian Cornille and Stephen Abbott and Florian Schäfer and Spencer {H}. Bryngelson},
      year={2025},
      eprint={2505.07392},
      archivePrefix={arXiv},
      primaryClass={physics.comp-ph},
      url={https://arxiv.org/abs/2505.07392}, 
}

@book{abraham2008foundations,
  title={Foundations of mechanics},
  author={Abraham, Ralph and Marsden, Jerrold {E}},
  number={364},
  year={2008},
  publisher={American Mathematical Soc.}
}

@book{marsden2013introduction,
  title={Introduction to mechanics and symmetry: a basic exposition of classical mechanical systems},
  author={Marsden, Jerrold {E} and Ratiu, Tudor {S}},
  volume={17},
  year={2013},
  publisher={Springer Science \& Business Media}
}

@book{arnol2013mathematical,
  title={Mathematical methods of classical mechanics},
  author={Arnold, Vladimir Igorevich},
  volume={60},
  year={2013},
  publisher={Springer Science \& Business Media}
}

@article{kaufman1984dissipative,
  title={Dissipative {H}amiltonian systems: {A} unifying principle},
  author={Kaufman, Allan {N}},
  journal={Physics Letters {A}},
  volume={100},
  number={8},
  pages={419--422},
  year={1984},
  publisher={Elsevier}
}

@article{grmela1997dynamics,
  title={Dynamics and thermodynamics of complex fluids. {I}. Development of a general formalism},
  author={Grmela, Miroslav and {\"O}ttinger, Hans Christian},
  journal={Physical Review {E}},
  volume={56},
  number={6},
  pages={6620},
  year={1997},
  publisher={APS}
}

@article{ottinger1997dynamics,
  title={Dynamics and thermodynamics of complex fluids. II. Illustrations of a general formalism},
  author={{\"O}ttinger, Hans Christian and Grmela, Miroslav},
  journal={Physical Review {E}},
  volume={56},
  number={6},
  pages={6633},
  year={1997},
  publisher={APS}
}

@article{grmela1984bracket,
  title={Bracket formulation of dissipative fluid mechanics equations},
  author={Grmela, Miroslav},
  journal={Physics Letters {A}},
  volume={102},
  number={8},
  pages={355--358},
  year={1984},
  publisher={Elsevier}
}

@article{kouranbaeva1999camassa,
  title={The {C}amassa--{H}olm equation as a geodesic flow on the diffeomorphism group},
  author={Kouranbaeva, Shinar},
  journal={Journal of Mathematical Physics},
  volume={40},
  number={2},
  pages={857--868},
  year={1999},
  publisher={American Institute of Physics}
}

@article{misiolek1998shallow,
  title={{A} shallow water equation as a geodesic flow on the {B}ott-{V}irasoro group},
  author={Misio{\l}ek, Gerard},
  journal={Journal of Geometry and Physics},
  volume={24},
  number={3},
  pages={203--208},
  year={1998},
  publisher={Elsevier}
}

@article{holm1998euler,
  title={The {E}uler--{P}oincar{\'e} equations and semidirect products with applications to continuum theories},
  author={Holm, Darryl {D} and Marsden, Jerrold {E} and Ratiu, Tudor {S}},
  journal={Advances in Mathematics},
  volume={137},
  number={1},
  pages={1--81},
  year={1998},
  publisher={Elsevier}
}

@article{khesin2003euler,
  title={Euler equations on homogeneous spaces and {V}irasoro orbits},
  author={Khesin, Boris and Misio{\l}ek, Gerard},
  journal={Advances in Mathematics},
  volume={176},
  number={1},
  pages={116--144},
  year={2003},
  publisher={Elsevier}
}

@article{bhat2005lagrangian,
  title={Lagrangian averaging for compressible fluids},
  author={Bhat, Harish {S} and Fetecau, Razvan {C} and Marsden, Jerrold {E} and Mohseni, Kamran and West, Matthew},
  journal={Multiscale Modeling \& Simulation},
  volume={3},
  number={4},
  pages={818--837},
  year={2005},
  publisher={SIAM}
}

@article{constantin2003geodesic,
  title={Geodesic flow on the diffeomorphism group of the circle},
  author={Constantin, Adrian and Kolev, Boris},
  journal={Commentarii Mathematici Helvetici},
  volume={78},
  number={4},
  pages={787--804},
  year={2003},
  publisher={Springer}
}

@article{lenells2008hunter,
  title={The {H}unter--{S}axton equation: a geometric approach},
  author={Lenells, Jonatan},
  journal={SIAM journal on mathematical analysis},
  volume={40},
  number={1},
  pages={266--277},
  year={2008},
  publisher={SIAM}
}

@article{holm1999nonlinear,
  title={{A} nonlinear analysis of the averaged {E}uler equations},
  author={Holm, Darryl {D} and Kouranbaeva, Shinar and Marsden, Jerrold {E} and Ratiu, Tudor and Shkoller, Steve},
  journal={arXiv preprint chao-dyn/9903036},
  year={1999}
}

@article{holm2002lagrangian,
  title={Lagrangian averages, averaged Lagrangians, and the mean effects of fluctuations in fluid dynamics},
  author={Holm, Darryl {D}},
  journal={Chaos: An Interdisciplinary Journal of Nonlinear Science},
  volume={12},
  number={2},
  pages={518--530},
  year={2002},
  publisher={American Institute of Physics}
}

@incollection{holm2005momentum,
  title={Momentum maps and measure-valued solutions (peakons, filaments, and sheets) for the {E}{P}{D}iff equation},
  author={Holm, Darryl {D} and Marsden, Jerrold {E}},
  booktitle={The Breadth of Symplectic and Poisson Geometry: Festschrift in Honor of Alan Weinstein},
  pages={203--235},
  year={2005},
  publisher={Springer}
}

@article{camassa1993integrable,
  title={An integrable shallow water equation with peaked solitons},
  author={Camassa, Roberto and Holm, Darryl {D}},
  journal={Physical review letters},
  volume={71},
  number={11},
  pages={1661},
  year={1993},
  publisher={APS}
}

@article{hunter1991dynamics,
  title={Dynamics of director fields},
  author={Hunter, John {K} and Saxton, Ralph},
  journal={SIAM Journal on Applied Mathematics},
  volume={51},
  number={6},
  pages={1498--1521},
  year={1991},
  publisher={SIAM}
}

@article{leok2017connecting,
  title={Connecting information geometry and geometric mechanics},
  author={Leok, Melvin and Zhang, Jun},
  journal={Entropy},
  volume={19},
  number={10},
  pages={518},
  year={2017},
  publisher={MDPI}
}

@article{levi1916nozione,
  title={Nozione di parallelismo in una variet{\`a} qualunque e conseguente specificazione geometrica della curvatura riemanniana},
  author={Levi-Civita, Memoria di {T}},
  journal={Rendiconti del Circolo Matematico di Palermo (1884-1940)},
  volume={42},
  number={1},
  pages={173--204},
  year={1916},
  publisher={Springer}
}

@book{lee2018introduction,
  title={Introduction to {R}iemannian manifolds},
  author={Lee, John {M}},
  volume={2},
  year={2018},
  publisher={Springer}
}

@book{godlewski2013numerical,
  title={Numerical approximation of hyperbolic systems of conservation laws},
  author={Godlewski, Edwige and Raviart, Pierre-Arnaud},
  volume={118},
  year={2013},
  publisher={Springer Science \& Business Media}
}

@article{pu2018weakly,
  title={Weakly singular shock profiles for a non-dispersive regularization of shallow-water equations},
  author={Pu, Yue and Pego, Robert and Dutykh, Denys and Clamond, Didier},
  journal={arXiv preprint arXiv:1805.06842},
  year={2018}
}

@article{liu2019well,
  title={Well-posedness and derivative blow-up for a dispersionless regularized shallow water system},
  author={Liu, Jian-Guo and Pego, Robert {L} and Pu, Yue},
  journal={Nonlinearity},
  volume={32},
  number={11},
  pages={4346},
  year={2019},
  publisher={IOP Publishing}
}

@article{ebin1970groups,
  title={Groups of diffeomorphisms and the motion of an incompressible fluid},
  author={Ebin, David {G} and Marsden, Jerrold},
  journal={Annals of Mathematics},
  volume={92},
  number={1},
  pages={102--163},
  year={1970},
  publisher={JSTOR}
}

@article{kawai2010assessment,
  title={Assessment of localized artificial diffusivity scheme for large-eddy simulation of compressible turbulent flows},
  author={Kawai, Soshi and Shankar, Santhosh {K} and Lele, Sanjiva {K}},
  journal={Journal of computational physics},
  volume={229},
  number={5},
  pages={1739--1762},
  year={2010},
  publisher={Elsevier}
}

@article{bhagatwala2009modified,
  title={{A} modified artificial viscosity approach for compressible turbulence simulations},
  author={Bhagatwala, Ankit and Lele, Sanjiva {K}},
  journal={Journal of computational physics},
  volume={228},
  number={14},
  pages={4965--4969},
  year={2009},
  publisher={Elsevier}
}

@article{dolejvsi2003some,
  title={On some aspects of the discontinuous {G}alerkin finite element method for conservation laws},
  author={Dolej{\v{s}}{\'\i}, V{\'\i}t and Feistauer, Miloslav and Schwab, Christoph},
  journal={Mathematics and Computers in Simulation},
  volume={61},
  number={3-6},
  pages={333--346},
  year={2003},
  publisher={Elsevier}
}

@article{bruno2022fc,
  title={{F}{C}-based shock-dynamics solver with neural-network localized artificial-viscosity assignment},
  author={Bruno, Oscar {P} and Hesthaven, Jan {S} and Leibovici, Daniel {V}},
  journal={Journal of Computational Physics: {X}},
  volume={15},
  pages={100110},
  year={2022},
  publisher={Elsevier}
}

@article{guermond2011entropy,
  title={Entropy viscosity method for nonlinear conservation laws},
  author={Guermond, Jean-Luc and Pasquetti, Richard and Popov, Bojan},
  journal={Journal of Computational Physics},
  volume={230},
  number={11},
  pages={4248--4267},
  year={2011},
  publisher={Elsevier}
}

@article{barter2010shock,
  title={Shock capturing with {P}{D}{E}-based artificial viscosity for {D}{G}{F}{E}{M}: {P}art {I}. {F}ormulation},
  author={Barter, Garrett {E} and Darmofal, David {L}},
  journal={Journal of Computational Physics},
  volume={229},
  number={5},
  pages={1810--1827},
  year={2010},
  publisher={Elsevier}
}

@article{mani2009suitability,
  title={Suitability of artificial bulk viscosity for large-eddy simulation of turbulent flows with shocks},
  author={Mani, Ali and Larsson, Johan and Moin, Parviz},
  journal={Journal of Computational Physics},
  volume={228},
  number={19},
  pages={7368--7374},
  year={2009},
  publisher={Elsevier}
}

@article{fiorina2007artificial,
  title={An artificial nonlinear diffusivity method for supersonic reacting flows with shocks},
  author={Fiorina, Benoit and Lele, Sanjiva {K}},
  journal={Journal of Computational Physics},
  volume={222},
  number={1},
  pages={246--264},
  year={2007},
  publisher={Elsevier}
}

@article{cook2005hyperviscosity,
  title={Hyperviscosity for shock-turbulence interactions},
  author={Cook, Andrew {W} and Cabot, William {H}},
  journal={Journal of Computational Physics},
  volume={203},
  number={2},
  pages={379--385},
  year={2005},
  publisher={Elsevier}
}

@article{puppo2004numerical,
  title={Numerical entropy production for central schemes},
  author={Puppo, Gabriella},
  journal={SIAM Journal on Scientific Computing},
  volume={25},
  number={4},
  pages={1382--1415},
  year={2004},
  publisher={SIAM}
}

@article{vonneumann1950method,
  title={{A} method for the numerical calculation of hydrodynamic shocks},
  author={VonNeumann, John and Richtmyer, Robert {D}},
  journal={Journal of applied physics},
  volume={21},
  number={3},
  pages={232--237},
  year={1950},
  publisher={American Institute of Physics}
}

@article{chan2025artificial,
  title={An artificial viscosity approach to high order entropy stable discontinuous {G}alerkin methods},
  author={Chan, Jesse},
  journal={arXiv preprint arXiv:2501.16529},
  year={2025}
}

@article{guermond2018second,
  title={Second-order invariant domain preserving approximation of the {E}uler equations using convex limiting},
  author={Guermond, Jean-Luc and Nazarov, Murtazo and Popov, Bojan and Tomas, Ignacio},
  journal={SIAM Journal on Scientific Computing},
  volume={40},
  number={5},
  pages={A3211--A3239},
  year={2018},
  publisher={SIAM}
}

@article{guermond2019invariant,
  title={Invariant domain preserving discretization-independent schemes and convex limiting for hyperbolic systems},
  author={Guermond, Jean-Luc and Popov, Bojan and Tomas, Ignacio},
  journal={Computer Methods in Applied Mechanics and Engineering},
  volume={347},
  pages={143--175},
  year={2019},
  publisher={Elsevier}
}

@article{clayton2022invariant,
  title={Invariant domain-preserving approximations for the {E}uler equations with tabulated equation of state},
  author={Clayton, Bennett and Guermond, Jean-Luc and Popov, Bojan},
  journal={SIAM Journal on Scientific Computing},
  volume={44},
  number={1},
  pages={A444--A470},
  year={2022},
  publisher={SIAM}
}

@article{guelmame2024hamiltonian,
  title={On a Hamiltonian regularization of scalar conservation laws},
  author={Guelmame, Billel},
  journal={arXiv preprint arXiv:2403.02218},
  year={2024}
}

@article{CHAN2025114380,
title = {An artificial viscosity approach to high order entropy stable discontinuous {G}alerkin methods},
journal = {Journal of Computational Physics},
volume = {543},
pages = {114380},
year = {2025},
issn = {0021-9991},
doi = {https://doi.org/10.1016/j.jcp.2025.114380},
url = {https://www.sciencedirect.com/science/article/pii/S002199912500662X},
author = {Jesse Chan}
}

@article{guermond2014viscous,
    author = {Guermond, Jean-Luc and Popov, Bojan},
    title = {Viscous {R}egularization of the {E}uler {E}quations and {E}ntropy {P}rinciples},
    journal = {SIAM Journal on Applied Mathematics},
    volume = {74},
    number = {2},
    pages = {284-305},
    year = {2014},
    doi = {10.1137/120903312},
    URL = {https://doi.org/10.1137/120903312},
    eprint = {https://doi.org/10.1137/120903312}
}

@article{BRENNER200511,
title = {Kinematics of volume transport},
journal = {Physica A: Statistical Mechanics and its Applications},
volume = {349},
number = {1},
pages = {11-59},
year = {2005},
issn = {0378-4371},
doi = {https://doi.org/10.1016/j.physa.2004.10.033},
url = {https://www.sciencedirect.com/science/article/pii/S0378437104013238},
author = {Howard Brenner}
}

@article{BRENNER200560,
title = {Navier–{S}tokes revisited},
journal = {Physica A: Statistical Mechanics and its Applications},
volume = {349},
number = {1},
pages = {60-132},
year = {2005},
issn = {0378-4371},
doi = {https://doi.org/10.1016/j.physa.2004.10.034},
url = {https://www.sciencedirect.com/science/article/pii/S037843710401324X},
author = {Howard Brenner}
}

@article{BRENNER2006190,
title = {Fluid mechanics revisited},
journal = {Physica A: Statistical Mechanics and its Applications},
volume = {370},
number = {2},
pages = {190-224},
year = {2006},
issn = {0378-4371},
doi = {https://doi.org/10.1016/j.physa.2006.03.066},
url = {https://www.sciencedirect.com/science/article/pii/S0378437106007266},
author = {Howard Brenner}
}

@article{BRENNER201267,
title = {Beyond {N}avier–{S}tokes},
journal = {International Journal of Engineering Science},
volume = {54},
pages = {67-98},
year = {2012},
issn = {0020-7225},
doi = {https://doi.org/10.1016/j.ijengsci.2012.01.006},
url = {https://www.sciencedirect.com/science/article/pii/S0020722512000171},
author = {Howard Brenner}
}

%==================================================_==============
\appendix

%================================================================
\section{A brief review of bracket formalisms in mechanics} \label{appendix:bracket_formalisms}

Here, we provide a brief review of the defining qualities of Hamiltonian and metriplectic systems. The purview of Hamiltonian mechanics encompasses the entirety of classical physics, and hence represents an enormous body of knowledge. The following appendix is a minimal specification of such systems. For further reading on general Hamiltonian systems, see \cite{arnol2013mathematical, abraham2008foundations, marsden2013introduction}, and for its specific use in modeling fluid and plasma systems, see \cite{morrison1982poisson, morrison1998hamiltonian}. For further information on metriplectic formalism, see \cite{morrison1984bracket, MORRISON1986410, kaufman1984dissipative}, and for the independently-developed, formally-similar GENERIC formalism, see \cite{grmela1984bracket, ottinger1997dynamics, grmela1997dynamics}. For further information on the metriplectic $4$-bracket formalism, see \cite{PhysRevE.109.045202, zaidni2024thermodynamically, bressan2025metriplectic, zaidni2025metriplectic}. 

%----------------------------
\subsection{The Hamiltonian formalism}

The Hamiltonian formulation of mechanics is prescribed by a Hamiltonian, $H$, a functional of the dynamical fields, and a Poisson bracket, $\{\cdot, \cdot\}$, a bilinear map on functionals of the fields. A Poisson bracket has the following properties: $\forall F,G,H$ and $\forall a,b \in \mathbb{R}$,
\begin{equation}
\begin{aligned}
    &\{F, aG + bH\} = a \{F,G\} + b \{F, H\} \,, \\
    &\{F, G\} = - \{G, F\} \,,
    \\
    &\{ F, \{G, H\} \} + \{G, \{H, F \} \} + \{ H, \{F, G\} \} = 0 \,, \\
    &\{F, G H \} = G \{F, H \} + \{F, G \} H \,.
\end{aligned}
\end{equation}
For any observable, $F$, its evolution is prescribed by $\dot{F} = \{F, H\}$. Two kinds of conservation laws are conveniently obtained from this framework. The first class of conservation laws are functionals, $I$, which Poisson commute with the Hamiltonian; because of the antisymmetry of the bracket, we find that
\begin{equation}
    \dot{I} 
    =
    \{I, H\}
    =
    \{H, I\}
    =
    - \{H, I\} = 0 \,.
\end{equation}
Hamiltonian systems in non-canonical coordinates yield a second class of conservation laws known as Casimir invariants: a Casimir invariant is a degeneracy of the Poisson bracket, i.e., a functional $C$ such that $\{F, C\} = 0$ for all $F$.

%----------------------------
\subsection{The metriplectic formalism: a dissipative extension of Hamiltonian mechanics} \label{appendix:metriplectic}

The metriplectic formalism is a dissipative extension of Hamiltonian mechanics \cite{MORRISON1986410}. The $4$-bracket approach \cite{PhysRevE.109.045202} is a systematic way to construct metriplectic models. The evolution equations are prescribed by a Hamiltonian $H$ which is a functional of the dynamical fields, a Poisson bracket $\{\cdot, \cdot\}$, an entropy $S$ which is a Casimir invariant of the Poisson bracket which generates the dissipative dynamics, and a metriplectic $4$-bracket $(\cdot, \cdot; \cdot, \cdot)$ defined as follows. The metriplectic $4$-bracket is a $4$-linear map on the algebra of functionals with the following properties: for all $F,K,G,N$,
\begin{equation}\label{eq:4-bracket-symmetries}
\begin{aligned}
    (F, K; G, N) &= - (K, F; G, N) \,,
    \\
    (F, K; G, N) &= (G, N; F, K) \,,
    \\
    (FH, K; G, N) &= F(H, K; G, N) + (F, K; G, N) H \,,
    \\
    (F, G; F, G) &\geq 0 \,.
\end{aligned}
\end{equation}
For any observable, $F$, its evolution is prescribed by $\dot{F} = \{F, H\} + (F, H; S, H)$. 

Thermodynamic consistency is guaranteed in the metriplectic formalism by the following properties: (i) the entropy is a Casimir invariant of the Poisson bracket, $\{F,S\} = 0$ $\forall F$; (ii) antisymmetry of the Poisson bracket: $\{F,G\} = - \{G,F\}$; (iii) antisymmetry of the $4$-bracket, $(F, K; G,N) = -(K, F; G, N)$; (iv) finally, non-negative entropy production is ensured by the semi-definiteness of the bracket: $(S, H; S, H) \geq 0$. Together, these ensure thermodynamic consistency:
\begin{equation}
    \dot{H} = \{H, H\} + (H, H; S, H) = 0 \,,
    \quad \text{and} \quad
    \dot{S} = \{S, H\} + (S, H; S, H) = (S, H; S, H) \geq 0. 
\end{equation}

In fluids and kinetic models, it is frequently convenient to design the metriplectic structure by prescribing by a $4$-bracket constructed using the Kulkarni--Nomizu product (see e.g. \cite{PhysRevE.109.045202, zaidni2024thermodynamically, bressan2025metriplectic, zaidni2025metriplectic}). To do so, given two symmetric bilinear maps, $A(F,G)$ and $B(F,G)$, define the Kulkarni--Nomizu product
\begin{equation}
\begin{aligned}
    (A\varowedge B)(F, K, G, N)
    =
    A(F, G) B(K, N) 
    &- A(F, N) B(G, K) \\
    &+ B(F, G) A(K, N) 
    - B(F, N) A(G, K) \,.
\end{aligned}
\end{equation}
From this product, one can define a $4$-bracket satisfying the properties \eqref{eq:4-bracket-symmetries} by
\begin{equation}
    (F, K; G, N) 
    =
    \int_\Omega (A\varowedge B)(F, K, G, N) \mathsf{d}^d \mathbf{x} \,.
\end{equation}
A rationale for choosing the operators $A$ and $B$ comes from a closer examination of the implied dissipative flow:
\begin{equation} \label{eq:dissipative_flow}
\begin{aligned}
    (F, S)_H
    &=
    \int_\Omega
    (A\varowedge B)(F, K, G, N)
    \mathsf{d} x 
    =
    \int_\Omega
    \left(
    - A(F, H) B(S, H) 
    + B(F, S) A(H, H) 
    \right)
    \mathsf{d}^d \mathbf{x} \,.
\end{aligned}
\end{equation}
By letting $B = F_\sigma G_\sigma$, we find that $\dot{S} = A(H,H)$ is the entropy production rate while $A(F,H)$ gives rise to the reciprocal couplings which ensure energy conservation. Hence, given the desired entropy production rate of a model, then it is straightforward to determine the appropriate $A$ operator. This rationale for finding the metriplectic $4$-bracket is generally applicable to many compressible flow models, and directly connects with standard arguments from non-equilibrium thermodynamics \cite{de2013non}, e.g., the force and flux pairs from Onsager reciprocity \cite{CoMo_2020}.

%----------------------------
\subsection{The energy-Casimir method: linear stability in noncanonical Hamiltonian systems} \label{appendix:energy_casimir}

We briefly review some relevant elements of Hamiltonian perturbation theory, in the context of finite dimensional systems for simplicity of presentation. Consider a non-canonical Hamiltonian system with phase space $\mathbb{R}^m$, Hamiltonian $H: \mathbb{R}^m \to \mathbb{R}$, and Poisson matrix $J(z): \mathbb{R}^m \times \mathbb{R}^m \to \mathbb{R}$. The Hamiltonian flow is given by $\dot{z} = J(z) D_z H$. Suppose now that $z_0$ is an equilibrium, so that $J(z_0) D_z H(z_0) = 0$. Infinitesimal perturbations from this equilibrium evolve as
\begin{equation}
    \delta \dot{z}
    =
    \lim_{\epsilon \to 0} \left( J(z_0 + \epsilon \delta z) D_z H(z_0 + \epsilon \delta z) \right)
    =
    (DJ(z_0) \delta z) D H_z (z_0)
    +
    J(z_0) D^2 H(z_0) \delta z \,.
\end{equation}
Now, suppose $J(z)$ has a nontrivial nullspace, and that we have found functionals whose gradients span the nullspace, $\{ \mathcal{C}_i \}$, called Casimir invariants. Then we may shift the energy by scalar multiples of the Casimir invariants without changing the dynamics:
\begin{equation}
    \mathcal{H}(z)
    =
    H(z) + \sum_i \lambda_i \mathcal{C}_i(z) \,.
\end{equation}
The energy Casimir method \cite{morrison1980noncanonical, morrison1998hamiltonian, morrison1982poisson, holm1985nonlinear, arnold1966geometrie, arnold1966priori} seeks a particular set of coefficients such that
\begin{equation}
    D_z \mathcal{H}(z_0) = 0 
    \implies
    \delta \dot{z} = J(z_0) D^2 \mathcal{H}(z_0) \delta z \,.
\end{equation}
This reduces linear stability analysis to finding the eigenspace of the matrix $J(z_0) D^2 \mathcal{H}(z_0)$. 

\begin{remark}
    Matrices of the form $J(z_0) D^2 \mathcal{H}(z_0)$ are known as Hamiltonian matrices. Their spectral properties are well-studied \cite{mackay2020stability, morrison1998hamiltonian, kreuin1950generalization}. 
\end{remark}

%================================================================
\section{Properties of the IGR elliptic operators} \label{appendix:igr_elliptic_operators}

The information geometric regularization may be interpreted as a dual connection geodesic flow on a Hessian manifold with a logarithmic barrier enforcing positivity of the Jacobian determinant of the flow map. The Hamiltonian regularization may be interpreted as an elliptic regularization of the kinetic energy. Both models add an additional elliptic equation to be solved alongside the compressible Euler equations. This appendix establishes some useful properties of the elliptic operator in question. 

%================================================================
\subsection{A maximum principle for the scalar elliptic operator} \label{appendix:maximum_principle}
Let $\rho$ be a positive function and consider the elliptic equation
\begin{equation}
    \rho^{-1} \phi - \alpha \nabla_{\mathbf{x}} \cdot ( \rho^{-1} \nabla_{\mathbf{x}} \phi) = g \,.
\end{equation}
We wish to establish a maximum principle, so that on any subdomain, $g > 0$ implies $\phi > 0$. 

\begin{theorem}[Maximum principle for the weighted scalar elliptic operator]
\label{thm:weighted_max_principle}
Let $\Omega \subset \mathbb{R}^n$ be a bounded domain with Lipschitz boundary. Let $\rho \in C^1(\overline{\Omega})$ be uniformly bounded from below by $\rho_{\min} > 0$, i.e.,
$$ 0 < \rho_{\min} < \rho(x) \text{ for all } x \in \overline{\Omega}, $$
and let $\alpha>0$. Consider the elliptic problem
\begin{equation}
\label{eq:weighted_elliptic}
    \mathcal{L}_\alpha(\rho) \phi := \rho^{-1} \phi - \alpha \nabla_{\mathbf{x}} \cdot (\rho^{-1} \nabla_{\mathbf{x}} \phi) = g(x) 
    \quad \text{in } \Omega \,,
\end{equation}
with Dirichlet boundary conditions
\begin{equation}
    \phi \geq 0 \quad \text{on } \partial \Omega \,,
\end{equation}
and $g \in C(\overline{\Omega})$. Then
\begin{equation}
    g(x) \geq 0 \quad \text{in } \Omega \quad \implies \quad \phi(x) \geq 0 \quad \text{in } \overline{\Omega} \,.
\end{equation}
Moreover, if $\left. \phi \right|_{\partial \Omega} \geq 0$ is not identically zero and $g(x) > 0$ somewhere in $\Omega$, then $\phi > 0$ in $\Omega$.
\end{theorem}

\begin{proof}
Rewrite the operator by expanding the divergence term:
\begin{equation}
    \nabla_{\mathbf{x}} \cdot (\rho^{-1} \nabla_{\mathbf{x}} \phi) = \rho^{-1} \Delta_{\mathbf{x}} \phi - \rho^{-2} \nabla_{\mathbf{x}} \rho \cdot \nabla_{\mathbf{x}} \phi \,.
\end{equation}
Thus,
\begin{equation}
    \mathcal{L}_\alpha(\rho) \phi = \rho^{-1} \phi - \alpha \rho^{-1} \Delta_{\mathbf{x}} \phi + \alpha \rho^{-2} \nabla_{\mathbf{x}} \rho \cdot \nabla_{\mathbf{x}} \phi
    = \rho^{-1} \big( \phi - \alpha \Delta_{\mathbf{x}} \phi + \alpha \rho^{-1} \nabla_{\mathbf{x}} \rho \cdot \nabla_{\mathbf{x}} \phi \big) \,.
\end{equation}
Suppose, by way of contradiction, that $\phi$ attains a negative minimum at an interior point $x_0 \in \Omega$, i.e.,
\begin{equation}
    \phi(x_0) = \min_{x\in\overline{\Omega}} \phi(x) < 0 \,.
\end{equation}
At $x_0$, we have
\begin{equation}
    \nabla_{\mathbf{x}} \phi(x_0) = 0, \quad \Delta_{\mathbf{x}} \phi(x_0) \geq 0 \,.
\end{equation}
Evaluate the second-order term at $x_0$:
\begin{equation}
    - \alpha \Delta_{\mathbf{x}} \phi(x_0) + \alpha \rho^{-1}(x_0) \nabla_{\mathbf{x}} \rho(x_0) \cdot \nabla_{\mathbf{x}} \phi(x_0)
    = - \alpha \Delta_{\mathbf{x}} \phi(x_0) \leq 0 \,,
\end{equation}
since $\nabla_{\mathbf{x}} \phi(x_0) = 0$. Thus, at $x_0$,
\begin{equation}
    \mathcal{L}_\alpha(\rho) \phi(x_0) = \rho^{-1}(x_0) \big( \phi(x_0) - \alpha \Delta_{\mathbf{x}} \phi(x_0) + \alpha \rho^{-1} \nabla_{\mathbf{x}} \rho \cdot \nabla_{\mathbf{x}} \phi \big)
    \leq \rho^{-1}(x_0) \phi(x_0) < 0 \,.
\end{equation}
But \eqref{eq:weighted_elliptic} gives $\mathcal{L}_\alpha(\rho) \phi(x_0) = g(x_0) \geq 0$, a contradiction. Therefore, $\phi$ cannot attain a negative minimum in the interior of $\Omega$. Combined with the boundary condition $\phi \geq 0$ on $\partial \Omega$, we conclude
\begin{equation}
    \phi(x) \geq 0 \quad \forall x \in \overline{\Omega} \,.
\end{equation}
If $\left. \phi \right|_{\partial \Omega} \geq 0$ and not identically zero, a constant solution is impossible, and thus $\phi > 0$ in the interior. Hence, if $g(x) > 0$ somewhere in $\Omega$, then $\phi > 0$ in $\Omega$. 
\end{proof}

%================================================================
\subsection{Entropic pressure right-hand side decomposition} \label{appendix:rhs_decomp}

Here, we decompose $\mathrm{tr}^2(\nabla_{\mathbf{x}} \mathbf{u})$ and $\mathrm{tr}(\nabla_{\mathbf{x}} \mathbf{u}^2)$. Consider the symmetric-antisymmetric decomposition $\nabla_{\mathbf{x}} \mathbf{u} = \mathbb{S} + \mathbb{\Omega}$, where 
\begin{equation}
    \mathbb{S} = \mathrm{symm}(\nabla_{\mathbf{x}} \mathbf{u}) = \frac{1}{2}(\nabla_{\mathbf{x}} \mathbf{u} + (\nabla_{\mathbf{x}} \mathbf{u})^T)
\end{equation}
is the symmetric strain tensor and $\mathbb{\Omega} = \nabla_{\mathbf{x}} \mathbf{u} - \mathbb{S}$ is the vorticity tensor. Then,
\begin{align*}
    \mathrm{tr}^2(\nabla_{\mathbf{x}} \mathbf{u}) &= \mathrm{tr}^2(\mathbb{S}), \\
    \mathrm{tr}(\nabla_{\mathbf{x}} \mathbf{u}^2) &= \mathrm{tr}(\mathbb{S}^2) + 2 \mathrm{tr}(\mathbb{S}\mathbb{\Omega}) + \mathrm{tr}(\mathbb{\Omega}^2) = \mathrm{tr}(\mathbb{S}^2) + \mathrm{tr}(\mathbb{\Omega}^2),
\end{align*}
We can further decompose $\mathbb{S} = \mathbb{S}_D + \mathbb{S}_I$ where $\mathbb{S}_D = \mathbb{S} - \mathbb{S}_I$ is the deviatoric strain and $\mathbb{S}_I = \mathrm{tr}(\mathbb{S})I/d$ is the isotropic strain (here, $d$ is the dimension of the domain $\Omega \subset \mathbb{R}^d$). Then,
\begin{align*}
    \mathrm{tr}^2(\nabla_{\mathbf{x}} \mathbf{u}) &= \mathrm{tr}^2(\mathbb{S}), \\
    \mathrm{tr}(\nabla_{\mathbf{x}} \mathbf{u}^2) &= \mathrm{tr}(\mathbb{S}_D^2) + \frac{1}{d^2}\mathrm{tr}^2(\mathbb{S}) + \mathrm{tr}(\mathbb{\Omega}^2);
\end{align*}
thus, the non-conservative component $\mathrm{tr}^2(\nabla_\bx \bu)$ of the right-hand side of the entropic pressure equation in IGR has a contribution from the isotropic strain, whereas the conservative component $\mathrm{tr}(\nabla_\bx\bu^2)$ has a weaker contribution from the isotropic strain, plus a contribution from the deviatoric strain and the vorticity. In dimension $d=1$, the deviatoric strain and vorticity vanish, and both have equivalent isotropic strain contributions.

Denote by $\|\mathbb{A}\|_F := \mathrm{tr}(\mathbb{A}\mathbb{A}^T)^{1/2}$ the Frobenius norm. Then, using symmetry of $\mathbb{S}$ and asymmetry of $\mathbb{\Omega}$,
$$ \mathrm{tr}(\nabla_{\mathbf{x}} \mathbf{u}^2) = \|\mathbb{S}\|^2_F - \|\mathbb{\Omega}\|^2_F. $$
Thus, we have a positive component due to the strain and a negative component due to vorticity; in particular, large vorticity can lead to sign indefiniteness. 

%================================================================
\subsection{Commutation relations of weighted elliptic operators} \label{appendix:commutation_relations}

First, we establish some helpful commutation relations. Consider the vector and scalar elliptic operators:
\begin{equation}
    \mathcal{L}^v_\alpha(\rho) \mathbf{u}
    =
    \rho \mathbf{u}
    -
    \alpha \nabla_{\mathbf{x}} (\rho \nabla_{\mathbf{x}} \cdot \mathbf{u}) \,,
    \quad \mathsf{and} \quad
    \mathcal{L}^s_\alpha(\rho) f
    =
    \rho f
    -
    \alpha \nabla_{\mathbf{x}} \cdot ( \rho \nabla_{\mathbf{x}} f) \,.
\end{equation}
If we interpret these operators in a suitably weak sense, they have the domains $H_0(\mathrm{div}_{\mathbf{x}},\Omega)$ and $H^1_0(\Omega)$, respectively. Vanishing trace (or periodic boundary conditions) ensures invertibility of these operators. If we let $M_\rho: g \mapsto \rho g$ be the multiplication operator, then we see that these operators may be written:
\begin{equation}
    \mathcal{L}^v_\alpha(\rho)
    =
    M_\rho - \alpha \, \mathrm{grad}_{\mathbf{x}} \circ M_\rho \circ \mathrm{div}_{\mathbf{x}} \,,
    \quad \mathsf{and} \quad
    \mathcal{L}^s_\alpha(\rho)
    =
    M_\rho - \alpha \, \mathrm{div}_{\mathbf{x}} \circ M_\rho \circ \mathrm{grad}_{\mathbf{x}} \,.
\end{equation}
The following commutation relation (or rather a tensorial generalization) is used in the information geometric regularization. 

\begin{proposition} \label{prop:elliptic_identity}
On the domain $H^1_0(\Omega)$, the following relation holds:
\begin{equation}
    M_\rho \circ (\mathcal{L}^v_\alpha(\rho))^{-1} \circ \mathrm{grad}_{\mathbf{x}} = \mathrm{grad}_{\mathbf{x}} \circ (\mathcal{L}^s_\alpha(\rho^{-1}))^{-1} \circ M_\rho^{-1} \,.
\end{equation}
\end{proposition}
\begin{proof}
Let $f$ be an arbitrary scalar function, and let $\mathbf{u} = (\mathcal{L}^v_\alpha(\rho))^{-1} \nabla_{\mathbf{x}} f$. Then
\begin{equation}
    \rho \mathbf{u} - \alpha \nabla_{\mathbf{x}} (\rho \nabla_{\mathbf{x}} \cdot \mathbf{u}) = \nabla_{\mathbf{x}} f \,.
\end{equation}
Hence, rearranging, we find that
\begin{equation}
    (M_\rho \circ (\mathcal{L}^v_\alpha(\rho))^{-1} \circ \nabla_{\mathbf{x}}) f = \rho \mathbf{u} = \nabla_{\mathbf{x}} (f + \alpha \rho \nabla_{\mathbf{x}} \cdot \mathbf{u}) \,.
\end{equation}
Define the scalar function $g = f + \alpha \rho \nabla_{\mathbf{x}} \cdot \mathbf{u}$. Then $\nabla_{\mathbf{x}} g = \rho \mathbf{u} = M_\rho( \mathcal{L}^v_\alpha(\rho))^{-1} \nabla_{\mathbf{x}} f$. But note that
\begin{equation}
\begin{aligned}
    \mathcal{L}^s_\alpha(\rho^{-1}) g
    &=
    \rho^{-1} g - \alpha \nabla_{\mathbf{x}} \cdot(\rho^{-1} \nabla_{\mathbf{x}} g)
    =
    \rho^{-1} (f + \alpha \rho \nabla_{\mathbf{x}} \cdot \mathbf{u}) - \alpha \nabla_{\mathbf{x}} \cdot(\rho^{-1} (\rho \mathbf{u})) \\
    &=
    \rho^{-1} f + \alpha \nabla_{\mathbf{x}} \cdot \mathbf{u} - \alpha \nabla_{\mathbf{x}} \cdot \mathbf{u}
    =
    \rho^{-1} f \,.
\end{aligned}
\end{equation}
So, we find that $\nabla_{\mathbf{x}} (\mathcal{L}^s_\alpha(\rho^{-1}))^{-1} M_{\rho}^{-1} f = \nabla_{\mathbf{x}} g = M_\rho( \mathcal{L}^v_\alpha(\rho))^{-1} \nabla_{\mathbf{x}} f$. 
\end{proof}
\noindent As an immediate corollary we also find the following formula.
\begin{proposition}
On the domain $H_0(\mathrm{div}_{\mathbf{x}},\Omega)$, the following relation holds:
\begin{equation}
    \mathrm{div}_{\mathbf{x}} \circ (\mathcal{L}^v_\alpha(\rho))^{-1} \circ M_\rho = M_{\rho^{-1}} \circ (\mathcal{L}^s_\alpha(\rho^{-1}))^{-1} \circ \mathrm{div}_{\mathbf{x}} \,.
\end{equation}
\end{proposition}
\begin{proof}
    This is obtained by taking the adjoint of the formula established in Proposition \ref{prop:elliptic_identity}. As the elliptic operators (and thus there inverses) are self-adjoint, the formula follows. 
\end{proof}

We next consider the tensorial version of the identity. Let $\mathbb{\Sigma}$ be a matrix. Define
\begin{equation}
    \mathcal{L}^m_\alpha(\rho)
    \mathbb{\Sigma}
    =
    \rho \mathbb{\Sigma}
    -
    \alpha \nabla_{\mathbf{x}} \cdot (\rho \nabla_{\mathbf{x}} \cdot \mathbb{\Sigma}) \mathbb{I} \,.
\end{equation}
Further, notice that we can rewrite $\mathcal{L}^v_\alpha(\rho)$:
\begin{equation}
    \mathcal{L}^v_\alpha(\rho)
    \mathbf{u}
    =
    \rho \mathbf{u}
    -
    \alpha \nabla_{\mathbf{x}} \cdot (\rho (\nabla_{\mathbf{x}} \cdot \mathbf{u}) \mathbb{I}) \,,
\end{equation}
since the divergence of a scalar times the identity is equivalently the gradient of the scalar.

\begin{proposition}
On the domain $H^1_0(\Omega;\mathbb{R}^{d \times d})$, the $H^1$ Sobolev space of $d \times d$ matrix fields on $\Omega$ with vanishing boundary trace, the following commutation relation holds:
\begin{equation}
    M_\rho \circ (\mathcal{L}^v_\alpha(\rho))^{-1} \circ \mathrm{div}_{\mathbf{x}}
    =
    \mathrm{div}_{\mathbf{x}} \circ (\mathcal{L}^m_\alpha(\rho^{-1}))^{-1} \circ M_\rho^{-1} \,.
\end{equation}
\end{proposition}

\begin{proof}
    We proceed nearly the same as in the previous proof. Let $\mathbb{\Sigma} \in H^1_0(\Omega; \mathbb{R}^{d \times d})$ be arbitrary. To begin, define
    \begin{equation}
        \mathbf{v}
        =
        \nabla_{\mathbf{x}} \cdot \mathbb{\Sigma} \,,
        \quad \text{and} \quad
        \mathbf{u}
        =
        (\mathcal{L}^v_\alpha(\rho))^{-1} \mathbf{v} \,.
    \end{equation}
    Then the proof is equivalent to showing that $M_\rho \mathbf{u} = (\mathrm{div}_{\mathbf{x}} \circ (\mathcal{L}^m_\alpha(\rho^{-1}))^{-1} \circ M_\rho^{-1}) \mathbb{\Sigma}$. But notice,
    \begin{equation}
        \rho \mathbf{u}
        =
        \nabla_{\mathbf{x}} \cdot ( \mathbb{\Sigma} + \alpha \rho (\nabla_{\mathbf{x}} \cdot \mathbf{u}) \mathbb{I} ) \,.
    \end{equation}
    Define $\mathbb{\Upsilon} = \mathbb{\Sigma} + \alpha \rho (\nabla_{\mathbf{x}} \cdot \mathbf{u}) \mathbb{I}$. Then
    \begin{equation}
        \nabla_{\mathbf{x}} \cdot \mathbb{\Upsilon}
        =
        \rho \mathbf{u}
        =
        (M_\rho \circ (\mathcal{L}^v_\alpha(\rho))^{-1} \circ \mathrm{div}_{\mathbf{x}}) \mathbb{\Sigma} \,.
    \end{equation}
    But
    \begin{equation}
    \begin{aligned}
        \mathcal{L}^m_\alpha(\rho^{-1}) \mathbb{\Upsilon}
        &=
        \rho^{-1} \mathbb{\Upsilon} - \alpha \nabla_{\mathbf{x}} \cdot(\rho^{-1} \nabla_{\mathbf{x}} \cdot \mathbb{\Upsilon}) \mathbb{I} \\
        &=
        \rho^{-1} (\mathbb{\Sigma} + \alpha \rho (\nabla_{\mathbf{x}} \cdot \mathbf{u}) \mathbb{I}) - \alpha \nabla_{\mathbf{x}} \cdot(\rho^{-1} \rho \mathbf{u}) \mathbb{I}
        =
        \rho^{-1} \mathbb{\Sigma} \,.
    \end{aligned}
    \end{equation}
    Rearranging this last equation, we find that 
    \begin{equation}
        (M_\rho \circ (\mathcal{L}^v_\alpha(\rho))^{-1} \circ \mathrm{div}_{\mathbf{x}}) \mathbb{\Sigma} 
        = 
        \nabla_{\mathbf{x}} \cdot \mathbb{\Upsilon} 
        = 
        (\mathrm{div}_{\mathbf{x}} \circ (\mathcal{L}^m_\alpha(\rho^{-1}))^{-1} \circ M_{\rho^{-1}}) \mathbb{\Sigma} \,.
    \end{equation}
\end{proof}

%================================================================
\subsection{Simplifying the conservative entropic pressure matrix} \label{appendix:simplifying_entropic_pressure}

Consider an elliptic equation of the form
\begin{equation} \label{eq:general_matrix_elliptic_eqn}
    \rho^{-1} \mathbb{\Sigma} - \alpha \nabla_\bx \cdot (\rho^{-1} \nabla_\bx \cdot \mathbb{\Sigma}) \mathbb{I}
    =
    f(\bx) \mathbb{I} + \mathbb{G}(\bx) \,,
\end{equation}
where $f: \Omega \to \mathbb{R}$ and $\mathbb{G}(\bx): \Omega \to \mathbb{R}^{d \times d}$. If $\mathbb{G}$ is sufficiently regular, it is possible to reduce this to a scalar elliptic equation. 

\begin{lemma} \label{lemma:matrix_elliptic_eqn}
If $\mathbb{G}(\bx): \Omega \to \mathbb{R}^{d \times d}$ is twice differentiable, then the solution to the matrix elliptic equation,
\begin{equation}
    \rho^{-1} \mathbb{\Sigma} - \alpha \nabla_\bx \cdot (\rho^{-1} \nabla_\bx \cdot \mathbb{\Sigma}) \mathbb{I}
    =
    \mathbb{G}(\bx) \,,
\end{equation}
is given by
\begin{equation}
\left\{
\begin{aligned}
    &\mathbb{\Sigma}
    =
    \rho \mathbb{G} + \Sigma \mathbb{I} \,, \\
    &\rho^{-1} \Sigma - \alpha \nabla_{\bx} \cdot (\rho^{-1} \nabla_{\bx} \Sigma)
    =
    \alpha \nabla_{\bx} \cdot (\rho^{-1} \nabla_\bx \cdot (\rho \mathbb{G})) \,.
\end{aligned}
\right. 
\end{equation}
\end{lemma}

\begin{proof}
\textbf{Step 1: splitting the diagonal and off-diagonal components.} First, note that we can separate the diagonal and off-diagonal equations. The off-diagonal component is solved algebraically:
\begin{equation}
    \rho^{-1} \mathbb{\Sigma}_0
    =
    \mathbb{G}_0(\bx) \,,
\end{equation}
where $(\mathbb{\Sigma}_0)_{ij} = \mathbb{\Sigma}_{ij}$ for $i \neq j$ and $(\mathbb{\Sigma}_0)_{ii} = 0$, and similarly for $\mathbb{G}$. The diagonal solve is then given by
\begin{equation}
    \rho^{-1} \mathbf{\Sigma} - \alpha \nabla_\bx \cdot (\rho^{-1} \nabla_\bx \cdot \mathrm{diag}(\mathbf{\Sigma}))
    =
    \mathbf{G}  + \alpha \nabla_{\bx} \cdot (\rho^{-1} \nabla_\bx \cdot (\rho \mathbb{G}_0))\mathbf{1} 
    =
    \mathbf{h}
    \,,
\end{equation}
where $\mathbf{\Sigma}_i = \mathbb{\Sigma}_{ii}$, $\mathbf{G}_i = \mathbb{G}_{ii}$, and $\mathbf{1}_i = 1$. The elliptic equation for the diagonal components of $\mathbb{\Sigma}$ may be decomposed as a product of simpler operators, which facilitates its simplification. 

\textbf{Step 2: reduction to a scalar elliptic solve.} Define the scalar elliptic operators
\begin{equation}
    \Delta_\rho f = \nabla_{\bx} \cdot (\rho \nabla_{\bx} f)
    \quad \text{and} \quad
    \mathcal{K}_\alpha(\rho) f = f - \alpha \Delta_{\rho^{-1}} (\rho f) \,,
\end{equation}
and the following maps:
\begin{itemize}
    \item $M_\rho: \mathbf{f} \mapsto \rho \mathbf{f}$ is the multiplication operator,
    \item $V_\rho: \mathbf{f} \mapsto \sum_{k=1}^d \partial_k (\rho \partial_k f_k)$, 
    \item and $E: f \mapsto (f, f, \hdots, f) = f \mathbf{1}$. 
\end{itemize}
Then the elliptic equation for $\mathbf{\Sigma}$ takes the form
\begin{equation} \label{eq:diag_elliptic}
    (M_{\rho^{-1}} - \alpha E \circ V_{\rho^{-1}}) \mathbf{\Sigma} = \mathbf{h} \,.
\end{equation}
The Sherman--Morrison--Woodbury identity, which holds for bounded linear operators on Hilbert [CITE] and Banach [CITE] spaces, implies
\begin{equation}
    (M_{\rho^{-1}} - \alpha E \circ V_{\rho^{-1}})^{-1}
    =
    M_\rho + \alpha M_\rho \circ E \circ (I - \alpha V_{\rho^{-1}} \circ M_\rho \circ E)^{-1} V_{\rho^{-1}} \circ M_\rho \,,
\end{equation}
if and only if $I - \alpha V_{\rho^{-1}} \circ M_\rho \circ E = \mathcal{K}_\alpha(\rho)$ is bijective. It is a straightforward exercise to demonstrate the coercivity of this operator with mild assumptions on the regularity of $\rho$, and we omit the details. Hence, we find that 
\begin{equation}
    \mathbf{\Sigma}
    =
    \rho \mathbf{h}
    +
    \alpha \rho [\mathcal{K}_\alpha^{-1}(\rho)(V_{\rho^{-1}}(\rho \mathbf{h})] \mathbf{1} \,.
\end{equation}

\textbf{Step 3: algebraic cancellation.} Summing the expressions for the diagonal and off-diagonal components of $\mathbb{\Sigma}$, we find that the matrix entropic pressure takes the form
\begin{equation}
\begin{aligned}
    \mathbb{\Sigma}
    =
    \mathbb{\Sigma}_0
    +
    \mathrm{diag}(\mathbf{\Sigma}) 
    &=
    \rho \mathbb{G}_0(\bx)
    + \rho \mathbf{h} + \alpha \rho [\mathcal{K}_\alpha^{-1}(\rho)(V_{\rho^{-1}}(\rho \mathbf{h})] \mathbf{1} \\
    &=
    \rho \mathbb{G}(\bx)
    +
    \alpha \rho \left[
    \mathcal{K}_\alpha^{-1}(\rho)(V_{\rho^{-1}} (\rho \mathbf{h}))
    +
    \nabla_{\bx} \cdot (\rho^{-1} \nabla_{\bx} \cdot (\rho \mathbb{G}_0))
    \right] \mathbb{I} \,.
\end{aligned}
\end{equation}
Observe that
\begin{equation}
    \mathcal{K}_\alpha^{-1}(\rho)(V_{\rho^{-1}} (\rho \mathbf{h}))
    =
    \mathcal{K}_\alpha^{-1}(\rho)(V_{\rho^{-1}} \left( \rho \mathbf{G}  + \alpha \rho \nabla_{\bx} \cdot (\rho^{-1} \nabla_\bx \cdot (\rho \mathbb{G}_0))\mathbf{1} \right) \,.
\end{equation}
For any scalar $g$, we find that $V_{\rho^{-1}}(g \mathbf{1}) = \Delta_{\rho^{-1}} g$. Moreover, recall that $(I - A)^{-1} Af = [I - A]^{-1} f - f$ for a general operator $A$ such that $I - A$ is invertible. Hence, if $A = \alpha \Delta_{\rho^{-1}} \circ \rho$, we find
\begin{equation}
    \mathcal{K}_\alpha^{-1}(\rho) \Delta_{\rho^{-1}}(\alpha \rho g)
    =
    (I - \alpha \Delta_{\rho^{-1}} \circ \rho)^{-1} \circ ( \alpha \Delta_{\rho^{-1}} \circ \rho) g
    =
    [I - \alpha \Delta_{\rho^{-1}} \circ \rho]^{-1} g - g
    =
    \mathcal{K}_\alpha^{-1}(\rho) g - g \,.
\end{equation}
Next, observe that 
\begin{equation}
    V_{\rho^{-1}}(\rho \mathbf{G})
    =
    \nabla_\bx \cdot (\rho^{-1} \nabla_\bx \cdot (\rho \mathrm{diag}(\mathbf{G}))) \,.
\end{equation}
So, we find that
\begin{equation}
\begin{aligned}
    \alpha \mathcal{K}_\alpha^{-1}(\rho)(V_{\rho^{-1}} (\rho \mathbf{h}))
    &= \mathcal{K}_\alpha^{-1}(\rho) [\alpha \nabla_{\bx} \cdot (\rho^{-1} \nabla_\bx \cdot (\rho \mathbb{G}_0))] 
    + \alpha \mathcal{K}_\alpha^{-1}(V_{\rho^{-1}}(\rho \mathbf{G})) 
    - \alpha \nabla_{\bx} \cdot (\rho^{-1} \nabla_\bx \cdot (\rho \mathbb{G}_0)) \\
    &= \mathcal{K}_\alpha^{-1}(\rho) [\alpha \nabla_{\bx} \cdot (\rho^{-1} \nabla_\bx \cdot (\rho \mathbb{G}))] 
    - \alpha \nabla_{\bx} \cdot (\rho^{-1} \nabla_\bx \cdot (\rho \mathbb{G}_0)) \,.
\end{aligned}
\end{equation}

Putting this together, we see that
\begin{equation}
\begin{aligned}
    \alpha \rho \left[
    \mathcal{K}_\alpha^{-1}(\rho)(V_{\rho^{-1}} (\rho \mathbf{h}))
    +
    \nabla_{\bx} \cdot (\rho^{-1} \nabla_{\bx} \cdot (\rho \mathbb{G}_0))
    \right]
    &=
    \alpha \rho 
    \mathcal{K}_\alpha^{-1}(\rho) [\nabla_{\bx} \cdot (\rho^{-1} \nabla_\bx \cdot (\rho \mathbb{G}))] 
\end{aligned}
\end{equation}
As a final simplification, notice that $\rho \mathcal{K}_\alpha^{-1}(\rho) = (\mathcal{L}_\alpha^{s}(\rho^{-1}))^{-1}$, where
\begin{equation}
    \mathcal{L}_\alpha^{s}(\rho^{-1})f = \rho^{-1} f - \alpha \nabla_{\bx} \cdot (\rho^{-1} \nabla_{\bx} f)
\end{equation}
is the scalar elliptic operator from standard IGR [CITE], and whose properties are discussed elsewhere in this appendix. Therefore, we obtain the final expression:
\begin{equation}
\left\{
\begin{aligned}
    &\mathbb{\Sigma}
    =
    \rho \mathbb{G} + \Sigma \mathbb{I} \,, \\
    &\rho^{-1} \Sigma - \alpha \nabla_{\bx} \cdot (\rho^{-1} \nabla_{\bx} \Sigma)
    =
    \alpha \nabla_{\bx} \cdot (\rho^{-1} \nabla_\bx \cdot (\rho \mathbb{G})) \,.
\end{aligned}
\right. 
\end{equation}
\end{proof}

\begin{theorem}
Assuming $\mathbb{G}$ is twice differentiable, the solution to~\Cref{eq:general_matrix_elliptic_eqn} is given by
\begin{equation}
    \mathbb{\Sigma}
    =
    \rho \mathbb{G} + \Sigma \mathbb{I} \,,
    \quad \text{where} \quad
    \rho^{-1} \Sigma - \alpha \nabla_\bx \cdot (\rho^{-1} \nabla_\bx \Sigma) 
    =
    f + \alpha \nabla_{\bx} \cdot (\rho^{-1} \nabla_\bx \cdot (\rho \mathbb{G})) \,.
\end{equation}
\end{theorem}

\begin{proof}
By linearity, we can split the elliptic equation into two equations:
\begin{equation}
    \rho^{-1} \mathbb{\Sigma}_s - \alpha \nabla_\bx \cdot (\rho^{-1} \nabla_\bx \cdot \mathbb{\Sigma}_s) \mathbb{I}
    =
    f(\bx) \mathbb{I} \,,
    \quad \text{and} \quad
    \rho^{-1} \mathbb{\Sigma}_m - \alpha \nabla_\bx \cdot (\rho^{-1} \nabla_\bx \cdot \mathbb{\Sigma}_m) \mathbb{I}
    =
    \mathbb{G}(\bx) \,.
\end{equation}
We then recover $\mathbb{\Sigma} = \mathbb{\Sigma}_s + \mathbb{\Sigma}_m$. The first reduces to a scalar equation: $\mathbb{\Sigma}_s = \Sigma_s \mathbb{I}$, where
\begin{equation}
    \rho^{-1} \Sigma_s - \alpha \nabla_\bx \cdot (\rho^{-1} \nabla_\bx \Sigma_s) 
    =
    f(\bx) \,. 
\end{equation}
Moreover, from~\Cref{lemma:matrix_elliptic_eqn}, we see that
\begin{equation}
    \mathbb{\Sigma}_m
    =
    \rho \mathbb{G} + \Sigma_m \mathbb{I} \,,
    \quad \text{where} \quad
    \rho^{-1} \Sigma_m - \alpha \nabla_\bx \cdot (\rho^{-1} \nabla_\bx \Sigma_m) 
    =
    \alpha \nabla_{\bx} \cdot (\rho^{-1} \nabla_\bx \cdot (\rho \mathbb{G})) \,.
\end{equation}
\end{proof}

%================================================================
\section{Coordinate changes in the Hamiltonian regularized Euler model} \label{appendix:change_of_variables}

In Eulerian coordinates, the Lagrangian for the HRE model takes the form
\begin{equation}
    \ell[\mathbf{u}, \rho, s]
    =
    \int_\Omega \frac{1}{2} \left( \rho ( |\mathbf{u}|^2 + \alpha (\nabla_{\mathbf{x}} \cdot \mathbf{u})^2) - e(\rho, s, \nabla_{\mathbf{x}} \rho) \right) \mathsf{d}^d \mathbf{x} \,,
\end{equation}
where the dependence of the internal energy, $e = e(\rho, s, \nabla_{\mathbf{x}} \rho)$, on $\nabla_{\mathbf{x}} \rho$ takes a particular form in order to cancel the spurious dispersion relation of sound waves from the regularization of the kinetic energy. In Lagrangian coordinates, the Lagrangian takes the following form:
\begin{multline}
    L[\bm{\Phi}, \dot{\bm{\Phi}}]
    =
    \int_\Omega
    \Bigg(
    \rho_0
    \bigg[
    \frac{1}{2}
    \left(
    |\dot{\bm{\Phi}}|^2
    +
    \alpha \mathrm{tr}( \nabla_{\mathbf{X}} \bm{\Phi}^{-1} \nabla_{\mathbf{X}} \dot{\bm{\Phi}})^2
    \right) \\
    -
    e\left( \frac{\rho_0}{\det(\nabla_{\mathbf{X}} \bm{\Phi})}, s_0, \nabla_{\mathbf{X}} \bm{\Phi}^{-T} \nabla_{\mathbf{X}} \left( \frac{\rho_0}{\det(\nabla_{\mathbf{X}} \bm{\Phi})} \right) \right)
    \bigg] 
    \Bigg)
    \mathsf{d}^d \mathbf{X}_0
    \,.
\end{multline}
In the main body of this paper, we derive the evolution equations by deriving the Euler--Lagrange equation in the Lagrangian reference frame, however it is possible to derive the evolution equations directly in the Eulerian reference frame using either a constrained variational principle, the Euler--Poincar{\'e} formalism \cite{holm1998euler, marsden2013introduction}, or a noncanonical Hamiltonian formalism \cite{morrison1982poisson, morrison1998hamiltonian, morrison1980noncanonical, marsden2013introduction, holm1998euler}. We describe how the latter approach is obtained directly from the canonical Hamiltonian structure in this appendix. 

%----------------------------
\subsection{The canonical Hamiltonian formulation of the pressureless HRE equation} \label{appendix:canonical_pressureless_HRE}

In order to obtain the pressureless HRE system in Hamiltonian form in canonical coordinates, we begin by finding the momentum coordinate which is canonically conjugate to the flow map, $\bm{\Phi}_t$. We define $\bm{\Pi}_t \in T^* \mathrm{Diff}(\Omega)$ to be
\begin{equation}
    \bm{\Pi}_t
    =
    D_{\dot{\bm{\Phi}}_t} L[\bm{\Phi}_t, \dot{\bm{\Phi}}_t]
    =
    D^2 \tilde{\psi}_\alpha[\iota(\bm{\Phi}_t)](D \iota[\bm{\Phi}_t] \dot{\bm{\Phi}}_t, D \iota[\bm{\Phi}_t] (\cdot) ) 
    =
    (\iota^* D^2 \tilde{\psi}_\alpha [\bm{\Phi}_t])(\dot{\bm{\Phi}}_t, \cdot) \,.
\end{equation}
Recall that $\bm{\Pi}_t \in T^* \mathrm{Diff}(\Omega)$ is a linear functional which acts on tangent vectors, whose action we denote by $\langle\bm{\Pi}_t,\cdot\rangle$. The Hamiltonian is then defined via the Legendre transform:
\begin{equation}
    H[\bm{\Phi}_t, \bm{\Pi}_t]
    =
    \left(
    \left\langle
        \bm{\Pi}_t, \dot{\bm{\Phi}}_t
    \right\rangle
    -
    L[\bm{\Phi}_t, \dot{\bm{\Phi}}_t]
    \right)_{\dot{\bm{\Phi}}_t = \dot{\bm{\Phi}}_t[\bm{\Phi}_t, \bm{\Pi}_t]} \,.
\end{equation}
Notice that the canonical momentum and the velocity are related via
\begin{equation}
    \dot{\bm{\Phi}}_t
    =
    ((\iota^* D^2 \tilde{\psi}_\alpha)[\bm{\Phi}_t])^{-1} (\bm{\Pi}_t, \cdot) \,.
\end{equation}
We may think of the velocity as an element of the bi-dual by the canonical embedding: $T \mathrm{Diff}(\Omega) \hookrightarrow T (\mathrm{Diff}(\Omega))^{**}$. Hence, we see that
\begin{equation}
    L[\bm{\Phi}_t, \dot{\bm{\Phi}}_t]
    =
    \frac{1}{2}
    (\iota^* D^2 \tilde{\psi}_\alpha [\bm{\Phi}_t])(\dot{\bm{\Phi}}_t, \dot{\bm{\Phi}}_t)
    =
    \frac{1}{2}
    \left\langle
        \bm{\Pi}_t, \dot{\bm{\Phi}}_t
    \right\rangle
    =
    \frac{1}{2}
    ((\iota^* D^2 \tilde{\psi}_\alpha)[\bm{\Phi}_t])^{-1} (\bm{\Pi}_t, \bm{\Pi}_t) \,.
\end{equation}
Therefore, the Hamiltonian may be written as 
\begin{equation}
    H[\bm{\Phi}_t, \bm{\Pi}_t]
    =
    \frac{1}{2}
    ((\iota^* D^2 \tilde{\psi}_\alpha)[\bm{\Phi}_t])^{-1} (\bm{\Pi}_t, \bm{\Pi}_t) \,.
\end{equation}
The canonical Poisson bracket is a bilinear map on smooth functionals of the cotangent bundle:
\begin{equation}
    \{F, G\}
    =
    \left\langle
    D_{\bm{\Phi}_t} F, 
    D_{\bm{\Pi}_t} G
    \right\rangle
    -
    \left\langle
    D_{\bm{\Phi}_t} F, 
    D_{\bm{\Pi}_t} G
    \right\rangle \,.
\end{equation}
The evolution equations (in a weak form) are obtained via the formula
\begin{equation}
    \dot{F}
    =
    \{F, H\} \,.
\end{equation}
To write this explicitly, we need to compute derivatives of the Hamiltonian. The derivative of the Hamiltonian with respect to momentum is simply the inverse Legendre transform:
\begin{equation}
    D_{\bm{\Pi}_t} H
    =
    ((\iota^* D^2 \tilde{\psi}_\alpha)[\bm{\Phi}_t])^{-1} (\bm{\Pi}_t, \cdot)    
    =
    \dot{\bm{\Phi}}_t \,.
\end{equation}
The derivative with respect to the flow map is a bit more involved:
\begin{equation}
\begin{aligned}
    D_{\bm{\Phi}_t} &H[\bm{\Phi}_t, \bm{\Pi}_t] \cdot \mathbf{U}
    =
    \frac{1}{2} D_{\bm{\Phi}_t} \left[ ((\iota^* D^2 \tilde{\psi}_\alpha)[\bm{\Phi}_t])^{-1} (\bm{\Pi}_t, \bm{\Pi}_t) \right] \cdot \mathbf{U} \\
    &=
    - \frac{1}{2} \left( D_{\bm{\Phi}_t} (\iota^* D^2 \tilde{\psi}_\alpha)[\bm{\Phi}_t] \cdot \mathbf{U} \right) (((\iota^* D^2 \tilde{\psi}_\alpha)[\bm{\Phi}_t])^{-1} \bm{\Pi}_t, ((\iota^* D^2 \tilde{\psi}_\alpha)[\bm{\Phi}_t])^{-1} \bm{\Pi}_t) \\
    &=
    -\frac{1}{2} \left( D^3 \tilde{\psi}_\alpha[\iota(\bm{\Phi}_t)](D\iota[\bm{\Phi}_t]\mathbf{U}, \cdot, \cdot) 
    + 2 D^2 \tilde{\psi}_\alpha[\iota(\bm{\Phi}_t)](D^2\iota[\bm{\Phi}_t](\mathbf{U}, \cdot), \cdot) \right) (\dot{\bm{\Phi}}_t, \dot{\bm{\Phi}}_t) \\
    &=
    -\frac{1}{2} D^3 \tilde{\psi}_\alpha[\iota(\bm{\Phi}_t)](D\iota[\bm{\Phi}_t]\mathbf{U}, D\iota[\bm{\Phi}_t]\dot{\bm{\Phi}}_t, D\iota[\bm{\Phi}_t]\dot{\bm{\Phi}}_t) \\
    &\hspace{10em}
    - D^2 \tilde{\psi}_\alpha[\iota(\bm{\Phi}_t)](D^2\iota[\bm{\Phi}_t](\mathbf{U}, \dot{\bm{\Phi}}_t), D\iota[\bm{\Phi}_t]\dot{\bm{\Phi}}_t) 
    \,,
\end{aligned}
\end{equation}
where we applied the identity
\begin{equation}
    (D_{\mathbf{X}} (g_{\mathbf{X}})^{-1})(\mathbf{U}, \mathbf{V},\mathbf{W})
    =
    - (D_{\mathbf{X}} g_{\mathbf{X}})((g_{\mathbf{X}})^{-1}\mathbf{U}, (g_{\mathbf{X}})^{-1}\mathbf{V}, \mathbf{W}) 
    \,,
\end{equation}
where $(g_{\mathbf{X}})^{-1}: T \mathrm{Diff}(\Omega) \times T \mathrm{Diff}(\Omega) \to \mathbb{R}$ is the inverse of $g_{\mathbf{X}}: T \mathrm{Diff}(\Omega) \times T \mathrm{Diff}(\Omega) \to \mathbb{R}$, and we interpret $(g_{\mathbf{X}})^{-1} \mathbf{U} \in T \mathrm{Diff}(\Omega)$ for $\mathbf{U} \in T\mathrm{Diff}(\Omega)$. 

\begin{comment}
Although the calculation is shown in full above, it might be expedited by exploiting the identity $D_{\bm{\Phi}_t} H = - D_{\bm{\Phi}_t} L$, which holds in general due to the involutive character of the Legendre transform.
\end{comment} 
%BKT: I dont think this statement is necessary

Hence, we recover the following weak evolution equations:
\begin{equation}
\begin{aligned}
    \left\langle
    \dot{\bm{\Phi}}_t, \mathbf{V}
    \right\rangle
    &=
    (\iota^* D^2 \tilde{\psi}_\alpha)[\bm{\Phi}_t])^{-1} (\bm{\Pi}_t, \mathbf{V}) \\
    \left\langle
    \dot{\bm{\Pi}}_t, \mathbf{U}
    \right\rangle
    &=
    \frac{1}{2} D^3 \tilde{\psi}_\alpha[\iota(\bm{\Phi}_t)](D\iota[\bm{\Phi}_t]\mathbf{U}, D\iota[\bm{\Phi}_t]\dot{\bm{\Phi}}_t, D\iota[\bm{\Phi}_t]\dot{\bm{\Phi}}_t) \\
    &\hspace{10em}
    + D^2 \tilde{\psi}_\alpha[\iota(\bm{\Phi}_t)](D^2\iota[\bm{\Phi}_t](\mathbf{U}, \dot{\bm{\Phi}}_t), D\iota[\bm{\Phi}_t]\dot{\bm{\Phi}}_t) \,,
\end{aligned}
\end{equation}
for all $\mathbf{V} \in T^* \mathrm{Diff}(\Omega)$ and $\mathbf{U} \in T \mathrm{Diff}(\Omega)$. By substituting the first equation into the second, we find that
\begin{multline}
    \frac{\mathsf{d}}{\mathsf{d} t}
    \left(
    \iota^* D^2 \tilde{\psi}_\alpha(\dot{\bm{\Phi}}_t, \mathbf{U})
    \right)
    =
    \frac{1}{2} D^3 \tilde{\psi}_\alpha[\iota(\bm{\Phi}_t)](D\iota[\bm{\Phi}_t]\mathbf{U}, D\iota[\bm{\Phi}_t]\dot{\bm{\Phi}}_t, D\iota[\bm{\Phi}_t]\dot{\bm{\Phi}}_t) \\
    + D^2 \tilde{\psi}_\alpha[\iota(\bm{\Phi}_t)](D^2\iota[\bm{\Phi}_t](\mathbf{U}, \dot{\bm{\Phi}}_t), D\iota[\bm{\Phi}_t]\dot{\bm{\Phi}}_t) \,,
\end{multline}
for all $\mathbf{U} \in T \mathrm{Diff}(\Omega)$. Expanding the time-derivative, we recover the Euler--Lagrange equations for the Levi-Civita flow defining the pressureless HRE model in Equation~\eqref{eq:hre_geodesic_eqn}. 

%----------------------------
\subsection{Lagrangian to Eulerian change of variables formulae}

In principle, the previous section has obtained the desired evolution equation at least for the pressureless part of the HRE model. However, what we really want is to obtain the evolution equations in the Eulerian coordinates in conservative form. Note, we ultimately seek the evolution equations in the coordinates $(\rho \mathbf{u}, \rho, E)$, but we settle for an intermediate result in this appendix and simply find the evolution equations in the variables $(\mathbf{m}, \rho, e)$. We determine the evolution equations in these variables by identifying the coordinate change expressing the canonical Poisson bracket in the Lagrangian reference frame in the Eulerian coordinates $(\mathbf{m},\rho, e)$, closely following the approach taken in \cite{suzuki2020generic}. This will ensure thermodynamic consistency of the final system. 

The Eulerian mass, momentum, and entropy densities are related to their Lagrangian counterparts as follows:
\begin{equation} \label{eq:coord_transforms}
\begin{aligned}
    \rho(x,t)
    &=
    \phi_\rho(\bm{\Phi})(x)
    \eqcolon
    \frac{\rho_0(\bm{\Phi}^{-1}(x))}{\det( \nabla_{\mathbf{X}} \bm{\Phi}(\bm{\Phi}^{-1}(x), t))} 
    \,, \\
    \mathbf{m}(x,t)
    &=
    \phi_m(\bm{\Phi}, \bm{\Pi})(x)
    \eqcolon
    \frac{\bm{\Pi}(\bm{\Phi}^{-1}(x),t)}{\det( \nabla_{\mathbf{X}} \bm{\Phi}(\bm{\Phi}^{-1}(x), t))}
    \,, \\
    \sigma(x,t)
    &=
    \phi_\sigma(\bm{\Phi})(x)
    \eqcolon
    \frac{\rho_0(\bm{\Phi}^{-1}(x)) s_0(\bm{\Phi}^{-1}(x))}{\det( \nabla_{\mathbf{X}} \bm{\Phi}(\bm{\Phi}^{-1}(x), t))} \,.
\end{aligned}
\end{equation}
The internal energy density is somewhat more complicated, but may be written in terms of the prior definitions:
\begin{equation}
    e(x,t)
    =
    \phi_e(\bm{\Phi})(x)
    \eqcolon
    \phi_\rho(\bm{\Phi})(x)
    \varepsilon
    \left(
    \phi_\rho(\bm{\Phi})(x) \,,
    \frac{\phi_\sigma(\bm{\Phi})(x)}{\phi_\rho(\bm{\Phi})(x)} \,,
    \nabla_{\mathbf{x}}[\phi_\rho(\bm{\Phi})](x)
    \right) \,,
\end{equation}
where $\varepsilon(\rho, \sigma/\rho, \nabla_{\bx} \rho)$ is the specific internal energy, and the total internal energy is then given by
\begin{equation}
    E(x,t)
    =
    \phi_E(\bm{\Phi}, \bm{\Pi})(x)
    \eqcolon
    \frac{1}{2}
    \phi_m(\bm{\Phi}, \bm{\Pi})(x) \cdot \mathcal{L}_\alpha(\phi_\rho(\bm{\Phi}))^{-1} \phi_m(\bm{\Phi}, \bm{\Pi})(x)
    +
    \phi_e(\bm{\Phi})(x) \,.
\end{equation}
Hence, we have coordinate transformations relating the Lagrangian and Eulerian coordinates: with entropy as the thermodynamic variable
\begin{equation}
    (\mathbf{m},\rho,\sigma)
    =
    \phi^{(s)}(\bm{\Phi}, \bm{\Pi})
    =
    (\phi_m(\bm{\Phi}, \bm{\Pi}), \phi_\rho(\bm{\Phi}), \phi_\sigma(\bm{\Phi})) \,,
\end{equation}
with internal energy as the thermodynamic variable
\begin{equation}
    (\mathbf{m},\rho,e)
    =
    \phi^{(e)}(\bm{\Phi}, \bm{\Pi})
    =
    (\phi_m(\bm{\Phi}, \bm{\Pi}), \phi_\rho(\bm{\Phi}), \phi_e(\bm{\Phi})) \,,
\end{equation}
and with total energy as the thermodynamic variable
\begin{equation}
    (\mathbf{m},\rho,E)
    =
    \phi^{(E)}(\bm{\Phi}, \bm{\Pi})
    =
    (\phi_m(\bm{\Phi}, \bm{\Pi}), \phi_\rho(\bm{\Phi}), \phi_E(\bm{\Phi},\bm{\Pi})) \,.
\end{equation}
These formulas help us break down complex coordinate transformations into a sequence of simpler transformations.

With the coordinates $(m,\rho, \sigma)$, the Hamiltonian is given by
\begin{equation}
    H[\mathbf{m}, \rho, \sigma]
    =
    \int_\Omega 
    \left(
    \frac{1}{2}
    \mathbf{m} \cdot \mathbf{u}[\mathbf{m},\rho]
    +
    e(\rho, \sigma/\rho, \nabla_{\mathbf{x}} \rho)
    \right)
    \mathsf{d}^d \mathbf{x} \,,
\end{equation}
where $\mathbf{u}$ is obtained by solving $\mathbf{m} = \mathcal{L}_\alpha^{-1}(\rho) \mathbf{u} = \rho \mathbf{u} - \alpha \nabla_{\mathbf{x}} (\rho (\nabla_{\mathbf{x}} \cdot \mathbf{u}))$. Similarly, in the coordinates $(m,\rho, e)$, the Hamiltonian is given by
\begin{equation}
    H[\mathbf{m},\rho, e]
    =
    \int_\Omega 
    \left(
    \frac{1}{2}
    \mathbf{m} \cdot \mathbf{u}[\mathbf{m},\rho]
    +
    e
    \right)
    \mathsf{d}^d \mathbf{x} \,,
\end{equation}
where now $e$ is thought of as an independent variable. Finally, in the coordinates $(m, \rho, E)$, the Hamiltonian is given by
\begin{equation}
    H[\mathbf{m}, \rho, E]
    =
    \int_\Omega 
    E \,
    \mathsf{d}^d \mathbf{x} \,.
\end{equation}
The Hamiltonian becomes exceedingly simple in this final coordinate system, but the Poisson bracket becomes correspondingly more complicated. In the following sections, we derive the Poisson bracket structure for the entropy, $(m,\rho,\sigma)$, and internal energy, $(m,\rho,e)$, coordinate systems. 

%----------------------------
\subsection{The Poisson bracket in entropy coordinates}

A convenient way to think of the coordinate transformations of the mass, momentum, and entropy densities is the following:
\begin{equation} \label{eq:coord_change_shorthand}
\begin{aligned}
    \rho(x,t)
    &=
    \int_\Omega \rho_0(\mathbf{X}) \delta(x - \bm{\Phi}(\mathbf{X},t)) \mathsf{d}^d \mathbf{X} \,, \\
    \mathbf{m}(x,t)
    &=
    \int_\Omega \bm{\Pi}(\mathbf{X},t) \delta(x - \bm{\Phi}(\mathbf{X},t)) \mathsf{d}^d \mathbf{X} \,, \\
    \sigma(x,t)
    &=
    \int_\Omega \rho_0(\mathbf{X}) s_0(\mathbf{X}) \delta(x - \bm{\Phi}(\mathbf{X},t)) \mathsf{d}^d \mathbf{X} \,.
\end{aligned}
\end{equation}
This formal way of representing the coordinate transformations is a convenient representation for the purposes of computing derivatives of the transformations in Equation~\eqref{eq:coord_transforms}, see \cite{suzuki2020generic,morrison1998hamiltonian}. Taking derivatives of the expressions in equation \eqref{eq:coord_change_shorthand}, we find that the partial derivatives with respect to $\bm{\Phi}$ are given by
\begin{equation}
\begin{aligned}
    \frac{\delta \phi_\rho}{\delta \bm{\Phi}}
    &=
    - \rho_0(\mathbf{X}) \nabla_{\mathbf{x}} \delta(x - \bm{\Phi}(\mathbf{X}, t)) \,, \\
    \frac{\delta \phi_\bom}{\delta \bm{\Phi}}
    &= - \bm{\Pi}(\mathbf{X},t) \otimes \nabla_{\mathbf{x}} \delta(x - \bm{\Phi}(\mathbf{X},t)) \,, \\
    \frac{\delta \phi_\sigma}{\delta \bm{\Phi}}
    &=
    - \rho_0(\mathbf{X}) s_0(\mathbf{X}) \nabla_{\mathbf{x}} \delta(x - \bm{\Phi}(\mathbf{X}, t)) \,, 
\end{aligned}
\end{equation}
and the partial derivatives with respect to $\bm{\Pi}$ by
\begin{equation}
    \frac{\delta \phi_\bom}{\delta \bm{\Pi}}
    =
    \mathbb{I} \delta(x - \bm{\Phi}(\mathbf{X},t)) \,,
\end{equation}
whereas $\delta \phi_\rho/\delta \bm{\Pi} = 0 = \delta \phi_\sigma /\delta \bm{\Pi}$. 

Suppose $F$ and $G$ are functionals of the Eulerian coordinates. We obtain the Poisson bracket in the Eulerian coordinates by pulling these functionals back to the Lagrangian reference frame:
\begin{equation}
    (\phi^{(s)})^* F(\bm{\Phi},\bm{\Pi})
    =
    F(\phi_\rho(\bm{\Phi}), \phi_\bom(\bm{\Phi},\bm{\Pi}), \phi_\sigma(\bm{\Phi})) \,,
\end{equation}
and similarly for $G$. For shorthand, let $\tilde{F} = (\phi^{(s)})^* F$. We find that
\begin{equation}
    \frac{\delta \tilde{F}}{\delta \bm{\Phi}}
    =
    \rho_0(\mathbf{X}) \nabla_{\mathbf{x}} \frac{\delta F}{\delta \rho}(\bm{\Phi}(\mathbf{X},t))
    +
    \left[ \nabla_{\mathbf{x}} \frac{\delta F}{\delta \mathbf{m}}(\bm{\Phi}(\mathbf{X},t)) \right]^T \bm{\Pi}(\mathbf{X},t) 
    +
    \rho_0(\mathbf{X}) s_0(\mathbf{X}) \nabla_{\mathbf{x}} \frac{\delta F}{\delta \sigma}(\bm{\Phi}(\mathbf{X},t)) \,,
\end{equation}
and
\begin{equation} \label{eq:change_of_coords_momentum_1}
    \frac{\delta \tilde{F}}{\delta \bm{\Pi}}
    =
    \frac{\delta F}{\delta \mathbf{m}}(\bm{\Phi}(\mathbf{X},t)) \,.
\end{equation}
Plugging these expressions into the canonical Poisson bracket in the Lagrangian reference frame, and pushing forward to the Eulerian reference frame, we find the Lie--Poisson bracket:
\begin{equation}
\begin{aligned}
    \{F, G\}
    =
    &-\int_\Omega \mathbf{m} \left[ \left(\frac{\delta F}{\delta \mathbf{m}} \cdot \nabla_{\mathbf{x}} \right) \frac{\delta G}{\delta \mathbf{m}} - \left(\frac{\delta G}{\delta \mathbf{m}} \cdot \nabla_{\mathbf{x}} \right) \frac{\delta F}{\delta \mathbf{m}} \right] \mathsf{d}^d \mathbf{x} \\
    &-\int_\Omega \rho \left[ \frac{\delta F}{\delta \mathbf{m}} \cdot \nabla_{\mathbf{x}} \frac{\delta G}{\delta \rho} - \frac{\delta G}{\delta \mathbf{m}} \cdot \nabla_{\mathbf{x}} \frac{\delta F}{\delta \rho} \right] \mathsf{d}^d \mathbf{x} \\
    &-\int_\Omega \sigma \left[ \frac{\delta F}{\delta \mathbf{m}} \cdot \nabla_{\mathbf{x}} \frac{\delta G}{\delta \sigma} - \frac{\delta G}{\delta \mathbf{m}} \cdot \nabla_{\mathbf{x}} \frac{\delta F}{\delta \sigma} \right] \mathsf{d}^d \mathbf{x} \,.
\end{aligned}
\end{equation}
We make the following definitions to simplify notation:
\begin{equation}
    p
    =
    \rho^2 \varepsilon_\rho \,,
    \quad
    \vartheta
    =
    \varepsilon_s \,, 
    \quad
    g
    =
    \varepsilon + \frac{p}{\rho} - s \vartheta \,,
\end{equation}
which are the pressure, temperature, and specific Gibbs free energy, respectively. We further define, following \cite{suzuki2020generic}, the vector
\begin{equation}
    \bm{\xi}
    =
    \rho \varepsilon_{\nabla_{\mathbf{x}} \rho} \,,
\end{equation}
which aids in writing the Korteweg-type stress compactly. The evolution equations may be obtained by noting that
\begin{equation}
    \frac{\delta H}{\delta \mathbf{m}}
    =
    \mathbf{u} \,,
    \quad
    \frac{\delta H}{\delta \rho}
    =
    - \frac{1}{2} \left( |\mathbf{u}|^2 + \alpha (\nabla_{\mathbf{x}} \cdot \mathbf{u})^2 \right)
    + g - \nabla_{\mathbf{x}} \cdot \bm{\xi} \,,
    \quad \text{and} \quad
    \frac{\delta H}{\delta \sigma} = \vartheta \,.
\end{equation}
If we let 
\begin{equation}
    F[\mathbf{m},\rho,\sigma]
    =
    \int_\Omega \left(
    \mathbf{m} \cdot \psi_m
    +
    \rho \cdot \psi_\rho
    +
    \sigma \cdot \psi_\sigma
    \right)
    \mathsf{d}^d \mathbf{x} \,,
\end{equation}
for arbitrary test functions, $\psi_m, \psi_\rho, \psi_\sigma$, we obtain the following weak evolution equations
\begin{equation}
\begin{aligned}
    \{F, H\}
    =
    &-\int_\Omega \mathbf{m} \left[ \left( \psi_m \cdot \nabla_{\mathbf{x}} \right) \mathbf{u} - \left( \mathbf{u} \cdot \nabla_{\mathbf{x}} \right) \psi_m \right] \mathsf{d}^d \mathbf{x} \\
    &-\int_\Omega \rho \left[ \psi_m \cdot \nabla_{\mathbf{x}} \left( - \frac{1}{2} \left( |\mathbf{u}|^2 + \alpha (\nabla_{\mathbf{x}} \cdot \mathbf{u} )^2 \right)
    + g - \nabla_{\mathbf{x}} \cdot \bm{\xi} \right) - \mathbf{u} \cdot \nabla_{\mathbf{x}} \psi_\rho \right] \mathsf{d}^d \mathbf{x} \\
    &-\int_\Omega \sigma \left[ \psi_m \cdot \nabla_{\mathbf{x}} \vartheta - \mathbf{u} \cdot \nabla_{\mathbf{x}} \psi_\sigma \right] \mathsf{d}^d \mathbf{x} \,.
\end{aligned}
\end{equation}
Notice that
\begin{multline}
    - (\nabla_{\mathbf{x}} \mathbf{u})^T \mathbf{m} + \frac{1}{2} \rho \nabla_{\mathbf{x}} \left( |\mathbf{u}|^2 + \alpha (\nabla_{\mathbf{x}} \cdot \mathbf{u})^2 \right) \\
    =
    - (\nabla_{\mathbf{x}} \mathbf{u})^T (\rho \mathbf{u} - \alpha \nabla_{\mathbf{x}} (\rho \nabla_{\mathbf{x}} \cdot \mathbf{u})) 
    + \frac{1}{2} \rho \nabla_{\mathbf{x}} \left( |\mathbf{u}|^2 + \alpha (\nabla_{\mathbf{x}} \cdot \mathbf{u})^2 \right) \\
    =
    \alpha
    \left(
    (\nabla_{\mathbf{x}} \mathbf{u})^T ( \rho \nabla_{\mathbf{x}} \mathbf{u}) + \frac{\rho}{2} \nabla_{\mathbf{x}}( (\nabla_\mathbf{x} \cdot \mathbf{u} )^2)
    \right)
    =
    \nabla_{\mathbf{x}} \cdot \left( \alpha \rho ( \nabla_{\mathbf{x}} \cdot \mathbf{u}) (\nabla_{\mathbf{x}} \mathbf{u})^T \right) \,.
\end{multline}
We further use the fact that
\begin{equation}
    \rho \nabla_{\mathbf{x}} (\nabla_{\mathbf{x}} \cdot \bm{\xi}) - (\nabla_{\mathbf{x}}^2 \rho) \bm{\xi}
    =
    \nabla_{\mathbf{x}} (\rho \nabla_{\mathbf{x}} \cdot \bm{\xi}) - \nabla_{\mathbf{x}} \cdot( \nabla_{\mathbf{x}} \rho \otimes \bm{\xi}) \,,
\end{equation}
and the thermodynamic relations
\begin{equation} \label{eq:thermo_relations}
\begin{aligned}
    \rho(g + s \vartheta) 
    &= \rho \left( \varepsilon + \frac{p}{\rho} \right)
    = e + p \,, \\
    \rho(\nabla_{\mathbf{x}} g + s \nabla_{\mathbf{x}} \vartheta)
    &= \rho \left( \nabla_{\mathbf{x}} \varepsilon + \nabla_{\mathbf{x}} \left( \frac{p}{\rho} \right) - \vartheta \nabla_{\mathbf{x}} s \right)
    =
    \nabla_{\mathbf{x}} p + (\nabla_{\mathbf{x}}^2 \rho) \bm{\xi} \,,
\end{aligned}
\end{equation}
and find that
\begin{equation}
\begin{aligned}
    \partial_t \rho &= - \nabla_{\mathbf{x}} \cdot(\rho \mathbf{u}) \,, \\
    \partial_t \mathbf{m} &= - \nabla_{\mathbf{x}} \cdot \left( \mathbf{m} \otimes \mathbf{u} + \alpha \rho ( \nabla_{\mathbf{x}} \cdot \mathbf{u}) (\nabla_{\mathbf{x}} \mathbf{u})^T + (p - \rho (\nabla_{\mathbf{x}} \cdot \bm{\xi})) \mathbb{I} + \nabla_{\mathbf{x}} \rho \otimes \bm{\xi} \right) \\
    \partial_t \sigma &= - \nabla_{\mathbf{x}} \cdot \left( \sigma \mathbf{u} \right) \,.
\end{aligned}
\end{equation}
Hence, we have found the regularized Euler model in conservative form in the coordinates $(\mathbf{m},\rho,\sigma)$. Recall, we recover the model with non-dispersive acoustic waves by letting the internal energy take the particular form
\begin{equation}
    \rho \varepsilon(\rho, s, \nabla_{\mathbf{x}} \rho)
    =
    \rho \varepsilon_0(\rho, s)
    +
    \frac{\alpha (\rho \varepsilon_0)_{\rho\rho}}{2} | \nabla_{\mathbf{x}} \rho|^2 \,,
    \implies
    \bm{\xi}
    =
    \rho \varepsilon_{\nabla_{\mathbf{x}} \rho}
    =
    \alpha \rho (\rho \varepsilon_0)_{\rho\rho} \nabla_{\mathbf{x}} \rho \,.
\end{equation}

%----------------------------
\subsection{The Poisson bracket in internal energy coordinates}

Slightly more effort is required to obtain a Poisson bracket in the coordinates $(\mathbf{m},\rho, e)$ due to a more complex transformation rule between the Lagrangian and Eulerian reference frames. 

Note that
\begin{equation}
\begin{aligned}
    \delta &\phi_e(\bm{\Phi}; \delta \bm{\Phi}) \\
    &=
    \left. \frac{\mathsf{d}}{\mathsf{d} \epsilon} \right|_{\epsilon = 0}
    \phi_\rho(\bm{\Phi} + \epsilon \delta \bm{\Phi})(x)
    \varepsilon
    \left(
    \phi_\rho(\bm{\Phi} + \epsilon \delta \bm{\Phi})(x) \,,
    \frac{\phi_\sigma(\bm{\Phi} + \epsilon \delta \bm{\Phi})(x)}{\phi_\rho(\bm{\Phi} + \epsilon \delta \bm{\Phi})(x)} \,,
    \nabla_{\mathbf{x}}[\phi_\rho(\bm{\Phi} + \epsilon \delta \bm{\Phi})](x)
    \right) \\
    &=
    \varepsilon \frac{\delta \phi_\rho}{\delta \bm{\Phi}} \cdot \delta \bm{\Phi}
    +
    \phi_\rho
    \left[
    \varepsilon_\rho \frac{\delta \phi_\rho}{\delta \bm{\Phi}} \cdot \delta \bm{\Phi}
    +
    \varepsilon_s \left( - \frac{\phi_\sigma}{\phi_\rho^2} \frac{\delta \phi_\rho}{\delta \bm{\Phi}} + \frac{1}{\phi_\rho} \frac{\delta \phi_\sigma}{\delta \bm{\Phi}} \right) \cdot \delta \bm{\Phi}
    +
    \varepsilon_{\nabla_{\mathbf{x}} \rho} \cdot \nabla_{\mathbf{x}} \left( \frac{\delta \phi_\rho}{\delta \bm{\Phi}} \cdot \delta \bm{\Phi} \right)
    \right] \,,
\end{aligned}
\end{equation}
where subscripts of the specific internal energy, $\varepsilon$, indicate partial derivatives with respect to its arguments. Hence, we find that
\begin{equation}
\begin{aligned}
    \frac{\delta \phi_e}{\delta \bm{\Phi}}
    &=
    \left( \varepsilon + \rho \varepsilon_\rho - \frac{\sigma}{\rho} \varepsilon_s \right) \frac{\delta \phi_\rho}{\delta \bm{\Phi}}
    +
    \varepsilon_s \frac{\delta \phi_\sigma}{\delta \bm{\Phi}}
    +
    (\rho \varepsilon_{\nabla_{\mathbf{x}} \rho} \cdot \nabla_{\mathbf{x}}) \frac{\delta \phi_\rho}{\delta \bm{\Phi}} \\
    &=
    \left(g + \bm{\xi} \cdot \nabla_{\mathbf{x}} \right) \frac{\delta \phi_\rho}{\delta \bm{\Phi}}
    +
    \vartheta \frac{\delta \phi_\sigma}{\delta \bm{\Phi}} \\
    &=
    -\rho_0(\mathbf{X}) \left(g(x,t) + s_0(\mathbf{X}) \vartheta(x) + \bm{\xi}(x,t) \cdot \nabla_{\mathbf{x}} \right) \nabla_{\mathbf{x}} \delta(x - \bm{\Phi}(\mathbf{X},t)) \,.
\end{aligned}
\end{equation}
Note the mixture of Lagrangian and Eulerian quantities, which must be handled with care. 

If we define $\tilde{F}[\bm{\Phi},\bm{\Pi}] = F[\mathbf{m},\rho,e]$, we find that
\begin{equation}
\begin{aligned}
    \frac{\delta \tilde{F}}{\delta \bm{\Phi}}
    &=
    \rho_0(\mathbf{X}) \nabla_{\mathbf{x}} \frac{\delta F}{\delta \rho} (\bm{\Phi}(\mathbf{X},t))
    +
    \left[ \nabla_{\mathbf{x}} \frac{\delta F}{\delta \mathbf{m}} (\bm{\Phi}(\mathbf{X},t) \right]^T \bm{\Pi}(\mathbf{X},t) \\
    &\qquad 
    + \rho_0(\mathbf{X}) \nabla_{\mathbf{x}} \left( g \frac{\delta F}{\delta e} \right)(\bm{\Phi}(\mathbf{X},t))
    +
    \rho_0(\mathbf{X}) s_0(\mathbf{X}) \nabla_{\mathbf{x}} \left( \vartheta \frac{\delta F}{\delta e} \right) (\bm{\Phi}(\mathbf{X},t)) \\
    &\qquad
    - \rho_0(\mathbf{X}) \nabla_{\mathbf{x}} \left[ \nabla_{\mathbf{x}} \cdot \left( \frac{\delta F}{\delta e} \bm{\xi} \right) \right] (\bm{\Phi}(\mathbf{X},t) )\,.
\end{aligned}
\end{equation}
The derivative with respect to the canonical momentum remains as given in equation \eqref{eq:change_of_coords_momentum_1}. Using the thermodynamic relations in equation \eqref{eq:thermo_relations}, it is possible to show that the Poisson bracket becomes
\begin{equation} \label{eq:internal_nrg_pb}
\begin{aligned}
    \{F, G\}
    =
    &-\int_\Omega \mathbf{m} \cdot \left[ \left(\frac{\delta F}{\delta \mathbf{m}} \cdot \nabla_{\mathbf{x}} \right) \frac{\delta G}{\delta \mathbf{m}} - \left(\frac{\delta G}{\delta \mathbf{m}} \cdot \nabla_{\mathbf{x}} \right) \frac{\delta F}{\delta \mathbf{m}} \right] \mathsf{d}^d \mathbf{x} \\
    &-\int_\Omega \rho \left[ \frac{\delta F}{\delta \mathbf{m}} \cdot \nabla_{\mathbf{x}} \frac{\delta G}{\delta \rho} - \frac{\delta G}{\delta \mathbf{m}} \cdot \nabla_{\mathbf{x}} \frac{\delta F}{\delta \rho} \right] \mathsf{d}^d \mathbf{x} \\
    &-\int_\Omega e \left[ \frac{\delta F}{\delta \mathbf{m}} \cdot \nabla_{\mathbf{x}} \frac{\delta G}{\delta e} - \frac{\delta G}{\delta \mathbf{m}} \cdot \nabla_{\mathbf{x}} \frac{\delta F}{\delta e} \right] \mathsf{d}^d \mathbf{x} \\
    &-\int_\Omega \left[ \frac{\delta F}{\delta \mathbf{m}} \cdot \nabla_{\mathbf{x}} \left( p \frac{\delta G}{\delta e}\right) - \frac{\delta G}{\delta \mathbf{m}} \cdot \nabla_{\mathbf{x}} \left( p \frac{\delta F}{\delta e}\right) \right] \mathsf{d}^d \mathbf{x} \\
    &+\int_\Omega (\nabla_{\mathbf{x}}^2 \rho) \bm{\xi} \cdot \left[ \frac{\delta F}{\delta e} \frac{\delta G}{\delta \mathbf{m}} - \frac{\delta G}{\delta e} \frac{\delta F}{\delta \mathbf{m}} \right] \mathsf{d}^d \mathbf{x} \\
    &+\int_\Omega \rho \left[
    \left( \frac{\delta F}{\delta \mathbf{m}} \cdot \nabla_{\mathbf{x}} \right) \nabla_{\mathbf{x}} \cdot \left( \frac{\delta G}{\delta e} \bm{\xi} \right)
    -
    \left( \frac{\delta G}{\delta \mathbf{m}} \cdot \nabla_{\mathbf{x}} \right) \nabla_{\mathbf{x}} \cdot \left( \frac{\delta F}{\delta e} \bm{\xi} \right)
    \right] \mathsf{d}^d \mathbf{x} \,.
\end{aligned}
\end{equation}
The evolution equations may be obtained by noting that
\begin{equation}
    \frac{\delta H}{\delta \mathbf{m}} = \mathbf{u} \,,
    \quad
    \frac{\delta H}{\delta \rho}
    =
    - \frac{1}{2} \left( |\mathbf{u}|^2 + \alpha (\nabla_{\mathbf{x}} \cdot \mathbf{u})^2 \right) \,,
    \quad \text{and} \quad
    \frac{\delta H}{\delta e} = 1 \,.
\end{equation}
If we let 
\begin{equation}
    F[\mathbf{m},\rho,\sigma]
    =
    \int_\Omega \left(
    \mathbf{m} \cdot \psi_m
    +
    \rho \cdot \psi_\rho
    +
    e \cdot \psi_e
    \right)
    \mathsf{d}^d \mathbf{x} \,,
\end{equation}
for arbitrary test functions, $\psi_m, \psi_\rho, \psi_e$, we obtain the following weak evolution equations
\begin{equation}
\begin{aligned}
    \{F, H\}
    =
    &-\int_\Omega \mathbf{m} \cdot \left[ \left(\psi_m \cdot \nabla_{\mathbf{x}} \right) \mathbf{u} - \left(\mathbf{u} \cdot \nabla_{\mathbf{x}} \right) \psi_m \right] \mathsf{d}^d \mathbf{x} \\
    &-\int_\Omega \rho \left[ \psi_m \cdot \nabla_{\mathbf{x}} \left( - \frac{1}{2} \left( |\mathbf{u}|^2 + \alpha (\nabla_{\mathbf{x}} \cdot \mathbf{u} )^2 \right) \right) - \mathbf{u} \cdot \nabla_{\mathbf{x}} \psi_\rho \right] \mathsf{d}^d \mathbf{x} \\
    &-\int_\Omega e \left[ \psi_m \cdot \nabla_{\mathbf{x}} (1) - \mathbf{u} \cdot \nabla_{\mathbf{x}} \psi_e \right] \mathsf{d}^d \mathbf{x} \\
    &-\int_\Omega \left[ \psi_m \cdot \nabla_{\mathbf{x}} \left( p \right) - \mathbf{u} \cdot \nabla_{\mathbf{x}} \left( p \psi_e \right) \right] \mathsf{d}^d \mathbf{x} \\
    &+\int_\Omega (\nabla_{\mathbf{x}}^2 \rho) \bm{\xi} \cdot \left[ \psi_e \mathbf{u} - \psi_m \right] \mathsf{d}^d \mathbf{x} \\
    &+\int_\Omega \rho \left[
    \left( \psi_m \cdot \nabla_{\mathbf{x}} \right) \nabla_{\mathbf{x}} \cdot \left( \bm{\xi} \right)
    -
    \left( \mathbf{u} \cdot \nabla_{\mathbf{x}} \right) \nabla_{\mathbf{x}} \cdot \left( \psi_e \bm{\xi} \right)
    \right] \mathsf{d}^d \mathbf{x} \,.
\end{aligned}
\end{equation}
The identity
\begin{multline}
    \mathbf{u} \cdot (D^2_{\mathbf{x}} \rho) \bm{\xi} - \bm{\xi} \cdot \nabla_{\mathbf{x}} \left( \mathbf{u} \cdot \nabla_{\mathbf{x}} \rho - \rho (\nabla_{\mathbf{x}} \cdot \mathbf{u}) \right) \\
    =
    - \left( \nabla_{\mathbf{x}} \mathbf{u} \right) \cdot (\nabla_{\mathbf{x}} \rho \otimes \bm{\xi}) + \left( \rho (\nabla_{\mathbf{x}} \cdot \mathbf{u}) \right) (\nabla_{\mathbf{x}} \cdot\bm{\xi}) - \nabla_{\mathbf{x}} \cdot \left[ (\rho (\nabla_{\mathbf{x}} \cdot\mathbf{u})) \bm{\xi} \right] 
\end{multline}
allows us to recover the strong-form evolution equations
\begin{equation}
\begin{aligned}
    \partial_t \rho &= - \nabla_{\mathbf{x}} \cdot(\rho \mathbf{u}) \,, \\
    \partial_t \mathbf{m} &= - \nabla_{\mathbf{x}} \cdot \left( \mathbf{m} \otimes \mathbf{u} + \alpha \rho ( \nabla_{\mathbf{x}} \cdot \mathbf{u}) (\nabla_{\mathbf{x}} \mathbf{u})^T + (p - \rho (\nabla_{\mathbf{x}} \cdot \bm{\xi})) \mathbb{I} + \nabla_{\mathbf{x}} \rho \otimes \bm{\xi} \right) \\
    \partial_t e &= - \nabla_{\mathbf{x}} \cdot( e \mathbf{u}) + (\nabla_{\mathbf{x}} \mathbf{u}) \cdot [ (- p + \rho (\nabla_{\mathbf{x}} \cdot \bm{\xi})) \mathbb{I} - \nabla_{\mathbf{x}} \rho \otimes \bm{\xi} ] - \nabla_{\mathbf{x}} \cdot( \rho (\nabla_{\mathbf{x}} \cdot \mathbf{u}) \bm{\xi}) \,.
\end{aligned}
\end{equation}
The evolution equation in the coordinates $(\mathbf{m}, \rho, e)$ is not locally conservative, as the internal energy is not expected to be locally conservative. The non-conservative term couples with the local kinetic energy transport, to yield a locally conservative total energy transport equation. 

%----------------------------
\subsection{Coordinate changes between the canonical and kinematic momenta} \label{appendix:momentum_coord_change}

Two different momenta are used throughout this work: the canonically conjugate momentum, 
\begin{equation}
    \mathbf{m} = \mathcal{L}_\alpha(\rho) \mathbf{u} = \rho \mathbf{u} - \alpha \nabla_{\mathbf{x}}( \rho \nabla_{\mathbf{x}} \cdot \mathbf{u}) \,,
\end{equation}
and the kinematic momentum, $\rho \mathbf{u}$. Here, we describe the coordinate change from $\mathbf{m}$ to $\rho \mathbf{u}$ for bracket-based formalisms. We do so by following the sequence of coordinate changes $(\mathbf{m}, \rho) \mapsto (\mathbf{u}, \rho) \mapsto (\rho \mathbf{u}, \rho)$. 

Let $F = F[\mathbf{m}, \rho] = \tilde{F}[\mathbf{u}, \rho]$. To begin, notice that
\begin{equation}
    \begin{pmatrix}
        \mathbf{m} \\
        \rho
    \end{pmatrix}
    =
    \begin{pmatrix}
        \mathcal{L}_\alpha(\rho) \mathbf{u} \\
        \rho
    \end{pmatrix}
    \implies
    \begin{pmatrix}
        \delta \mathbf{m} \\
        \delta \rho
    \end{pmatrix}
    =
    \begin{pmatrix}
        \mathcal{L}_\alpha(\rho) & (D_\rho \mathcal{L}_\alpha(\rho) [\cdot]) \mathbf{u} \\
        0 & 1 
    \end{pmatrix}
    \begin{pmatrix}
        \delta \mathbf{u} \\
        \delta \rho
    \end{pmatrix} \,,
\end{equation}
where $(D_\rho \mathcal{L}_\alpha(\rho) \delta \rho) \mathbf{u} = \mathcal{L}_\alpha(\delta \rho) \mathbf{u}$. To obtain the change of variables formula for functional derivatives, we need the inverse of this transformation matrix for differentials:
\begin{equation}
    \begin{pmatrix}
        \delta \mathbf{u} \\
        \delta \rho
    \end{pmatrix}
    =
    \begin{pmatrix}
        \mathcal{L}_\alpha(\rho) & (D_\rho \mathcal{L}_\alpha(\rho) [\cdot]) \mathbf{u} \\
        0 & 1 
    \end{pmatrix}^{-1}
    \begin{pmatrix}
        \delta \mathbf{m} \\
        \delta \rho
    \end{pmatrix}
    =
    \begin{pmatrix}
        \mathcal{L}_\alpha^{-1}(\rho) & - \mathcal{L}_\alpha^{-1}(\rho) (D_\rho \mathcal{L}_\alpha(\rho) [\cdot]) \mathbf{u} \\
        0 & 1 
    \end{pmatrix}
    \begin{pmatrix}
        \delta \mathbf{m} \\
        \delta \rho
    \end{pmatrix} \,.
\end{equation}
Therefore, because $DF = D_{\mathbf{m}} F \delta \mathbf{m} + D_\rho F \delta \rho = D_{\mathbf{u}} \tilde{F} \delta \mathbf{u} + D_\rho \tilde{F} \delta \rho$, we find that
\begin{equation}
\begin{aligned}
    \frac{\delta F}{\delta \mathbf{m}}
    &=
    \mathcal{L}_\alpha^{-1}(\rho) \frac{\delta \tilde{F}}{\delta \mathbf{u}} \,, \\
    \frac{\delta F}{\delta \rho}
    &=
    \frac{\delta \tilde{F}}{\delta \rho}
    -
    \mathbf{u} \cdot \left( \mathcal{L}_\alpha^{-1}(\rho) \frac{\delta \tilde{F}}{\delta \mathbf{u}} \right)
    -
    \alpha (\nabla \cdot \mathbf{u}) \left( \nabla \cdot \left( \mathcal{L}_\alpha^{-1}(\rho) \frac{\delta \tilde{F}}{\delta \mathbf{u}} \right) \right) \,.
\end{aligned}
\end{equation}
Similarly, we find that if $F'[\rho \mathbf{u}, \rho] = \tilde{F}[\mathbf{u}, \rho]$, then
\begin{equation}
    \frac{\delta \tilde{F}}{\delta \mathbf{u}}
    =
    \rho \frac{\delta F'}{\delta (\rho \mathbf{u})} \,, 
    \quad \text{and} \quad
    \frac{\delta \tilde{F}}{\delta \rho}
    = 
    \frac{\delta F'}{\delta \rho}
    +
    \mathbf{u} \cdot \frac{\delta F'}{\delta (\rho \mathbf{u})} \,.
\end{equation}
Hence, we find
\begin{equation}
\begin{aligned}
    \frac{\delta F}{\delta \mathbf{m}}
    &=
    \mathcal{L}_\alpha^{-1}(\rho) \left( \rho \frac{\delta F'}{\delta (\rho \mathbf{u})} \right) \,, \\
    \frac{\delta F}{\delta \rho} 
    &=
    \frac{\delta F'}{\delta \rho}
    +
    \mathbf{u} \cdot \frac{\delta F'}{\delta (\rho \mathbf{u})}
    -
    \mathbf{u} \cdot \left( \mathcal{L}_\alpha^{-1}(\rho) \left( \rho \frac{\delta F'}{\delta (\rho \mathbf{u})} \right) \right)
    -
    \alpha (\nabla \cdot \mathbf{u}) \left( \nabla \cdot \left( \mathcal{L}_\alpha^{-1}(\rho) \left( \rho \frac{\delta F'}{\delta (\rho \mathbf{u})} \right) \right) \right) \,.
\end{aligned}
\end{equation}
While the transformation of the derivative with respect to $\mathbf{m}$ is relatively straightforward, the derivative with respect to $\rho$ becomes quite complex. The final thermodynamic variable, either entropy density or energy density, is not involved in this coordinate transformation, and therefore does not contribute to the transformation rules for derivatives with respect to density and momentum, nor are its derivatives changed. 

%================================================================
\section{Ablation study: deleting high-order derivatives from HIGR} \label{appendix:ablation_study}

For pragmatic reasons, we consider a stabilization of the HIGR model obtained by deleting terms which we empirically observe lead to numerical instability. As such, there is no satisfactory theory associated with the model discussed in this section. Rather, this study is meant to build intuition about the HIGR model itself. The reason for considering this ablation study is illustrated in ~\Cref{subsec:1d_numerics}: HIGR is subject to numerical instability and spurious oscillations. We attribute the numerical instability to two things:
\begin{itemize}
    \item the Hamiltonian regularization introduces higher order derivatives in the energy flux,
    \item and the dissipative entropic pressure in HIGR adds excess heat in the energy equation. 
\end{itemize}
Despite the lack of theoretical motivation for this model, we nonetheless use symmetric, dissipative brackets to effect the deletion of terms in the energy equation to elucidate how these choices impact entropy production. 

We seek to remove the additional energy fluxes associated with the Korteweg-like correction, added in ~\Cref{sec:dispersion-free-extension}, and the dissipative entropic pressure. That is, we seek a symmetric bracket which contributes an additional term to the internal energy equation:
\begin{equation}
    \partial_t e + \hdots
    =
    \nabla_\bx \cdot ( \rho (\nabla_\bx \cdot \bu) \bm{\xi}) 
    +
    \nabla_\bx \cdot (\Sigma_D \bu) \,.
\end{equation}
This is obtained with the somewhat trivial and ad hoc symmetric bracket
\begin{multline}
    (F,G)
    =
    \int_\Omega
    \left[
    \frac{\delta F}{\delta e}
    \nabla_\bx \cdot \left( \rho (\nabla_\bx \cdot \bu) \bm{\xi} \frac{\delta G}{\delta e} \right)
    +
    \frac{\delta G}{\delta e}
    \nabla_\bx \cdot \left( \rho (\nabla_\bx \cdot \bu) \bm{\xi} \frac{\delta F}{\delta e} \right)
    \right]
    \mathsf{d}^3 \bx \\
    +
    \int_\Omega
    \left[
    \frac{\delta F}{\delta e}
    \nabla_\bx \cdot \left( \Sigma_D \bu \frac{\delta G}{\delta e} \right)
    +
    \frac{\delta G}{\delta e}
   \nabla_\bx \cdot \left( \Sigma_D \bu \frac{\delta F}{\delta e} \right)
    \right]
    \mathsf{d}^3 \bx
    \,.
\end{multline}
We use the Hamiltonian as the generating function:
\begin{multline}
    \dot{F}
    =
    (F,H)
    =
    \int_\Omega
    \left[
    \frac{\delta F}{\delta e}
    \nabla_\bx \cdot \left( \rho (\nabla_\bx \cdot \bu) \bm{\xi} \right)
    +
    \nabla_\bx \cdot \left( \rho (\nabla_\bx \cdot \bu) \bm{\xi} \frac{\delta F}{\delta e} \right)
    \right]
    \mathsf{d}^3 \bx \\
    +
    \int_\Omega
    \left[
    \frac{\delta F}{\delta e}
    \nabla_\bx \cdot \left( \Sigma_D \bu \right)
    +
   \nabla_\bx \cdot \left( \Sigma_D \bu \frac{\delta F}{\delta e} \right)
    \right]
    \mathsf{d}^3 \bx
    \,.
\end{multline}
which recovers the contributions to the energy equation we wanted. Total energy is globally conserved by this symmetric bracket, although we add local energy fluxes. Moreover, entropy is produced like
\begin{multline}
    \dot{S}
    =
    (S,H)
    =
    \int_\Omega
    \left[
    \vartheta^{-1}
    \nabla_\bx \cdot \left( \rho (\nabla_\bx \cdot \bu) \bm{\xi} \right)
    +
    \nabla_\bx \cdot \left( \rho (\nabla_\bx \cdot \bu) \bm{\xi} \vartheta^{-1} \right)
    \right]
    \mathsf{d}^3 \bx \\
    +
    \int_\Omega
    \left[
    \vartheta^{-1}
    \nabla_\bx \cdot \left( \Sigma_D \bu \right)
    +
   \nabla_\bx \cdot \left( \Sigma_D \bu \vartheta^{-1} \right)
    \right]
    \mathsf{d}^3 \bx
    \,.
\end{multline}
It is not entirely obvious whether the sign of the entropy production rate can be predicted in general. The stabilizing terms introduced in this section lack motivation beyond our observation that they improve the stability of HIGR simulations. That removing the energy flux from the Korteweg-like Hamiltonian regularization should improve stability is intuitive: it is unlikely that the model possesses adequate regularity for the higher-order derivatives in this term to be justified. Why removing the coupling of the dissipative entropic pressure from the energy equation improves stability is less clear. Rather than being a new model of interest in its own right, we think of this version of the HIGR model as an illustrative tool highlighting defects of the HIGR model. 

The HIGR model with higher-order derivatives deleted as described above is given by
\begin{equation}
\left\{
\begin{aligned}
    &\rho_t + (\rho u)_x = 0 \,, \\
    &(\rho u)_t + (\rho u^2 + (\gamma - 1) \rho \varepsilon + \Sigma_C + \Sigma_D)_x = 0 \,, \\
    &E_t + ( (E + (\gamma - 1) \rho \varepsilon + \Sigma_C)u)_x = 0 \,, \\
    &\rho^{-1} \Sigma_C - \alpha (\rho^{-1} (\Sigma_C)_x)_x = 
    \alpha 
    \left[
    u_x^2 
    +
    (\gamma - 1)
    \left(
    \varepsilon_x
    -
    (\gamma - 1) \varepsilon \rho_x/\rho
    \right)_x
    + 
    \frac{\gamma (\gamma - 1)^2 \varepsilon}{2} (\rho_x/\rho)^2
    \right]
    \,, \\
    &\rho^{-1} \Sigma_D - \alpha (\rho^{-1} (\Sigma_D)_x)_x = 
    \alpha 
    u_x^2 \,,
\end{aligned}
\right.
\end{equation}
where, again, $\varepsilon = \rho^{-1}(E - (1/2) \rho ( u^2 + \alpha (\partial_x u)^2)$. The discretization approach is identical to that described in ~\Cref{subsec:1d_numerics} with the same smoothed, periodized Sod shock initial conditions, see ~\Cref{fig:initial_condition}. See ~\Cref{fig:stabilized_higr} for the solution profile associated with this model. We observe a solution profile qualitatively similar to that obtained in ~\Cref{subsec:1d_numerics} for HIGR, but with smoother features and no oscillations from numerical instability. This confirms our belief that the presence of higher-order derivatives in HIGR and HRE is a key defect that requires adequate dissipative regularization to stabilize. 

\begin{figure}
    \centering
    \pgfplotscreateplotcyclelist{mycolors}{
    {blue, line width=1pt},
    {red, line width=1pt},
    {green!50!black, line width=1pt},
    {orange, line width=1pt}
}

\begin{tikzpicture}
\begin{groupplot}[
    group style={
        group size=2 by 2,        % 2 columns x 2 rows
        vertical sep=1.7cm,       % similar spacing to original
        horizontal sep=1cm
    },
    xmin=0, xmax=1,              % exact x-range, removes padding
    width=0.50\textwidth,        % two columns should fit across the page
    height=0.31\textwidth,       % keep the taller aspect from original
    xlabel={$x$},
    ylabel={},                   % keep y-label empty as before
    grid=both,
    cycle list name=mycolors,
    legend style={font=\small, at={(0.95,0.95)}, anchor=north east},
]

% -------------------
% Top-left: t = 0.125 (HIGR)
% -------------------
\nextgroupplot[title={$t=0.125$}]
\addplot table [x=x, y=rho, col sep=comma] {figures/data/higr_igr_comparison/nc512_higr_2_t_0.125000.csv};
\addplot table [x=x, y=u,   col sep=comma] {figures/data/higr_igr_comparison/nc512_higr_2_t_0.125000.csv};
\addplot table [x=x, y=E,   col sep=comma] {figures/data/higr_igr_comparison/nc512_higr_2_t_0.125000.csv};
\addplot table [x=x, y=eps, col sep=comma] {figures/data/higr_igr_comparison/nc512_higr_2_t_0.125000.csv};
\legend{$\rho$, $u$, $E$, $\varepsilon$}

% -------------------
% Top-right: t = 0.25 (HIGR)
% -------------------
\nextgroupplot[title={$t=0.25$}]
\addplot table [x=x, y=rho, col sep=comma] {figures/data/higr_igr_comparison/nc512_higr_2_t_0.250000.csv};
\addplot table [x=x, y=u,   col sep=comma] {figures/data/higr_igr_comparison/nc512_higr_2_t_0.250000.csv};
\addplot table [x=x, y=E,   col sep=comma] {figures/data/higr_igr_comparison/nc512_higr_2_t_0.250000.csv};
\addplot table [x=x, y=eps, col sep=comma] {figures/data/higr_igr_comparison/nc512_higr_2_t_0.250000.csv};
\legend{$\rho$, $u$, $E$, $\varepsilon$}

% -------------------
% Bottom-left: t = 0.375 (HIGR)
% -------------------
\nextgroupplot[title={$t=0.375$}]
\addplot table [x=x, y=rho, col sep=comma] {figures/data/higr_igr_comparison/nc512_higr_2_t_0.375000.csv};
\addplot table [x=x, y=u,   col sep=comma] {figures/data/higr_igr_comparison/nc512_higr_2_t_0.375000.csv};
\addplot table [x=x, y=E,   col sep=comma] {figures/data/higr_igr_comparison/nc512_higr_2_t_0.375000.csv};
\addplot table [x=x, y=eps, col sep=comma] {figures/data/higr_igr_comparison/nc512_higr_2_t_0.375000.csv};
\legend{$\rho$, $u$, $E$, $\varepsilon$}

% -------------------
% Bottom-right: t = 0.5 (HIGR)
% -------------------
\nextgroupplot[title={$t=0.5$}]
\addplot table [x=x, y=rho, col sep=comma] {figures/data/higr_igr_comparison/nc512_higr_2_t_0.500000.csv};
\addplot table [x=x, y=u,   col sep=comma] {figures/data/higr_igr_comparison/nc512_higr_2_t_0.500000.csv};
\addplot table [x=x, y=E,   col sep=comma] {figures/data/higr_igr_comparison/nc512_higr_2_t_0.500000.csv};
\addplot table [x=x, y=eps, col sep=comma] {figures/data/higr_igr_comparison/nc512_higr_2_t_0.500000.csv};
\legend{$\rho$, $u$, $E$, $\varepsilon$}

\end{groupplot}
\end{tikzpicture}
    \caption{Solution profile of the HIGR system with the higher order derivative terms deleted for the smoothed, periodized Sod shock problem.}
    \label{fig:stabilized_higr}
\end{figure}
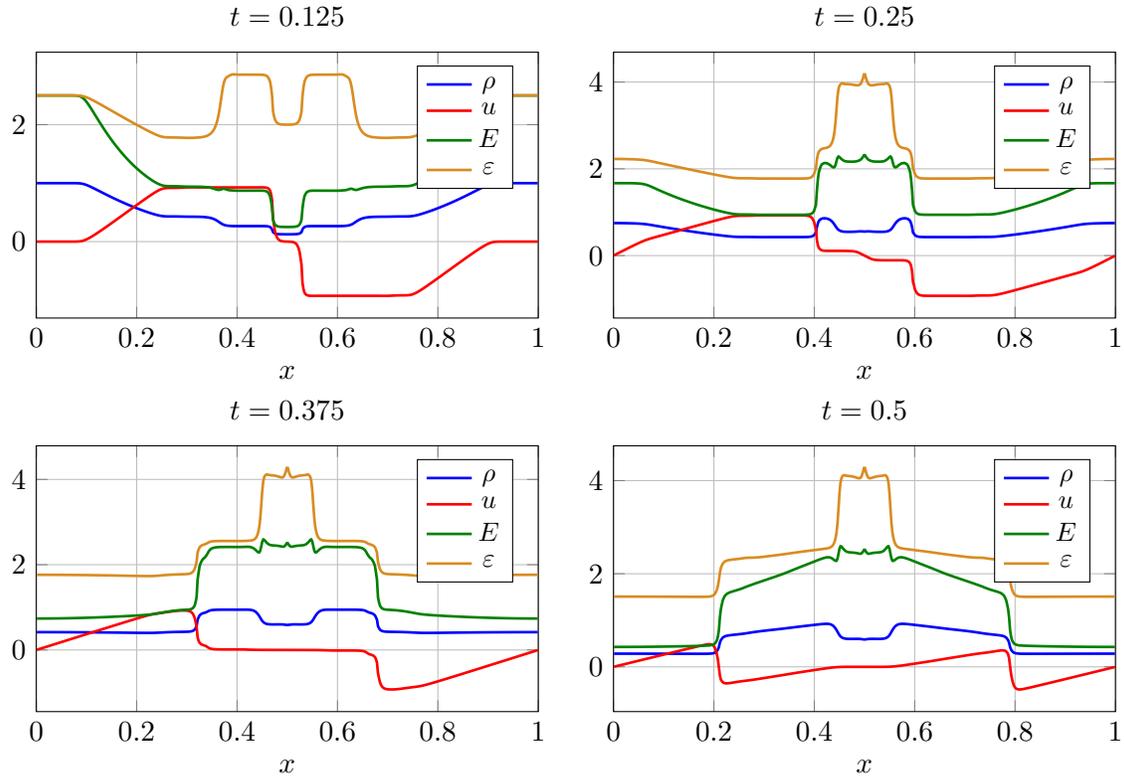

\end{document}